\documentclass{article}

\usepackage{booktabs}
\usepackage{todonotes}

\usepackage{bm}
\usepackage{color}
\usepackage{url}
\usepackage{graphicx}
\usepackage{pgfplots}
\usepackage{tikz}
\usepackage{algorithm}
\usepackage{algpseudocode}
\usepackage{mathrsfs}
\usepackage{relsize}
\usepackage{yfonts}
\usepackage{nicefrac}
\usepackage{multicol}
\usepackage{subcaption}
\usepackage{wrapfig}
\usepackage{sgame}

\usepackage{mathrsfs}
\usepackage{verbatim}
\usepackage{microtype}
\usepackage[normalem]{ulem}

\usepackage{amsmath,amssymb,amsfonts,amsthm}
\usepackage{mathtools}
\usepackage{thmtools}
\usepackage{thm-restate}
\usepackage{bbm}
\usepackage{dsfont}
\usepackage{hyperref}
\usepackage{cleveref}
\usepackage{balance}

\usepackage{upgreek}

\usetikzlibrary{arrows, automata}

\pgfplotsset{compat=1.15}

\renewcommand{\epsilon}{\varepsilon}

\renewcommand{\delta}{\updelta}

\newcommand{\vargaussian}{\sigma_{\mathsmaller{\mathcal{N}}}^{2}}
\newcommand{\vargamma}{\sigma_{\mathsmaller{\Gamma}}^{\smash{2}\vphantom{c}}}
\newcommand{\ssum}{\mathsmaller{\sum}\limits}

\newcommand{\probop}[1]{\mathop{\mathbb{#1}}\displaylimits}

\newcommand{\Expop}{\probop{E}}

\newcommand{\Expwrt}[2]{\Expop_{#1}\left[#2\right]}

\newcommand{\rand}[1]{\mathscr{#1}}

\newcommand{\distributed}{\sim}

\newcommand{\naive}{na\"{\i}ve\@}

\newcommand{\floor}[1]{\left\lfloor{} #1 \right\rfloor{}}
\newcommand{\ceil}[1]{\left\lceil{} #1 \right\rceil{}}
\newcommand{\abs}[1]{\left\lvert{} #1 \right\rvert{}}
\newcommand{\norm}[1]{\left\lVert{} #1 \right\rVert{}}

\DeclareMathOperator{\sgn}{sgn}

\newcommand{\F}{\mathcal{F}}

\newcommand{\R}{\mathbb{R}}
\newcommand{\N}{\mathbb{N}}

\newif\iflsymb
\lsymbfalse

\iflsymb

\else

\fi

\allowdisplaybreaks

\newcommand{\slfrac}[2]{#1/#2}

\newcommand{\sabs}[1]{\lvert{} #1 \rvert{}}

\DeclareMathOperator*{\Prob}{\mathbb{P}}

\DeclareMathOperator*{\Expect}{\mathbb{E}}

\DeclareMathOperator*{\Var}{\mathbb{V}}

\newcommand{\VarWimpy}{v}
\newcommand{\EVarWimpy}{\hat{\VarWimpy}}

\let\todo\undefined

\newcommand{\mydef}[1]{\emph{#1}}

\newcommand{\amy}[1]{{\color{red}[\textsc{Amy}: \emph{#1}]}}

\newcommand{\samy}[2]{{\color{pink}\sout{#1}\color{blue}#2}}
\newcommand{\enrique}[1]{{\color{blue}[\textsc{Enrique}: #1]}}
\newcommand{\cyrus}[1]{{\color{green!25!black}[\textsc{Cyrus}: \emph{#1}]}}
\newcommand{\bhaskar}[1]{{\color{purple}[\textsc{Bhaskar}: \emph{#1}]}}
\newcommand{\sbhaskar}[2]{{\color{pink}\sout{#1}\color{purple}#2}}

\newcommand{\todo}[1]{{\color{red!75!black}[\textsc{ToDo}: \emph{#1}]}}

\renewcommand{\amy}[1]{}
\renewcommand{\enrique}[1]{}
\renewcommand{\cyrus}[1]{}
\renewcommand{\bhaskar}[1]{}
\renewcommand{\todo}[1]{}

\renewcommand{\samy}[2]{#2}
\renewcommand{\sbhaskar}[2]{#2}

\newfloat{algorithm}{t}{lop}
\floatname{algorithm}{Algorithm}
\algnewcommand{\Input}{\textbf{input:}~}
\algnewcommand{\Output}{\textbf{output:}~}
\algnewcommand{\Error}{\textbf{error}~}
\algnewcommand{\Continue}{\textbf{continue}~}
\algnewcommand{\Break}{\textbf{break}~}

\newcommand{\Samples}{\bm{Y}}
\newcommand{\SamplePoint}{\ConditionValue}

\newcommand{\NumberOfSamples}{m}

\newcommand{\UtilityRange}{c}
\newcommand{\UtilityVariance}{\bm{v}}
\newcommand{\EUtilityVariance}{\hat{\bm{v}}}
\newcommand{\ConditionSpace}{\mathcal{Y}}
\newcommand{\ConditionValue}{y}

\newcommand{\ConditionDistribution}{\rand{D}}
\newcommand{\SampleIndex}{j}

\newcommand{\TimeIndex}{t}

\newcommand{\ScheduleLength}{T}

\newcommand{\GameTuple}{\Gamma}
\newcommand{\ConditionalGame}[1]{\GameTuple_{#1}}    
\newcommand{\InducedGame}[1]{\GameTuple_{#1}}        
\newcommand{\EmpiricalGame}[1]{\hat{\GameTuple}_{#1}} 

\newcommand{\NumberOfPlayers}{|\SetOfPlayers|}
\newcommand{\SetOfPlayers}{P}
\newcommand{\PlayerIndex}{p}
\newcommand{\StrategySet}{S}

\newcommand{\Strategy}{s}

\newcommand{\StratProfile}{\bm{s}}
\newcommand{\StratProfileAlt}{\bm{t}}
\newcommand{\StratProfileSpace}{\bm{\StrategySet}}

\newcommand{\MixedStrategySet}{{\StrategySet}^{\diamond}}

\newcommand{\MixedStratProfileAlt}{\bm{\tau}}
\newcommand{\MixedStratProfileSpace}{\bm{\MixedStrategySet}}

\newcommand{\Utility}{\bm{u}}

\newcommand{\SizeOfGame}[1]{N_{#1}}

\newcommand{\MV}{\mathrm{R}}

\newcommand{\AdversarialValue}{\mathrm{A}}
\newcommand{\MinimaxValue}{\mathrm{MM}}
\newcommand{\AnarchyRatio}{\mathrm{AR}}
\newcommand{\AnarchyGap}{\mathrm{AG}}
\newcommand{\StabilityRatio}{\mathrm{SR}}

\DeclareMathOperator{\Welfare}{W}

\DeclareMathOperator{\Nash}{E}

\DeclareMathOperator{\Adjacent}{Adj}
\DeclareMathOperator{\Regret}{Reg}

\newcommand{\GS}{\ensuremath{\operatorname{GS}}}
\newcommand{\PSP}{\ensuremath{\operatorname{PSP}}}
\newcommand{\UtilityIndices}{\bm{\mathcal{I}}}

\newcommand{\Indices}{\mathcal{I}}

\newcommand{\SetOfFacilities}{E}
\newcommand{\NumberOfFacilities}{|\SetOfFacilities|}
\newcommand{\FacilityIndex}{e}
\newcommand{\FacilityCostFunction}{f}
\newcommand{\PlayerCost}{C}
\newcommand{\CongestionGame}{\mathcal{C}}

\newcommand{\RandomZeroSum}{\mathrm{RZ}}
\newcommand{\RandomCongestionGame}{\mathrm{RC}}

\newcommand{\smallabs}[1]{\Bigl\lvert #1 \Bigr\rvert}
\newcommand{\smallsup}[1]{\mathrel{\raisebox{0.25ex}{\ensuremath{\displaystyle\sup_{\mathsmaller{#1}}}}}}
\newcommand{\smallinf}[1]{\mathrel{\raisebox{0.25ex}{\ensuremath{\displaystyle\inf_{\mathsmaller{#1}}}}}}

\newtheorem{observation}{Observation}[section]

\usepackage{natbib}

\usepackage{microtype}

\usepackage{sidecap}
\RequirePackage{caption, float}

\usepackage{algorithmicx}

\usepackage{placeins}

\usepackage{geometry}
\theoremstyle{definition}

\newtheorem{definition}{Definition}[section]

\newtheorem{lemma}{Lemma}[section]
\newtheorem{theorem}{Theorem}[section]
\newtheorem{corollary}{Corollary}[section]

\title
{Learning Properties of Simulation-Based Games}

\title{Computational and Data Requirements for Learning Generic Properties of Simulation-Based Games}

\author{Cyrus Cousins\quad Bhaskar Mishra\quad Enrique Areyan Viqueira\quad Amy Greenwald}

\date{August 2022}

\newif\ifpoa
\poafalse

\begin{document}


\maketitle
\begin{abstract}
Empirical game-theoretic analysis (EGTA) is primarily focused on learning the equilibria of simulation-based games.
Recent approaches have tackled this problem by learning a uniform approximation of the game's utilities, and then applying precision-recall theorems: i.e., all equilibria of the true game are approximate equilibria in the estimated game, and vice-versa. 
In this work, we generalize this approach to all game properties that are 
well behaved (i.e., Lipschitz continuous in utilities), 
including regret (which defines Nash and correlated equilibria), adversarial values, and power-mean and Gini social welfare.
Further, we introduce a novel algorithm---progressive sampling with pruning (\PSP)---for learning a uniform approximation and thus any well-behaved property of a game, which prunes strategy profiles once the corresponding players' utilities are well-estimated, and we analyze its data and query complexities in terms of the \emph{a priori\/} unknown utility variances.
We experiment with our algorithm extensively, showing that
1)~the number of queries that PSP saves is highly sensitive to the utility variance distribution, and 
2)~PSP consistently outperforms theoretical upper bounds, achieving significantly lower query complexities than natural baselines.
We conclude with experiments that uncover some of the remaining difficulties with learning properties of simulation-based games, in spite of recent advances in statistical EGTA methodology, including those developed herein.

\if
While optimal welfare turns out to be a well-behaved property,
optimal equilibria (i.e., welfare-maximizing or minimizing), and likewise the price of anarchy, are not well-behaved properties, which implies that they cannot be estimated with uniformly bounded complexity.
We thus define a related property based on a Lagrangian relaxation of the equilibrium constraints that is well behaved.
We call this property $\Lambda$-stability. 
As determining the value of an optimal equilibrium is an essential step in computing the price of anarchy, we conclude with a discussion of an alternative, more stable notion of anarchy based on $\Lambda$-stability, which we call the anarchy gap.
\fi
\end{abstract}

\section{Introduction}
\label{sec:intro}

Game theory is the standard conceptual framework used to analyze strategic interactions among rational agents in multiagent systems.
Traditionally, a game theorist assumes access to a complete specification of a game, including any stochasticity in the environment.
More recently, driven by high-impact applications, e.g., wireless spectrum~\citep{gandhi2007general,weiss2017sats} and online advertisement auctions~\citep{jordan2010designing,varian2007position}, researchers have turned their attention to analyzing games for which a complete specification is either too complex or too expensive to produce~\citep{wellman2006methods}.
When such games are observable through simulation, they are called \emph{simulation-based\/}
games~\cite{vorobeychik2008stochastic}.
The literature on \mydef{empirical game-theoretic analysis} (EGTA) aims to analyze these games.
This paper concerns \mydef{statistical} EGTA, in which the goal is to develop algorithms, ideally with finite-sample guarantees, that learn properties of games (e.g., their equilibria), even when a complete description of the game is not available.
\cyrus{Assumption: Care about expectations}
\cyrus{Fix name / year sorting issue}

In the statistical analysis of simulation-based games, it is generally assumed that the game analyst has access to an \emph{oracle}, i.e., a simulator, that can produce agent utilities, given an arbitrary strategy profile and a random condition.
This random condition can represent any number of exogenous random variables (e.g., the weather), and it can also incorporate an entropy source based on which any randomness within the game (e.g., the roll of a dice) is generated.

Randomness renders it necessary to query the simulator multiple times at each strategy profile to estimate the game's utilities.
This, in turn, necessitates that we sample many random conditions.
To measure the amount of ``work'' needed to produce estimates within some specified error with high probability, we define two complexity measures.
\mydef{Data complexity} is the requisite number of sampled random conditions,
while \mydef{query complexity} is the requisite number of simulator queries. 
%
Data complexity is relevant in Bayesian games, such as private-value auctions.
Doing market research (e.g., collecting data to learn about the bidders' private values) can be expensive, but each of these expensive samples can then be used to query all strategy profiles.
Query complexity is the more natural metric in computation-bound games, such as Starcraft~\cite{tavares2016rock}, where simulating the game to determine utilities 
can be expensive, regardless of access to a random seed.

Prior work in EGTA has focused primarily on identifying a specific property of games, usually one or all equilibria.
For example, \citet{tuyls2020bounds} show that if one can accurately estimate all the players' utilities in a game, then one can also accurately estimate all equilibria.
If one is willing to settle for a single equilibrium, more efficient algorithms can be employed~\citep{fearnley2015learning,wellman2006methods}, and even when all equilibria are of interest, algorithms that prune strategy profiles, even heuristically~\citep{areyan2019learning,areyan2020improved}, can potentially reduce the requisite amount of work.
In this paper, we aim to broaden the scope of statistical EGTA beyond its primary focus on equilibria to arbitrary properties of games.
In particular, we consider properties like welfare at various strategy profiles, as well as extreme properties (e.g., maximum or minimum welfare) and their witnesses (e.g., strategy profiles that realize maximum welfare).

We adopt the uniform approximation framework of \citet{tuyls2020bounds}, wherein one game is an $\epsilon$-approximation of another if players' utilities differ by no more than $\epsilon$ everywhere.
\Citet{tuyls2020bounds} show that an $\epsilon$-uniform approximation of a game implies a $2\epsilon$-uniform approximation of regret (i.e., equilibria).
We extend this idea to approximate many such properties, by showing that Lipschitz continuity is a sufficient condition for a property (and its extrema) to be uniformly approximable.
For example, both power-mean and Gini social welfare --- two classes of welfare measures that capture utilitarian and egalitarian welfare --- and their extrema, 
can be estimated using our framework.
We also rederive known results about approximating
equilibria~\citep{areyan2019learning,tuyls2020bounds} as special cases of our general theory, and we argue that our theory applies to equilibria of derived games, such as correlated equilibria, and to derived games themselves, such as team games.

Next, we develop a novel learning
algorithm that estimates simulation-based games. 
%
We move beyond earlier work by \citet{tuyls2020bounds},
who employed standard Hoeffding bounds for mean estimation,
and follow \citet{areyan2020improved},
who derive empirical variance-aware bounds: they bound each utility in terms of an empirical estimate of the largest variance across all utilities.
We rederive their bounds, obtaining tighter constants; however, instead of uniform bounds, we present non-uniform, variance-sensitive bounds: we bound each utility in terms of its own (empirical) variance.
\citeauthor{areyan2020improved} employed their bounds to prune strategy profiles
with provably high regret.
We develop an algorithm that uses our bounds to prune based on an alternative criterion: profiles are pruned as soon as they are provably well estimated.
Finally, we derive a sampling schedule that guarantees the correctness of our algorithm upon termination.

In addition to establishing its correctness, we also characterize the efficiency of our pruning algorithm.
We argue that our algorithm is nearly optimal with respect to both data and query complexities.
It can thus be used to learn myriad game properties, like welfare and regret, with complexity dependent on the size of the game, without requiring any \emph{a priori\/} information about the players' utilities beyond their ranges; even variances need not be known \emph{a priori\/}.

We also experiment extensively with our algorithm and a variety of natural baselines on a suite of games generated using GAMUT~\cite{nudelman2004run}, enhanced with various noise structures so that they mimic simulation-based games.
Our findings can be summarized as follows:
\begin{itemize}
\item We show that the number of queries our pruning algorithm requires to reach its target error \amy{double check: accuracy} guarantee is highly sensitive to the utility variance distribution.

\item We demonstrate that our pruning algorithm consistently outperforms theoretical upper bounds, and achieves significantly lower query complexities than all baselines.
\end{itemize}
We conclude with experiments that uncover some of the remaining difficulties with learning properties of simulation-based games, in spite of recent advances in statistical EGTA methodology, including those developed herein.

\ifpoa
In the final sections of the paper, we give several examples of properties---extreme equilibria (i.e., optimal and pessimal), the price of anarchy~\cite{koutsoupias1999worst}, and the price of stability~\cite{Schulz03,Anshelevich08}---which are not amenable to 
statistical EGTA methodology.
Indeed, \naive{} estimation of extreme equilibria via a uniform approximation
can be arbitrarily incorrect.
The prices of anarchy and stability are likewise inapproximable, because they are functions of extreme equilibria.

In light of these observations, we introduce a relaxation of extreme equilibria, which we call $\Lambda$-stable outcomes.
These outcomes are defined based on an extreme value of welfare discounted by the distance from an equilibrium, i.e., $\Lambda$ acts like a Lagrangian.
Unlike extreme equilibria, this new concept is amenable to statistical EGTA methodology.

Finally, we investigate the price of anarchy, i.e., the ratio between the best and worst welfare, the latter restricted to values at equilibria.
We argue that even if the worst equilibrium were replaced by a pessimal $\Lambda$-stable outcome, we still could not estimate the price of anarchy, because the sensitivity to the best welfare outcome approaches infinity as the worst equilibrium welfare approaches zero.
We thus define an alternative concept, the \emph{anarchy gap}, which instead measures the \emph{difference\/} between the best welfare and the maximally dissonant 
$\Lambda$-stable outcome (a relaxation of a worst welfare equilibrium).
This new concept is amenable to statistical EGTA methodology, as we show experimentally in Appendix~\ref{app:poa_appendix}.
An analogous story holds for the price of stability.

\fi

\paragraph{Related Work.}
Many important strategic situations of interest are too complex for standard game-theoretic methodologies to apply directly.
Empirical game-theoretical analysis (EGTA) has emerged as a methodological approach to extending game-theoretic analysis with computational tools \cite{wellman2006methods}.
EGTA typically assumes the existence of an oracle
that can be queried to (actively) learn about a game, or past (batch) data that comprise relevant information, or both.
Since the central assumption is that the game is unavailable in 
closed-form, games that are encoded in this way have been dubbed black-box games~\citep{picheny2016bayesian}.
At the same time, 
since each query often requires potentially expensive computer simulations---for which trace data may be available (i.e., white boxes)---they are better known as simulation-based games~\citep{vorobeychik2008stochastic}.

The EGTA literature, while relatively young, is growing rapidly, with researchers actively contributing methods for myriad game models.
Some of these methods are designed for normal-form games~\citep{areyan2020improved,areyan2019learning,tavares2016rock,fearnley2015learning,vorobeychik2008stochastic}, and others, for extensive-form games~\citep{learningmaximin,Gattieqapprox,Zhang_Sandholm_2021}.
Most methods apply to games with finite strategy spaces, but some apply to games with infinite strategy spaces \citep{learningmaximin,vorobeychik2007learning,wiedenbeck2018regression}.
A related line of work aims to empirically design mechanisms via EGTA methodologies~\citep{vorobeychik2006empirical,DBLP:conf/uai/ViqueiraCMG19}.
Our analysis of game properties applies to normal-form games with finite or mixed (hence, infinite) strategy spaces, but our learning methodology applies only in the finite case.

While the aforementioned methods were designed for a wide variety of game models, most share the same goal: to estimate a Nash equilibrium of the game. \Citet{areyan2019learning,areyan2020improved} are exceptions, as they prove a dual containment, thereby bounding the \emph{set\/} of equilibria of a game.
In this paper, we extend their methodology to estimate not only all equilibria, but any well-behaved game property of interest, such as power-mean and Gini social welfare.
Since regret (which defines equilibria) is one such well-behaved property, our methodology subsumes this earlier work.

\Citet{areyan2020improved} also develop two learning algorithms.
The first (Global Sampling) uniformly estimates the game using empirical variance-aware bounds.
The second, also a progressive sampling algorithm with pruning, prunes high regret strategy profiles to improve query efficiency, no longer learning empirical games that are $\epsilon$-uniform approximations, but still learning $2\epsilon$-approximate equilibria of the simulation-based game.
Although they demonstrate the efficacy of their pruning algorithm experimentally, they fail to provide an efficiency analysis.


\section{Approximation Framework}
\label{sec:approx}

We begin by presenting our approximation framework: defining games, their properties, and
well-behaved game properties.
Doing so requires two pieces of technical machinery: the notion of uniform approximation and that of Lipschitz continuity.
Given this machinery, it is immediate that if a game property is Lipschitz continuous in utilities, then it is \emph{well behaved}, meaning it can be learned by any algorithm that produces a uniform approximation of the game.

In this work, we emphasize \emph{extremal\/} game properties, such as the optimal or pessimal welfare (the latter being relevant, for example, when considering the price of anarchy~\citep{christodoulou2005price}).
We are interested in learning not only the \emph{values\/} of these extremal properties, but further the
strategy profiles
that generate those values: i.e., the arguments that realize the solutions to the optimization problems that correspond to extremal properties.
We refer to these strategy profiles as \mydef{witnesses} of the corresponding property.
For example, an equilibrium is a witness of the regret property: i.e., it is a strategy profile that minimizes regret (and thereby attains an extremum).

Given an approximation of one game by another, there is not necessarily a connection between their extremal properties.
For example, there may be equilibria in one game with no corresponding equilibria in the other, as small changes to the utilities can add or remove equilibria.
Nonetheless, it is known~\citep{areyan2019learning} that finding the equilibria of a uniform approximation of a game is sufficient for finding the approximate equilibria of the game itself.
The main result of this section is to generalize this result beyond regret (which defines equilibria) to all well-behaved game properties.
We thus explain the aforementioned result as a consequence of a more general theory.
%
We begin by defining several focal properties of games that we are interested in learning.


\begin{definition}[Normal-Form Game]
A \mydef{normal-form game} (NFG) $\smash{\GameTuple \doteq \langle \SetOfPlayers, \{ \StrategySet_\PlayerIndex \}_{\PlayerIndex \in \SetOfPlayers},
\Utility \rangle}$
consists of a set of players $\SetOfPlayers$,
with \mydef{pure strategy set} $\smash{\StrategySet_\PlayerIndex}$ available to player $\PlayerIndex \in \SetOfPlayers$.
We define $\smash{\StratProfileSpace \doteq \StrategySet_1 \times \dots \times \StrategySet_{\NumberOfPlayers}}$ to be the \mydef{pure strategy profile space}, and then $\smash{\Utility : \StratProfileSpace \to \R^{\NumberOfPlayers}}$ is a vector-valued utility function (equivalently, a vector of $\NumberOfPlayers$ scalar utility functions $\smash{\Utility_\PlayerIndex}$).  
\end{definition}

Given an NFG $\GameTuple$ with finite $\smash{\StrategySet_\PlayerIndex}$ for all $\PlayerIndex \in \SetOfPlayers$, we denote by $\smash{\StrategySet_\PlayerIndex^\diamond}$ the set of distributions over $\smash{\StrategySet_\PlayerIndex}$; this set is called player $\PlayerIndex$'s \mydef{mixed strategy set}. We define $\smash{\MixedStratProfileSpace = \MixedStrategySet_1 \times \dots \times \MixedStrategySet_{\NumberOfPlayers}}$ to be the \mydef{mixed strategy profile space}, and then, overloading notation, we write $\Utility(\StratProfile)$
to denote the expected utility of a mixed strategy profile $\smash{\StratProfile \in \MixedStratProfileSpace}$.
We denote this \emph{mixed\/} game that comprises mixed strategies $\MixedStratProfileSpace$ by $\GameTuple^{\diamond}$.

We call two NFGs with the same player sets and strategy profile spaces \mydef{compatible}.
Our goal in this paper is to develop algorithms for estimating one NFG by another compatible one, by which we 
mean estimating one game's \emph{utilities} by another's.
Thus, we define the game properties of interest in terms of $\Utility$ rather than $\GameTuple$, as the players and their strategy sets are usually clear from context.
(Likewise, we write $\Utility^{\diamond}$, rather than $\GameTuple^{\diamond}$.)
We assume a NFG 
$\Utility$ in the definitions that follow.

\begin{definition}[Property]
A \mydef{property} of a game $\Utility$ is a functional mapping an index set $\mathcal{X}$ and utilities $\Utility$ to real values:
i.e., $f: \mathcal{X} \times (\StratProfileSpace \to \R^{\NumberOfPlayers}) \to \R$.
\end{definition}

In this work, two common choices for $\mathcal{X}$ are the set of mixed and pure strategy profiles $\MixedStratProfileSpace$ and $\StratProfileSpace$, respectively.
Another plausible choice is the set of pure strategies for just one player $\PlayerIndex$, namely $\StrategySet_{\PlayerIndex}$.

In this section, we focus on four properties, which 
we find to be well behaved: power-mean welfare, Gini social welfare, adversarial values, and regret.

\begin{definition}[Power-Mean Welfare]
Given strategy profile $\StratProfile$, power $\rho \in \R$, and stochastic weight vector $\bm{w} \in \triangle^{\NumberOfPlayers}$, i.e., $\bm{w} \in 
\R_{0+}^{\NumberOfPlayers}$ s.t.\ $\norm{\bm{w}}_{1} = 1$, the \mydef{$\rho$-power-mean welfare}%
\footnote{Power-mean welfare is more precisely defined as $\lim_{\rho' \to \rho} \sqrt[\rho']{\bm{w} \cdot \Utility (\StratProfile)^{\rho'}}$,
to handle the special cases when $\rho$ is $0$ or $\pm\infty$.}
at $\StratProfile$ is defined as $\Welfare_{\rho, \bm{w}} (\StratProfile; \Utility) \doteq 
\sqrt[\rho]{\bm{w} \cdot \Utility (\StratProfile)^{\rho}}$.
\end{definition}

A few special cases of $\rho$-power mean welfare are worth mentioning.
When $\rho = 1$, power-mean welfare corresponds to utilitarian welfare, while when $\rho = -\infty$, power-mean welfare corresponds to egalitarian welfare.
Finally, when $\rho = 0$, taking limits yields
$\Welfare_{0, (\frac{1}{\NumberOfPlayers}, \ldots, \frac{1}{\NumberOfPlayers})} (\StratProfile; \Utility) \doteq \sqrt[|\SetOfPlayers|]{\prod_{\PlayerIndex \in \SetOfPlayers}  \Utility_\PlayerIndex (\StratProfile)}$, which defines Nash social welfare~\citep{nash1953bargaining}.

\newcommand{\Gini}{\Welfare}


\begin{definition}[Gini Social Welfare \citep{weymark1981generalized}]
Given a strategy profile $\StratProfile$ and a decreasing stochastic weight vector $\bm{w}^{\downarrow} \in \triangle^{\NumberOfPlayers}$,
the \mydef{Gini social welfare} at $\StratProfile$ is defined as
$\Gini_{\bm{w}^{\downarrow}} (\StratProfile; \Utility) \doteq 
\bm{w}^{\downarrow} \cdot \Utility^{\uparrow}(\StratProfile)$,
where $\Utility^{\uparrow}(\StratProfile)$ denotes the entries in $\Utility(\StratProfile)$ in ascending sorted order.
\end{definition}

The weights vector $\bm{w}$ controls the trade-off between society's attitude towards well-off and impoverished players.
We recover utilitarian welfare with $\bm{w} = (\frac{1}{\NumberOfPlayers}, \frac{1}{\NumberOfPlayers}, \dots, \frac{1}{\NumberOfPlayers})$, and egalitarian welfare with $\bm{w} = (1, 0, \dots, 0)$.

The last two properties we consider, adversarial values and regret, relate to solution concepts.
A player's \mydef{adversarial value} for playing a strategy is the value they obtains assuming worst-case behavior on the part of the other players: i.e., assuming all the other players were out to get them.
A player's \mydef{regret} for playing one strategy measures how much they regret not playing another, fixing all the other players' strategies.

\begin{definition}[Adversarial Values]
A player $\PlayerIndex$'s \mydef{adversarial value} at strategy $\tilde{\Strategy} \in \MixedStrategySet_\PlayerIndex$ is defined as $\AdversarialValue_{\PlayerIndex} (\tilde{\Strategy}; \Utility) \doteq \inf_{\StratProfile \in \StratProfileSpace \mid \StratProfile_\PlayerIndex = \tilde{\Strategy}} \Utility_\PlayerIndex (\StratProfile)$.%
\footnote{Equipping the players other than $\PlayerIndex$ with mixed strategies affords them no added power.}
A player $\PlayerIndex$'s \mydef{maximin value} is given by $\MinimaxValue_{\PlayerIndex} (\Utility) \doteq \sup_{\tilde{\Strategy} \in \MixedStrategySet_\PlayerIndex} \AdversarialValue_{\PlayerIndex} (\tilde{\Strategy}; \Utility) =
\sup_{\tilde{\Strategy} \in \MixedStrategySet_\PlayerIndex} 
\inf_{\StratProfile \in \StratProfileSpace \mid \StratProfile_\PlayerIndex = \tilde{\Strategy}} \Utility_\PlayerIndex (\StratProfile)$.
A strategy $\tilde{\Strategy}$ is $\epsilon$-\mydef{maximin optimal} for player $\PlayerIndex$ if $ \AdversarialValue_{\PlayerIndex} (\tilde{\Strategy}; \Utility) \geq \MinimaxValue_{\PlayerIndex} (\Utility) - \epsilon$.
\end{definition}

Fix a player $\PlayerIndex$ and a strategy profile $\StratProfile \in \StratProfileSpace$.
We define $\Adjacent_{\PlayerIndex, \StratProfile} \doteq \{ \StratProfileAlt \in \StratProfileSpace \mid \StratProfileAlt_q = \StratProfile_q, \forall q \neq \PlayerIndex \}$: i.e., the set of adjacent strategy profiles, meaning those in which the strategies of all players $q \ne \PlayerIndex$ are fixed at $\StratProfile_q$,  while $\PlayerIndex$'s strategy may vary.


\begin{definition}[Regret]
A player $\PlayerIndex$'s \mydef{regret} at 
$\StratProfile \in \StratProfileSpace$ is defined as $\Regret_\PlayerIndex (\StratProfile; \Utility) \doteq \sup_{\StratProfile' \in \Adjacent_{\PlayerIndex, \StratProfile}} \Utility_{\PlayerIndex} (\StratProfile') - \Utility_{\PlayerIndex} (\StratProfile)$, with
$\Regret (\StratProfile; \Utility) \doteq \max_{\PlayerIndex \in \SetOfPlayers} \Regret_{\PlayerIndex} (\StratProfile; \Utility)$.
\end{definition}

Note that $\Regret_\PlayerIndex (\StratProfile; \Utility)\ge 0$, since player $\PlayerIndex$ can deviate to any strategy,
including $\StratProfile_\PlayerIndex$ itself.
Hence, $\Regret (\StratProfile; \Utility) \ge 0$.
A strategy profile $\StratProfile \in \StratProfileSpace$ that has regret at most $\epsilon \ge 0$
is a called an \mydef{$\epsilon$-Nash equilibrium}~\citep{nash1950equilibrium}: i.e., $\StratProfile$ is an $\epsilon$-Nash equilibrium if and only if $0 \le \Regret (\StratProfile; \Utility) \le \epsilon$.

\paragraph{Lipschitz Continuous Game Properties}
\label{sec:lipschitz}

Next, we define Lipschitz continuity, and show that the aforementioned game properties are all Lipschitz continuous in utilities.
Then later, we define uniform approximation, and prove that these properties and their extrema are all well behaved, 
meaning they can be well estimated via a uniform approximation of the game.

\begin{definition}[Lipschitz Property]
Given $\lambda \ge 0$,
a \mydef{$\lambda$-Lipschitz property} is one that is $\lambda$-Lipschitz continuous in utilities: i.e.,
$\norm{f(\cdot; \Utility) - f(\cdot; \Utility')}_{\infty} \doteq \sup_{x \in \mathcal{X}} \abs{f(x; \Utility) - f(x; \Utility')} \leq \lambda\norm{\Utility - \Utility'}_\infty$,
for all pairs of compatible games $\Utility$ and $\Utility'$.
\label{def:lipschitz}
\end{definition}

To show that the game properties of interest are Lipschitz properties, we instantiate $\mathcal{X}$ and the property $f$ in Definition~\ref{def:lipschitz} as follows:
1.~Power-Mean Welfare:
Let $\mathcal{X} = \StratProfileSpace$ and
$f(\StratProfile; \Utility) = \Welfare_{\rho, \bm{w}} (\StratProfile; \Utility)$, for some $\rho \in \R$ and $\bm{w} \in \triangle^{\NumberOfPlayers}$.
%
2.~Gini Social Welfare:
Let $\mathcal{X} = \StratProfileSpace$ and
$f(\StratProfile; \Utility) = \Gini_{\bm{w}^{\downarrow}} (\StratProfile; \Utility)$, for some $\bm{w}^{\downarrow} \in \triangle^{\NumberOfPlayers}$.
%
3.~Adversarial Values: For 
$\PlayerIndex \in \SetOfPlayers$, let $\mathcal{X} = \StrategySet_{\PlayerIndex}$ and
$f_{\PlayerIndex} = \AdversarialValue_{\PlayerIndex}$: i.e.,
$f_{\PlayerIndex} (\tilde{\Strategy}; \Utility) = \AdversarialValue_{\PlayerIndex} (\tilde{\Strategy}; \Utility)$.
%
4.~Regret: Let $\mathcal{X} = \StratProfileSpace$ and $f = \Regret$:
i.e., $f(\StratProfile; \Utility) = \Regret (\StratProfile; \Utility)$.

Results on gradients and Lipschitz constants for power-mean welfare can be found in \citet{beliakov2009some,cousins2021axiomatic,cousins2022uncertainty}.
In short, whenever $\rho \ge 1$, power-mean welfare is Lipschitz continuous with $\lambda = 1$.
It is also $\max_{\PlayerIndex \in \SetOfPlayers} \bm{w}_{\PlayerIndex}^{\nicefrac1\rho}$-Lipschitz continuous for $\rho \in [-\infty,0)$,
but it is Lipschitz-discontinuous for $\rho \in [0, 1)$.

Regarding the latter three properties, Gini social welfare and adversarial value are both $1$-Lipschitz properties, while regret is a $2$-Lipschitz property.
These proofs of these claims rely on a ``Lipschitz calculus'' 
(\Cref{thm:lipschitz-facts}), which is a straightforward consequence of the definition of Lipschitz continuity.
The interested reader is invited to consult \citet{heinonen2005lectures} for details.
%
We use this Lipschitz calculus to prove the aforementioned Lipschitz continuity claims 
(see Appendix~\ref{app:lipschitz_proofs}).%
\footnote{All proofs are deferred to the appendix.}

The property $f$ that computes a convex combination of utilities is 1-Lipschitz by the linear combination rule (\Cref{thm:lipschitz-facts}), because utilities are 1-Lipschitz in themselves.
Consequently, any findings about the Lipschitz continuity of game properties immediately apply to games with mixed strategies, because any $\lambda$-property $g$ of a game $\Utility$ may be composed with $f$ to arrive at a property $g \circ f$ of the mixed game $\Utility^{\diamond}$, which, by the composition rule  (\Cref{thm:lipschitz-facts}), is $\lambda$-Lipschitz.

The Lipschitz calculus (Theorem~\ref{thm:lipschitz-facts}) implies that any findings about the Lipschitz continuity of game properties immediately apply to mixed games, i.e., games with mixed strategies (Lemma~\ref{lem:mixed}).
Moreover, any $\lambda$-property $g$ of a game $\Utility$
may be composed with the functional $f$ in \Cref{lem:mixed} to arrive at a property $g \circ f$ of the mixed game $\Utility^{\diamond}$, which, by composition, is then $\lambda$-Lipschitz.

\paragraph{Uniform Approximations of Game Properties}



Next we observe that the Lipschitz properties of a game are well behaved, and thus can be well estimated by a uniform approximation of the game.
Thus, all four of our focal game-theoretic properties are well behaved.
Using this observation, we proceed to show that the extrema of such properties (e.g., optimal and pessimal welfare, maximin values, etc.) are also well behaved, and that this result extends to their witnesses, all of which remain approximately optimal in their approximate game counterparts.

This final result can be understood as a form of recall and (approximate) precision:
recall, because the set of witnesses of the approximate game $\hat{\GameTuple}$ contains all the true positives (i.e., all the witnesses in the game $\GameTuple$); 
precision, because all false positives (witnesses in the approximate game $\hat{\GameTuple}$ that are not witnesses in the game $\GameTuple$) are nonetheless approximate witnesses in the game $\GameTuple$.
Taking the property of interest to be regret, so that its minima correspond to Nash equilibria, this final result can be understood as the precision and recall result obtained in \citet{areyan2019learning}.

In a uniform approximation of one NFG by a compatible one, the bound between utility deviations in the two games holds uniformly over \emph{all\/} players and strategy profiles:
%
One game $\Utility$ is said to be a \mydef{uniform $\epsilon$-approximation} of another game $\Utility'$ whenever $\norm{\Utility - \Utility'}_{\!\infty} \doteq \sup_{\PlayerIndex \in \SetOfPlayers, \StratProfile \in \StratProfileSpace} \abs{\Utility_{\PlayerIndex} (\StratProfile) -\Utility'_{\PlayerIndex} (\StratProfile)} \leq \epsilon$.
%
%
Our first observation, which follows immediately from the definitions of Lipschitz property 
and uniform approximation, characterizes well-behaved game properties:

\begin{observation}
\label{obs:lipschitz-uniform}
Let $\lambda, \epsilon \ge 0$.
If $f$ is $\lambda$-Lipschitz
and $\norm{\Utility - \Utility'}_{\infty} \leq \epsilon$,
then
$\norm{f(\cdot; \Utility) - f(\cdot; \Utility')}_{\infty} 
\leq \lambda \epsilon$.
%
Further, if $\mathcal{X} = \{ x \}$ is a singleton, then $\abs{f(x; \Utility) - f(x; \Utility')} \leq \lambda \epsilon$.
\end{observation}

Equivalently, if $f$ is $\lambda$-Lipschitz
and $\norm{\Utility - \Utility'}_{\infty} \leq \nicefrac{\epsilon}{\lambda}$,
then
$\norm{f(\cdot; \Utility) - f(\cdot; \Utility')}_{\infty} \leq \lambda (\nicefrac{\epsilon}{\lambda}) = \epsilon$.
For example, regret is 2-Lipschitz, and thus can be $\epsilon$-approximated, given an $\nicefrac{\epsilon}{2}$-uniform approximation.
Likewise for adversarial values and variants of welfare, with their corresponding Lipschitz constants.

Based on this observation, we derive two-sided approximation bounds on the \emph{extrema\/} of $\lambda$-Lipschitz game-theoretic properties.
Specifically, we derive a two-sided bound on their values,
and a dual containment (recall and approximate precision) result characterizing their witnesses.

\begin{restatable}[Approximating Extremal Properties of Normal-Form Games]{theorem}{thmExtremal}
\label{thm:lipschitz-properties}
Let $\epsilon, \lambda, \alpha > 0$.
Given a $\lambda$-Lipschitz property $f$, the following hold:
%
%
1.~\emph{Approximately-optimal values}: 
$\bigl\lvert \smallsup{x \in \mathcal{X}} f(x; \Utility) - \smallsup{x \in \mathcal{X}} f(x; \Utility') \bigr\rvert \leq \lambda\epsilon$ and $\bigl\lvert \smallinf{x \in \mathcal{X}} \hspace{-1.5mm} f(x; \Utility) - \smallinf{x \in \mathcal{X}} \hspace{-1.5mm} f(x; \Utility') \bigr\rvert \leq \lambda\epsilon$.
%
2.~\emph{Approximately-optimal witnesses}: for some $\hat{x} \in \mathcal{X}$, if $\hat{x}$ is $\alpha$-optimal according to $\Utility'$, then $\hat{x}$ is $2\lambda\epsilon + \alpha$-optimal according to $\Utility$: i.e.,
if $\displaystyle f(\hat{x}; \Utility') \geq \sup_{x \in \mathcal{X}} f(x; \Utility') - \alpha$, then $\displaystyle f(\hat{x}; \Utility) \geq \sup_{x \in \mathcal{X}} f(x; \Utility) - 2 \lambda \epsilon - \alpha$, and if $\displaystyle f(\hat{x}; \Utility') \leq \inf_{x \in \mathcal{X}} f(x; \Utility') + \alpha$, then $\displaystyle f(\hat{x}; \Utility) \leq \inf_{x \in \mathcal{X}} f(x; \Utility) + 2 \lambda \epsilon + \alpha$.
\end{restatable}

\noindent
Applying this theorem to welfare and adversarial values yields the following corollary.

\begin{corollary}
\label{cor:lipschitz-properties}
The following hold:
1.~Welfare: Let $\Welfare$ denote a $\lambda$-Lipschitz welfare function. Then
$\bigl\lvert \smallsup{x \in \mathcal{X}} \Welfare(x; \Utility) - \smallsup{x \in \mathcal{X}} \Welfare(x; \Utility') \bigr\rvert \leq \lambda\epsilon$ and $\bigl\lvert \smallinf{x \in \mathcal{X}} \Welfare(x; \Utility) - \smallinf{x \in \mathcal{X}} \Welfare(x; \Utility') \bigr\rvert \leq \lambda\epsilon$.
%
2.~Maximin-Optimal Strategies:
If strategy $\tilde{\Strategy} \in \StrategySet_{\PlayerIndex}$ is $\alpha$-maximin optimal for player $\PlayerIndex$ in $\Utility'$, so that $\AdversarialValue_{\PlayerIndex} (\tilde{\Strategy}; \Utility') \geq \MinimaxValue_{\PlayerIndex} (\Utility') - \alpha$, then it is $2 \epsilon - \alpha$ maximin-optimal in $\Utility$, meaning $\AdversarialValue_{\PlayerIndex} (\tilde{\Strategy}; \Utility) \geq \MinimaxValue_{\PlayerIndex} (\Utility) - 2 \epsilon - \alpha$.
\end{corollary}

We can likewise obtain a dual containment result as a corollary of Theorem~\ref{thm:lipschitz-properties} by applying it to regret.
Since the extreme value is known (at equilibrium, regret is 0), we obtain a stronger result.

\begin{restatable}[Approximating Witnesses of Normal-Form Games]{theorem}{thmWitness}
\label{thm:lipschitz-properties-target-case}
If $v^*$ denote a target value of a $\lambda$-Lipschitz property $f$ and $F_{\alpha} (\Utility) = \{ x \mid | f(x; \Utility) - v^* | \le \alpha \}$, then
$F_0 (\Utility) 
    \subseteq 
F_{\lambda \epsilon} (\Utility') 
    \subseteq 
F_{2 \lambda \epsilon} (\Utility)$.
\end{restatable}

Finally, 
we recover the dual containment theorem
of~\citet{areyan2019learning}.

\begin{corollary}[Approximating Equilibria in NFGs]
\label{cor:dual_containment}
If $\Nash_{\alpha} (\Utility) = \{ x \mid \Regret(x; \Utility) \le \alpha \}$,
then
$\Nash_{0} (\Utility) 
    \subseteq 
\Nash_{2\epsilon} (\Utility') 
    \subseteq 
\Nash_{4\epsilon} (\Utility)$
and
$\Nash^\diamond_{0} (\Utility) 
    \subseteq 
\Nash^\diamond_{2\epsilon} (\Utility') 
    \subseteq 
\Nash^\diamond_{4\epsilon} (\Utility)$.
\end{corollary}

\section{Learning Framework}
\label{sec:learn}

In this section, we move on from approximating properties of games to actively learning them via repeated sampling.
We present our formal model of simulation-based games, along with 
uniform convergence bounds that provide guarantees about \mydef{empirical games}, which comprise estimates of the expected utilities of simulation-based games.
These bounds form the basis of our progressive sampling algorithm, which we present in the next section.
This algorithm learns empirical games that uniformly approximate their expected counterparts,
which implies that well-behaved game properties of simulation-based games are well estimated by their empirical counterparts.

\paragraph{Empirical Game Theory}
\label{sec:empirical_games}

We start by developing a formalism for modeling simulation-based games, which we use to formalize empirical games, the main object of study in EGTA.


\begin{definition}[Conditional
NFG]
\label{def:conditional-game}
A \mydef{conditional 
NFG}
$\ConditionalGame{\ConditionSpace} \doteq \langle \ConditionSpace, \SetOfPlayers, \{ \StrategySet_\PlayerIndex \}_{\PlayerIndex \in \SetOfPlayers} , \Utility(\cdot) \rangle$ consists of a set of conditions $\smash{\ConditionSpace}$, a set of agents $\smash{\SetOfPlayers}$, with pure strategy set $\smash{\StrategySet_\PlayerIndex}$ available to agent $\PlayerIndex$, and a vector-valued conditional utility function $\smash{\Utility : \StratProfileSpace \times \ConditionSpace \to \R^{\NumberOfPlayers}}$.
Given a condition $\smash{\ConditionValue \in \ConditionSpace}$, $\Utility(\cdot; \ConditionValue)$ yields a standard utility function of the form $\smash{\StratProfileSpace \to \R^{\NumberOfPlayers}}$. 
\end{definition}


As the name suggests, a conditional NFG depends on a condition.
A \emph{condition} is a very general notion.
It can incorporate an entropy source (i.e., a random uniform
sample from $[0,1]$) based on which the simulator can generate any randomness endogenous to the game.
It can also represent 
any number of exogenous 
variables (e.g., the weather) that might influence the agents' utilities.
For example, in an auction, the bidders' private values for the goods are exogenous variables, 
whereas a randomized tie-breaking rule would constitute endogenous randomness.
%
Representing all potential sources of randomness as a single random condition simplifies our formal model of simulation-based games, without sacrificing generality.%
\footnote{In practice, only the exogenous random condition
would be passed to the simulator, while the simulator itself produces potentially stochastic utilities as per the rules of the game.}

\begin{definition}[Expected 
NFG]
\label{def:expected-game}
Given a conditional 
NFG $\smash{\ConditionalGame{\ConditionSpace}}$ and a distribution $\ConditionDistribution$ over the set of conditions $\ConditionSpace$, we define the 
\mydef{expected utility function} $\smash{\Utility(\StratProfile; \ConditionDistribution) = \Expwrt{\ConditionValue \distributed \ConditionDistribution}{\Utility(\StratProfile; \ConditionValue)}}$, and the corresponding 
\mydef{expected 
NFG} as $\smash{\InducedGame{\ConditionDistribution} \doteq \langle \SetOfPlayers, \{ \StrategySet_\PlayerIndex \}_{\PlayerIndex \in \SetOfPlayers} \}, \Utility(\cdot; \ConditionDistribution) \rangle}$.
\end{definition}

Expected
NFGs serve as our mathematical model of simulation-based games.
Given the generality of the conditional NFGs on which they are based, expected NFGs are sufficient to model 
simulation-based games where the rules are deterministic but the initial conditions are random
(e.g., poker, where the cards dealt are random, or auctions where bidders' valuations are random); as well as 
games with randomness throughout (e.g., computerized backgammon, where the computer rolls the dice, or auctions where ties are broken randomly). 

Since the only available access to an expected NFG (i.e., a simulation-based game) is via a simulator, a key component of EGTA methodology is an empirical estimate of the game, 
which is produced by sampling random conditions from $\ConditionDistribution$, and then querying the simulator to obtain sample utilities at various strategy profiles. 

\begin{definition}[Empirical
NFG]
\label{def:empirical-game}
Given a conditional
NFG $\smash{\ConditionalGame{\ConditionSpace}}$ together with a distribution $\ConditionDistribution$ over the set of conditions $\ConditionSpace$ from which we draw samples $\smash{\Samples = (\SamplePoint_1, \ldots, \SamplePoint_\NumberOfSamples) \distributed \ConditionDistribution^\NumberOfSamples}$ and then place $q(\StratProfile)\leq \NumberOfSamples$ simulation queries for each strategy profile $\StratProfile \in \StratProfileSpace$, we define the 
\mydef{empirical utility function} 
$\smash{\hat{\Utility}(\StratProfile; \Samples) \doteq \frac{1}{q(\StratProfile)}\sum_{\SampleIndex=1}^{q(\StratProfile)} \Utility(\StratProfile; \SamplePoint_\SampleIndex)}$, and the corresponding \mydef{empirical 
NFG}
as $\smash{\EmpiricalGame{\Samples} \doteq \langle \SetOfPlayers, \{ \StrategySet_\PlayerIndex \}_{\PlayerIndex \in \SetOfPlayers} \}, \hat{\Utility}(\cdot ; \Samples)\rangle}$.
\end{definition}


How expensive it is to sample random conditions depends on the sources of the randomness.
The distributions over the exogenous sources of randomness are unknown to the simulator, and thus potentially expensive data collection (e.g., market research) is required to gather samples.
In contrast, the distributions of endogenous sources of randomness are known (to the simulator, though not necessarily to the game analyst), and hence sampling from a single entropy source is sufficient for the simulator to simulate any randomness that is endogenous to the game.

We define \mydef{data complexity} as the number of samples drawn from distribution $\ConditionDistribution$, and \mydef{query complexity} as the number of times the simulator is queried to compute $\Utility(\StratProfile; \ConditionValue)$, given strategy profile $\StratProfile\in\StratProfileSpace$ and random condition $\ConditionValue\in \ConditionSpace$.
%
%
As noted above, query complexity is the natural metric for computation-bound settings (e.g., Starcraft~\cite{tavares2016rock}), while data complexity is more appropriate for data-intensive applications (e.g., auctions).

\paragraph{Tail Bounds}

Recall 
that it is necessary to approximate a game to within $\nicefrac{\epsilon}{\lambda}$ to guarantee that a $\lambda$-Lipschitz property of that game can be approximated via an $\epsilon$-uniform approximation.
Our present goal, then, is to design algorithms that ``uniformly estimate'' empirical games from finitely many samples,
so that we can apply the machinery of Theorems~\ref{thm:lipschitz-properties} and~\ref{thm:lipschitz-properties-target-case} to make inferences about the quality of the equilibria and other well-behaved properties of simulation-based games.

This approach to learning empirical game properties uses
tail bounds, to show that each
utility is, with high probability, close to its empirical estimate.
Then, by applying a union bound, we can generate a uniform guarantee for the game: i.e., for all its utilities, simultaneously.

Much of the variation among methods comes from the choice of tail bounds.
The most straightforward choice of tail bound is Hoeffding's inequality, which was used by \citet{tuyls2020bounds}.
Like us, \citet{areyan2020improved} also used a variant of Bennett's inequality.
In the remainder of this section, we 
report data and query complexities for which these tail bounds guarantee, with high probability, an $\epsilon$-uniform approximation of a simulation-based game.
All of our results depend on the following assumption, which we make throughout:

\newtheorem{assumption}{Assumption}
\begin{assumption}
We consider a finite, conditional 
game $\ConditionalGame{\ConditionSpace}$ together with distribution $\ConditionDistribution$ such that for all $\smash{\ConditionValue \in \ConditionSpace}$ and $\smash{(\PlayerIndex, \StratProfile) \in \SetOfPlayers\times \StratProfileSpace}$, it holds that $\smash{\Utility_{\PlayerIndex}(\StratProfile; \ConditionValue) \in [-{\nicefrac{\UtilityRange}{2}}, {\nicefrac{\UtilityRange}{2}}]}$,
where $\UtilityRange \in \mathbb{R}$.
\end{assumption}

Not surprisingly, the number of parameters being estimated, influences both the data and query complexities.
We use $\abs{\GameTuple}$ to denote the game size, i.e., the number of (scalar) utilities being estimated, and $\UtilityIndices\subseteq \SetOfPlayers\times \StratProfileSpace$ to index over these utilities.

In a basic normal-form representation of a simulation-based game, the game size is simply the number of players times the number of strategy profiles: i.e., $\abs{\GameTuple} = \NumberOfPlayers \cdot \abs{\StratProfileSpace}$. 
In many games, however, structure and symmetry can be exploited, which render the effective size of the game much smaller.
In Tic-Tac-Toe, for example, since many strategies are equivalent up to either rotation or reflection, the effective number of strategy profiles can be dramatically less than $\abs{\StratProfileSpace}$.
Furthermore, being a zero-sum game, the utilities of one player are determined by those of the other, and hence only the utilities of one of the $\abs{\SetOfPlayers}$ players requires estimation.
In other words, although there are two players in Tic-Tac-Toe and an intractable number of strategy profiles,
\samy{}{the game size $\abs{\GameTuple}$ is effectively much much smaller than $\abs{\SetOfPlayers} \cdot \abs{\StratProfileSpace}$, as only a small portion of all utilities need to be estimated}.

\paragraph{Hoeffding's Inequality}
\label{sec:hoeffding}

Hoeffding's inequality for sums of independent bounded random variables can be used to obtain tail bounds on the probability that an empirical mean differs greatly from its expectation.
We can use this inequality to estimate a single utility value, and then apply a union bound to estimate all utilities simultaneously.

\begin{theorem}[Finite-Sample Bounds for Expected 
NFGs via Hoeffding's Inequality]
\label{thm:hoeffding}
\label{thm:Hoeffding}
Given $\epsilon > 0$, Hoeffding's Inequality guarantees a data complexity of $m^H (\epsilon, \nicefrac{\delta}{\abs{\GameTuple}}, \UtilityRange) \doteq \frac{\UtilityRange^{2}\ln \left( \nicefrac{2 \abs{\GameTuple}}{\delta} \right)}{2\epsilon^{2}}$,
and a query complexity of $\abs{\StratProfileSpace} + m^H(\epsilon, \nicefrac{\delta}{\abs{\GameTuple}}, \UtilityRange \sqrt{\abs{\StratProfileSpace}}) = \abs{\StratProfileSpace} +  \frac{\UtilityRange^2\abs{\StratProfileSpace}\ln\left(\nicefrac{2\abs{\GameTuple}}{\delta}\right)}{2\epsilon^2}$. 
\end{theorem}


In other words, if we have access to at least $m^H(\nicefrac{\delta}{\abs{\GameTuple}}, \UtilityRange)$ samples from $\ConditionDistribution$, and can place at least $\abs{\StratProfileSpace} + m^H(\epsilon, \nicefrac{\delta}{\abs{\GameTuple}}, \UtilityRange\sqrt{\abs{\StratProfileSpace}})$ simulator queries, Hoeffding's Inequality guarantees that we can produce an empirical NFG $\hat{\GameTuple}$ that is an $\epsilon$-uniform approximation of the expected NFG $\GameTuple_\ConditionDistribution$ with probability at least $1 - \delta$.
This empirical NFG can be generated by querying the simulator at each strategy profile, for each of $\lceil m^H(\epsilon, \nicefrac{\delta}{\abs{\GameTuple}}, \UtilityRange)\rceil$ samples from $\ConditionDistribution$.

\paragraph{Bennett's Inequality, Known Variance}
\label{sec:bennett}


I.i.d.\ random variables $X_{1:m}$ with mean $\overline{X}$ are termed $\smash{\vargaussian}$-\mydef{sub-Gaussian} if they obey the Gaussian-Chernoff bound: i.e.,
$\Prob\!\left( \overline{X} \geq \Expect[\overline{X}] + \epsilon \right) \leq \exp\! \left( \frac{-\NumberOfSamples \epsilon^2}{2 \vargaussian} \right)$; 
equivalently,
$\Prob\!\left( \overline{X} \geq \Expect[\overline{X}] + \sqrt{\frac{2 {\sigma_{\mathcal{N}}^{\smash{2}}} \ln (\smash{\frac{1}{\delta}})}{\NumberOfSamples}} \right) \leq \delta$,
where $\smash{\vargaussian}$ is 
a \mydef{variance proxy}.
%
Using this characterization, Hoeffding's inequality
reads
``If $X_{i}$ has range $\UtilityRange$, then $X_{i}$ is $\mathsmaller{\nicefrac{\UtilityRange^2}{4}}$-sub-Gaussian,''
which matches the Gaussian bound under a worst-case assumption about variance,
because,
by Popoviciu's inequality (\citeyear{popoviciu1935equations}), the variance $\Var[X_{i}] \leq \mathsmaller{\nicefrac{\UtilityRange^2}{4}}$.
Thus, Hoeffding's inequality yields sub-Gaussian tail bounds; 
however, as it is stated in terms of the \emph{largest possible variance}, it is a loose
bound when $\Var[X_{i}] \ll \mathsmaller{\frac{\UtilityRange^2}{4}}$. 

The central limit theorem guarantees that $X_i$ is $\Var[X_i]$-sub-Gaussian, asymptotically (i.e., as $\NumberOfSamples\to \infty$); however, it provides no finite-sample guarantees.
By accounting for the range $\UtilityRange$ as well as the variance, we can derive finite-sample Gaussian-like tail bounds.

A ($\sigma_{\Gamma}^{2}, \UtilityRange_\Gamma$)-\mydef{sub-gamma}~\citep{boucheron2013concentration} random variable obeys
$\Prob\!\left( \overline{X} \geq \Expect[\overline{X}] + \frac{\UtilityRange_\Gamma \ln(\frac{1}{\delta})}{3\NumberOfSamples} + \sqrt{\frac{2\sigma_{\Gamma}^{\smash{2}}\ln(\smash{\frac{1}{\delta}})}{\NumberOfSamples}} \right) \leq \delta$.
%
This tail bound asymptotically matches the $\sigma_{\Gamma}$-sub-Gaussian tail bound, with an additional scale-dependent term acting as an asymptotically negligible correction for working with non-Gaussian distributions.
The key to understanding the tail behavior of sub-gamma random variables is to observe that the error consists of a \emph{hyperbolic} (fast-decaying) \emph{scale term}, $\slfrac{\UtilityRange_\Gamma \ln(\frac{1}{\delta})}{3\NumberOfSamples}$, and a \emph{root-hyperbolic} (slow-decaying) \emph{variance term}, ${\sqrt{\slfrac{2\vargamma\ln(\smash{\frac{1}{\delta}})}{\NumberOfSamples}}}$. 
Sub-gamma random variables thus yield mixed convergence rates, which decay quickly initially, while the $\UtilityRange_\Gamma$ term dominates, before slowing to the root-hyperbolic rate once the $\vargamma$ term comes to dominate.

While Bennett's inequality (\citeyear{bennett1962probability}) is usually stated as a sub-Poisson bound,%
\footnote{I.e., a bound of the form ${\Prob\!\left( \overline{X} \geq \Expect[\overline{X}] + \epsilon \right) \leq \exp\left(-\frac{\NumberOfSamples\sigma_{\GameTuple}^{\smash{2}}}{\UtilityRange_\GameTuple}h\left(\frac{\UtilityRange_\GameTuple \epsilon}{\sigma_{\GameTuple}^{\smash{2}}}\right)\right)}$, where $h(x)\doteq (1+x)\ln(1+x)-x$, for all $x\geq 0$.}
it immediately implies that if $X_{i}$ has range $\UtilityRange$ and variance $\sigma^2$, then $X_{i}$ is ($\sigma^2, \UtilityRange$)-sub-gamma.
We derive data and query complexity bounds that are consistent with this insight.


\begin{restatable}[Finite-Sample Bounds for Expected 
NFGs via Bennett's Inequality]{theorem}{thmBennett}
\label{thm:Bennett}
Given $\epsilon > 0$, Bennett's Inequality guarantees a data complexity of $$\smash{m^{B} (\epsilon, \nicefrac{\delta}{\abs{\GameTuple}}; \UtilityRange, \norm{\UtilityVariance}_{\infty})} \doteq \frac{\UtilityRange^2 \ln \left( \frac{2 \abs{\GameTuple}}{\delta} \right)}{\norm{\UtilityVariance}_\infty h \left( \frac{\UtilityRange\epsilon}{\norm{ \UtilityVariance}_\infty}\right)} \leq 2 \ln \frac{2 \abs{\GameTuple}}{\delta} \left( \frac{\UtilityRange}{3\varepsilon} + \frac{\norm{\UtilityVariance}_{\infty}}{\varepsilon^{2}} \right),$$
and a query complexity of
$\abs{\StratProfileSpace} + m^{B} (\epsilon, \nicefrac{\delta}{\abs{\GameTuple}}; \UtilityRange \abs{\StratProfileSpace}, \norm{\UtilityVariance}_{1,\infty}) \leq \abs{\StratProfileSpace} + 2 \ln \frac{2 \abs{\GameTuple}}{\delta} \left( \frac{\UtilityRange \abs{\StratProfileSpace}} {3 \varepsilon} + \frac{\norm{\UtilityVariance}_{1,\infty}}{\varepsilon^{2}} \right)$, 
where $\UtilityVariance_{\PlayerIndex} (\StratProfile) \doteq \Var_{\SamplePoint \distributed \ConditionDistribution}[ \Utility_{\PlayerIndex}(\StratProfile; \SamplePoint)]$ is the variance in the conditional game $\GameTuple_{\ConditionDistribution}$ of the utility $\Utility_{\PlayerIndex} (\StratProfile)$, so that, as usual, $\norm{\UtilityVariance}_{\infty}$ denotes the infinity norm of 
$\UtilityVariance$, namely
$\max_{\StratProfile \in \StratProfileSpace, \PlayerIndex \in \SetOfPlayers} \UtilityVariance_{\PlayerIndex}(\StratProfile)$.
Further, $\norm{\UtilityVariance}_{1,\infty}$ is defined as the $1$ norm over strategy profiles of the $\infty$ norm over players: i.e., the sum of the maxima over variances, namely
$\sum_{\StratProfile \in \StratProfileSpace} \max_{\PlayerIndex \in \SetOfPlayers} \UtilityVariance_{\PlayerIndex}(\StratProfile)$.
\end{restatable}


Hence, if variance is known, and if we have access to at least $m^B(\epsilon, \nicefrac{\delta}{\abs{\GameTuple}}; \UtilityRange, \norm{\UtilityVariance}_{\infty})$ samples from $\ConditionDistribution$ and can place at least $\abs{\StratProfileSpace} + m^B(\epsilon, \nicefrac{\delta}{\abs{\GameTuple}}; \UtilityRange \abs{\StratProfileSpace}, \norm{\UtilityVariance}_{1,\infty})$ simulator queries, then Bennett's Inequality guarantees that we can produce an empirical NFG $\hat{\GameTuple}$ that is an $\epsilon$-uniform approximation of the expected NFG $\GameTuple_\ConditionDistribution$ with probability at least $1 - \delta$.
This empirical NFG can be generated by querying the simulator at each strategy profile $\StratProfile\in\StratProfileSpace$ with only $\lceil m^B(\epsilon, \nicefrac{\delta}{\abs{\GameTuple}}; \UtilityRange, \norm{\UtilityVariance(\StratProfile)}_\infty)\rceil$ of the available samples, where $\norm{\UtilityVariance(\StratProfile)}_\infty\doteq \max_{\PlayerIndex\in\SetOfPlayers}\UtilityVariance_\PlayerIndex(\StratProfile)$ denotes the maximum variance at each strategy profile across all players.
Since $m^B(\epsilon, \nicefrac{\delta}{\abs{\GameTuple}}; \UtilityRange, \norm{\UtilityVariance(\StratProfile)}_\infty)$ is strictly increasing in $\norm{\UtilityVariance(\StratProfile)}_\infty$, this result is consistent with the intuition that strategy profiles with higher (resp.\ lower) variance require more (resp.\ fewer) queries to be well-estimated.

\paragraph{Bennett's Inequality, Empirical Variance}
\label{sec:empirical-bennett}

Bennett's inequality assumes the range and variance of the random variable
are \emph{known}.
Various \emph{empirical\/} Bennett bounds have been shown~\citep{audibert2007variance,audibert2007tuning,maurer2009empirical}, which all essentially operate by bounding the variance of a random variable in terms of its range and \emph{empirical\/} variance, and then applying Bennett's inequality.
These empirical bounds nearly match their non-empirical counterparts, with a larger scale-dependent term that also corrects for the empirical estimates of variance.

The next theorem forms the heart of our pruning algorithm.
To derive this theorem, we first bound variance in terms of empirical variance using a novel sub-gamma tail bound (\Cref{thm:generalEBennett}), which is a refinement of the one derived by~\citet{cousins2020sharp} and used in~\citet{areyan2020improved}.
Then, we apply Bennett's inequality to the upper and lower tails of each individual utility using said variance bounds.
Note that these non-uniform
bounds immediately imply a uniform approximation guarantee, by replacing the empirical variances with the maximum empirical variance $\norm{\EUtilityVariance}_\infty$.

We state this theorem in terms of arbitrary index sets $\UtilityIndices \subseteq \SetOfPlayers \times \StratProfileSpace$, so that our result generalizes to any compactly represented normal-form game in which not all utilities need be estimated, including extensive-form and graphical games, \samy{}{and the Tic-Tac-Toe example discussed earlier}.



\begin{restatable}[Bennett-Type Empirical-Variance Sensitive Finite-Sample Bounds]{theorem}{thmEBennett}
For all $(\PlayerIndex, \StratProfile) \in \UtilityIndices$,
let
\begin{align*}
\EUtilityVariance_{\PlayerIndex}(\StratProfile) 
&\doteq 
\mathsmaller{\frac{1}{\NumberOfSamples - 1}} 
\ssum_{j = 1}^{\NumberOfSamples} 
  \left( \Utility_{\PlayerIndex}(\StratProfile; \SamplePoint_j) -
  \hat{\Utility}_{\PlayerIndex}(\StratProfile; \Samples) \right)^2;\\
\epsilon_{\EUtilityVariance, (\PlayerIndex, \StratProfile)}
&\doteq \mathsmaller{\frac{2 \UtilityRange^{2} \ln \left( \frac{3 \abs{\GameTuple}}{\delta} \right)}{3 \NumberOfSamples} + \sqrt{\Bigl(\frac{1}{3} + \frac{1}{2 \ln \left( \frac{3\abs{\GameTuple}}{\delta} \right)}\Bigr) \Bigl( \frac{\UtilityRange^{2} \ln \left( \frac{3\abs{\GameTuple}}{\delta} \right)}{\NumberOfSamples - 1}\Bigr)^{\smash{2}} + \frac{2 \UtilityRange^{2} \EUtilityVariance_\PlayerIndex(\StratProfile) \ln \left( \frac{3\abs{\GameTuple}}{\delta} \right)}{\NumberOfSamples}}};\\
\epsilon_{\mu, (\PlayerIndex, \StratProfile)}
&\doteq \mathsmaller{\frac{\UtilityRange\ln \left( \frac{3 \abs{\GameTuple}}{\delta} \right)}{3 \NumberOfSamples} + \sqrt{\frac{2 (\EUtilityVariance_\PlayerIndex (\StratProfile) + \epsilon_{\EUtilityVariance, (\PlayerIndex, \StratProfile)}) \ln \left( \frac{3\abs{\GameTuple}}{\delta} \right)}{\NumberOfSamples}}}.
\end{align*}
%
Then, with probability at least $1 - \delta$,
it holds that
$\left\lvert \Utility_{\PlayerIndex}(\StratProfile; \ConditionDistribution) - \hat{\Utility}_{\PlayerIndex}(\StratProfile; \Samples) \right\rvert \leq \epsilon_{\mu, (\PlayerIndex, \StratProfile)}$ for all 
{$(\PlayerIndex, \StratProfile)\in \UtilityIndices$}. 
Furthermore, when $\nicefrac{\delta}{\abs{\GameTuple}}\leq 0.03$, it holds that 
$\epsilon_{\mu, (\PlayerIndex, \StratProfile)} \leq
  \underbrace{
  \mathsmaller{\frac{2 \UtilityRange \ln (\mathsmaller{\frac{3 \abs{\GameTuple}}{\delta}})}{\NumberOfSamples - 1}}}_{\textsc{Scale}}
  +
  \underbrace{
  \mathsmaller{\sqrt{\frac{2 \EUtilityVariance_\PlayerIndex (\StratProfile) \ln (\vphantom{\scalebox{0.99}{\ensuremath{|}}} \smash{\frac{3 \abs{\GameTuple}}{\delta}})}{\NumberOfSamples}}}}_{\textsc{Variance}}$ for all 
  {$(\PlayerIndex, \StratProfile)\in \UtilityIndices$}, matching Bennett's inequality up to constant factors, with dependence on $\EUtilityVariance$ instead of $\UtilityVariance$.
\label{thm:eBennett}
\end{restatable}

When $\EVarWimpy_\PlayerIndex(\StratProfile)$ is $\approx \nicefrac{\UtilityRange^2}{4}$ (near-maximal), this bound is strictly looser than \Cref{thm:Hoeffding}, with an additional scale term and a $\ln(\mathsmaller{\frac{3\abs{\GameTuple}}{\delta}})$ constant factor instead of $\ln(\mathsmaller{\frac{2\abs{\GameTuple}}{\delta}})$.
On the other hand, when $\EVarWimpy_\PlayerIndex(\StratProfile)$ is small, \Cref{thm:eBennett} is much sharper than \Cref{thm:Hoeffding}.
In the extreme, when $\EVarWimpy_\PlayerIndex(\StratProfile)\approx 0$ (i.e., the game is \emph{near-deterministic}), \Cref{thm:eBennett} is an asymptotic improvement over \Cref{thm:Hoeffding} by a $\Theta \left( {\sqrt{\frac{\ln(\vphantom{\scalebox{0.99}{\ensuremath{|}}}\smash{\mathsmaller{\frac{\abs{\GameTuple}}{\delta}}})}{m}}} \right)$ factor, despite a priori unknown variance.
The beauty of this bound is how the sampling cost gracefully adapts to the inherent difficulty of the task at hand.

\begin{restatable}{corollary}{corEBennettComplexity}
\label{cor:eBennettComplexity}
With probability at least $1 - \nicefrac{\delta}{3}$, \Cref{thm:eBennett} guarantees a data complexity of at most $1 + 2\ln\frac{3\abs{\GameTuple}}{\delta}\left(\frac{5\UtilityRange}{2\epsilon} + \frac{\norm{\UtilityVariance}_\infty}{\epsilon^2}\right)$.
\end{restatable}

\citet{areyan2020improved} use a similar empirical
variance bound in their Global Sampling (GS) algorithm,
though theirs is sensitive only to the largest variance over all utilities.
Because the largest variance over all utilities is the data complexity bottleneck, our data complexity results likewise depend on the largest variance.
Our pruning algorithm, however, improves upon the query complexity of their GS algorithm, by pruning low-variance utilities
before those with high variance.
Their approach requires fewer tail bounds,
which may yield a constant factor improvement in data complexity, but its query complexity is asymptotically inferior, except in pathological cases.%
\footnote{If variance is uniform, 
say $v^*$, then $\norm{\UtilityVariance}_{1,\infty}$ is $v^* \abs{\StratProfileSpace}$.
This setup yields the worst-case query complexity across all games with the same data complexity, 
as nothing can be pruned until near the final iteration, with high probability.
}

\section{Progressive Sampling with Pruning}
\label{sec:psp}

\amy{search for index and indices}\bhaskar{I think past the introduction, talking about pruning indices is more clear than pruning utilities. \amy{the intro doesn't talk about pruning utilities. i'm not sure anywhere does.} For the latter, we would have to say we prune utility functions maybe? A utility could be a sample utility, an expected utility, an empirical utility, a utility function. It doesn't seem clear to me that it's referring to the utility function. \amy{it's not referring to the function, and shouldn't be. the function is the game!} But an index clearly specifies the function, and pruning an index clearly prunes a utility function. \amy{a function? you've lost me. maybe i'm just too tired?}}

Having described the requisite tools, we now describe our algorithm for uniformly approximating a simulation-based game.
We call this algorithm Progressive Sampling with Pruning (\PSP{}),
as it does exactly that: it progressively samples \samy{utilities}{\sbhaskar{utility indices}{utilities}} (by querying strategy profiles), pruning those that are well estimated as soon as it determines that it can do so correctly.

As noted above, if we knew the maximum variance $\norm{\UtilityVariance(\StratProfile)}_\infty$ for each strategy profile $\StratProfile\in\StratProfileSpace$, we could derive the requisite number of queries, namely $\lceil m^B(\epsilon, \nicefrac{\delta}{\abs{\GameTuple}}; \UtilityRange, \norm{\UtilityVariance(\StratProfile)}_\infty)\rceil$, via Bennett's inequality (\Cref{thm:Bennett}).
By estimating the variance via empirical variance (\Cref{thm:eBennett}), we aspire to place a similar number of queries for each profile while remaining variance-agnostic.
As the empirical variance is a random variable, whether or not 
a particular number of queries is sufficient is itself random, and the probability of this event depends on the utilities
and the \emph{a priori\/} unknown tail behavior of their empirical variances.
Nevertheless, \PSP{} achieves our goal, matching both the data and query complexity of Theorem~\ref{thm:Bennett} with high probability, up to logarithmic factors.

Next, we sketch our basic algorithm.
We assume a sampling schedule that comprises successive sample sizes, after each of which we apply our empirical Bennett bound to all active (i.e., unpruned) utility indices, pruning an index once its empirical variance and sample size are sufficient to $\epsilon$-estimate it with high probability.
This process repeats until all indices have been well estimated, and are thus pruned, at which point the algorithm terminates and returns an $\epsilon$-approximation.

\cyrus{Did we want this, it's a tricky tangent:
\samy{}{\PSP{} 
continues to query a strategy profile until all its corresponding utility indices are pruned.
Pruning indices, rather than strategy profiles, can potentially reduce simulation costs when only a few players' utilities are required.
For example, in graphical games, where players' utilities only depend on the utilities of their neighbors, it may be sufficient to consider only a small subgraph of a game when querying a strategy profile with many pruned indices.}}

\amy{How many profiles are active at any point then depends on the distribution of the utility variances, because different variances lead to variations in the bounds on which \PSP{} depends.}

Since our algorithm does not assume \emph{a priori\/} knowledge of variance, it requires a 
carefully tailored (novel) schedule that repeatedly ``guesses and checks'' whether the current sample size is sufficient.
As we desire a 
competitive ratio guarantee\amy{discuss with cyrus},
we select geometrically increasing sample sizes, where each successive sample size exceeds the previous one by a constant \emph{geometric base factor} $\beta > 1$.
Here, there is an inherent trade off between the schedule length $T$ and $\beta$.
Large $T$ (small $\beta$) decreases statistical efficiency (due to a union-bound over iterations to correct for multiple comparisons), but large $\beta$ (small $T$) are also inefficient, as it may be that a sample of size $m$ was \emph{nearly\/} sufficient, in which case $\beta m$ overshoots the sufficient sample size by nearly a factor $\beta$.

Additionally, due to the structure of the empirical Bennett inequality, there exists a \emph{minimum sufficient sample size}, which we lower-bound as $\alpha$, before which no pruning is possible, as well as a \emph{maximum necessary sample size} $\omega$, after which we guarantee (via Hoeffding's inequality) that all utilities are $\epsilon$-well estimated w.h.p., so that the algorithm can terminate.
There is thus no benefit to any schedule beginning before $\alpha$ or ending after $\omega$, so our schedule need only span $\{\ceil{\alpha}, \ceil{\alpha}+1, \dots, \ceil{\omega}\}$.

In Appendix~\ref{sec:sampling_schedule}, we derive these extreme sample sizes $\alpha$ and $\omega$ for a given $\beta$ and $T$.
The number of times a sample of size $\alpha$ must expand by a factor $\beta$ to reach $\omega$ (and thus guarantee termination) is $\log_{\beta}\nicefrac{\omega}{\alpha}$, which gives the necessary schedule length.
Conveniently, since $\alpha$ and $\omega$ both depend on $T$, this dependence divides out in the log-ratio.
Thus we may solve for a geometric schedule in closed form using only basic algebra.


\begin{algorithm}
\algrenewcommand\algorithmicindent{1.0em}
\algrenewcommand\algorithmicindent{0.75em}
\begin{algorithmic}[1]
\Procedure{PSP}{$\ConditionalGame{\ConditionSpace}, \ConditionDistribution, \UtilityIndices, \UtilityRange, \delta, \epsilon, \beta$} $\to \tilde{\bm{\Utility}}$
 
\State \Input 
Conditional game $\ConditionalGame{\ConditionSpace}$,
condition distribution $\ConditionDistribution$,
index set $\UtilityIndices$,
utility range $\UtilityRange$,
failure probability $\delta \in (0, 1)$,
target error $\epsilon > 0$,
geometric ratio $\beta > 1$

\State \Output 
Empirical utilities $\tilde{\Utility}$,
for all indices $(\PlayerIndex, \StratProfile) \in \SetOfPlayers \times \StratProfileSpace$

\State $(U_{\PlayerIndex} (\StratProfile), V_{\PlayerIndex} (\StratProfile)) \gets (0, 0), \forall (\PlayerIndex, \StratProfile) \in \UtilityIndices$
\Comment{Initialize 
utility sum and sum of squares}\label{alg:psp:init-outputs}

\State $\alpha \gets \frac{2 \UtilityRange}{3 \varepsilon} \ln \frac{3 \abs{\GameTuple} \ScheduleLength}{\delta}$;
$\ScheduleLength \gets \log_{\beta}\frac{3 \UtilityRange}{4 \varepsilon}$; $m \gets 0$
\Comment{Initialize sampling schedule}

\For{$\TimeIndex \in 1, \dots, \ScheduleLength$} \Comment{Progressive sampling iterations}

    \State $\NumberOfSamples' \gets \ceil{\alpha\beta^{\TimeIndex}} - m$ \Comment{Geometric schedule sample size (Marginal)}

    \State $\NumberOfSamples \gets \ceil{\alpha\beta^{\TimeIndex}}$ \Comment{Geometric schedule sample size (Cumulative)}

    \State $\SamplePoint_{1}, \ldots, \SamplePoint_{m'} \sim \ConditionDistribution^{\NumberOfSamples'}$ \Comment{Draw $\NumberOfSamples'$ samples from condition distribution $\ConditionDistribution$}
    
    \For{$\StratProfile\in\StratProfileSpace$}
    \If{$(\PlayerIndex, \StratProfile)\in \UtilityIndices$ for some $\PlayerIndex\in\SetOfPlayers$}\Comment{If a profile $\StratProfile$ has an active index (i.e., is unpruned)}
    \State \textbf{query} $\Utility (\StratProfile; \SamplePoint_i) \textup{ for all } i \in \{1, \dots, m'\}$\Comment{Query profile $\StratProfile$ at all samples}
    \EndIf
    \EndFor
    
    \For{$(\PlayerIndex, \StratProfile) \in \UtilityIndices$} \Comment{Improve estimates at active (i.e., unpruned) indices}

    \State $\displaystyle (U_{\PlayerIndex} (\StratProfile), V_{\PlayerIndex} (\StratProfile)) \gets {(U_{\PlayerIndex} (\StratProfile), V_{\PlayerIndex} (\StratProfile)) + \left( 
    \sum_{i=1}^{\NumberOfSamples'} \Utility_{\PlayerIndex}(\StratProfile; \SamplePoint_{i}), \sum_{i=1}^{\NumberOfSamples'} \Utility^{2}_{\PlayerIndex}(\StratProfile; \SamplePoint_{i}) \right)}$
    \Comment{Update statistics}\label{alg:psp:interval_unpruned-gs}
     
    \vspace{-0.25cm}

    \State $\hat{\UtilityVariance}_{\PlayerIndex}(\StratProfile) \gets \frac{V_{\PlayerIndex} (\StratProfile) - \left( \frac{U_{\PlayerIndex} (\StratProfile)^2}{m} \right)}{m-1}$
    \Comment{Empirical variance}
\label{step:empiricalVariance}

    \State $\tilde{\UtilityVariance}_{\PlayerIndex}(\StratProfile) \gets  \mathsmaller{\hat{\UtilityVariance}_{\PlayerIndex}(\StratProfile) +\frac{2 \UtilityRange^{2} \ln \left( \frac{3 \abs{\GameTuple}\ScheduleLength}{\delta} \right)}{3 \NumberOfSamples} + \sqrt{\Bigl(\frac{1}{3} + \frac{1}{2 \ln \left( \frac{3\abs{\GameTuple}\ScheduleLength}{\delta} \right)}\Bigr) \Bigl( \frac{\UtilityRange^{2} \ln \left( \frac{3\abs{\GameTuple}\ScheduleLength}{\delta} \right)}{\NumberOfSamples - 1}\Bigr)^{\smash{2}} + \frac{2 \UtilityRange^{2} \EUtilityVariance_\PlayerIndex(\StratProfile) \ln \left( \frac{3\abs{\GameTuple}\ScheduleLength}{\delta} \right)}{\NumberOfSamples}}\!}$
    \Comment{Variance upper bound}\label{step:varianceUpperBound}
    
    \State $\bm{\epsilon}_{\PlayerIndex} (\StratProfile) \gets \min\! \left\{ c\sqrt{\frac{\ln({\frac{3\abs{\GameTuple} \ScheduleLength}{\delta}})}{2\NumberOfSamples}}, \ \frac{\UtilityRange \ln \left( \frac{3\abs{\GameTuple}\ScheduleLength}{\delta}\right)}{3\NumberOfSamples} + \sqrt{\frac{2\tilde{\UtilityVariance}_{\PlayerIndex}(\StratProfile) \ln \left({\frac{3\abs{\GameTuple}\ScheduleLength}{\delta}}\right)}{\NumberOfSamples}} \, \right\}$
    \Comment{Hoeffding + empirical Bennett}\label{step:empiricalBennett}
    
    \State $\tilde{\bm{\Utility}}_{\PlayerIndex} (\StratProfile) \gets \frac{U_\PlayerIndex (\StratProfile)}{\NumberOfSamples}$ 
    \Comment{Store empirical means}
    
    \EndFor
    
    \State $\UtilityIndices \gets \smash{\{ (\PlayerIndex, \StratProfile) \in \UtilityIndices \mid \bm{\epsilon}_{\PlayerIndex} (\StratProfile) > \epsilon \}}$
\label{alg:psp:pruning_well-estimated indices}
    \Comment{Prune well-estimated indices} 

    \If{$\UtilityIndices = \emptyset$}
\label{alg:psp:termination}
    \Comment{Termination condition}

    	\State $\Return~\smash{\tilde{\bm{\Utility}}}$ \Comment{Return estimated utilities\bhaskar{Shouldn't this be $\hat{\bm{\Utility}}$ everywhere}\amy{yes, i suppose so. i wonder why it isn't? there was a reason.}}\label{alg:psp:return}
    \EndIf

\EndFor

\EndProcedure
\end{algorithmic}

\caption{Progressive Sampling with Pruning}
\label{alg:psp}
\end{algorithm}


\paragraph{Correctness and Efficiency}

We use our sampling schedule to prove the correctness and efficiency of our algorithm: 
i.e., that it learns an $\epsilon$-approximation of a simulation-based game with high probability, with finite data and query complexity.

\begin{restatable}[Correctness]{theorem}{thmCorrectness}
If \textsc{PSP}$(\ConditionalGame{\ConditionSpace}, \ConditionDistribution, \UtilityRange, \delta, \epsilon, \beta)$ outputs $\tilde{\bm{\Utility}}$ 
then with probability at least $1 - \delta$, it holds that
$\norm{\Utility_{\PlayerIndex}(\StratProfile; \ConditionDistribution) - \tilde{\Utility}_{\PlayerIndex}(\StratProfile)}_{\infty} 
\leq \epsilon$.
\label{thm:correctness}
\end{restatable}

The next lemma bounds the number of queries placed at a strategy profile $\StratProfile\in\StratProfileSpace$.
As desired, we match the number of queries required by Bennett's inequality up to logarithmic factors.

\begin{lemma}
\label{lemma:efficiency}
With probability at least $1 - \nicefrac{\delta}{3}$, \textsc{PSP}$(\ConditionalGame{\ConditionSpace}, \ConditionDistribution, \UtilityRange, \delta, \epsilon, \beta)$ queries strategy profile 
$\StratProfile \in \StratProfileSpace$, at most $m^{\textup{PSP}} \left( \epsilon, \nicefrac{\delta}{T\abs{\GameTuple}}; \UtilityRange, \norm{\UtilityVariance (\StratProfile)}_\infty; \beta \right)$ times, defined as
\begin{align*}
&
1 + \beta\ln\frac{3T\abs{\GameTuple}}{\delta}\left(\frac{\kappa_\delta\UtilityRange}{\epsilon} + \frac{\norm{\UtilityVariance(\StratProfile)}_\infty}{\epsilon^2} + \sqrt{2\cdot\frac{\kappa_\delta\UtilityRange}{\epsilon}\cdot\frac{\norm{\UtilityVariance(\StratProfile)}_\infty}{\epsilon^2} + \left(\frac{\norm{\UtilityVariance(\StratProfile)}_\infty}{\epsilon^2}\right)^2}\right)\\
&
\leq 1 + 2\beta\ln \frac{3T\abs{\GameTuple}}{\delta}\left( \frac{5\UtilityRange}{2\varepsilon} + \frac{\norm{\UtilityVariance(\StratProfile)}_{\infty}}{\varepsilon^{2}} \right)
\\
&
\in \mathcal{O}\left(\log\left(\frac{\abs{\GameTuple}\log(\nicefrac{\UtilityRange}{\epsilon})}{\delta}\right)\left( \frac{\UtilityRange}{\epsilon} + \frac{\norm{\UtilityVariance(\StratProfile)}_{\infty}}{\epsilon^{2}} \right)\right)
\enspace .
\end{align*}
\end{lemma}

Next, we use \Cref{lemma:efficiency} to bound the data and query complexities of our algorithm,
interpolating between the best and worst case.
The best case is deterministic (i.e., $\norm{\UtilityVariance}_{\infty} = 0$); as the variance terms are zero, the bounds are asymptotically proportional to $\nicefrac{1}{\epsilon}$ instead of $\nicefrac{1}{\epsilon^2}$.
In the worst case, when $\norm{\UtilityVariance}_{\infty} = \nicefrac{\UtilityRange^2}{4}$ and $\norm{\UtilityVariance}_{1, \infty} = \abs{\StratProfileSpace}\norm{\UtilityVariance}_{\infty}$, the schedule is exhausted, 
and we recover Theorem~\ref{thm:Hoeffding}, but we pay a factor $\log \log \left( \nicefrac{c}{\epsilon} \right)$ 
for \samy{having}{sampling progressively to have} potentially stopped early.
In sum, even without assuming \emph{a priori\/} knowledge of the variances of the utilities, we 
match the efficiency bounds of Theorem~\ref{thm:Bennett} up to logarithmic factors.


\begin{restatable}[Efficiency]{theorem}{thmEfficiency}
\label{thm:efficiency}
Let $T = \ceil{ \log_{\beta} \frac{3 \UtilityRange}{4 \varepsilon} }$ be the schedule length and let $\kappa_\delta\doteq \frac{4}{3} + \sqrt{1 + \frac{1}{2\ln\left(\frac{3T\abs{\GameTuple}}{\delta}\right)}}$.
With probability at least $1 - \nicefrac{\delta}{3}$, the data complexity is at most
\begin{align*}
m^{\textup{PSP}}\left(\epsilon, \nicefrac{\delta}{T\abs{\GameTuple}}; \UtilityRange, \norm{\UtilityVariance}_\infty; \beta\right)
&
\leq 1 + 2\beta\ln \frac{3T\abs{\GameTuple}}{\delta}\left( \frac{5\UtilityRange}{2\varepsilon} + \frac{\norm{\UtilityVariance}_{\infty}}{\varepsilon^{2}} \right)
\\
&
\in \mathcal{O}\left(\log\left(\frac{\abs{\GameTuple}\log(\nicefrac{\UtilityRange}{\epsilon})}{\delta}\right)\left( \frac{\UtilityRange}{\epsilon} + \frac{\norm{\UtilityVariance}_{\infty}}{\epsilon^{2}} \right)\right)
\enspace ,
\end{align*}
and the query complexity is at most
\begin{align*}
2\abs{\StratProfileSpace} + m^{\textup{PSP}}\left(\epsilon, \nicefrac{\delta}{T\abs{\GameTuple}}; \UtilityRange\abs{\StratProfileSpace}, \norm{\UtilityVariance}_{1,\infty}; \beta\right)
&
\leq 1 + 2\abs{\StratProfileSpace} + 2\beta\ln \frac{3T\abs{\GameTuple}}{\delta}\left( \frac{5\UtilityRange \abs{\StratProfileSpace}}{2\varepsilon} + \frac{\norm{\UtilityVariance}_{1,\infty}}{\varepsilon^{2}} \right) 
\\
&
\in \mathcal{O}\!\left(\!\log\!\left(\frac{\abs{\GameTuple}\log(\nicefrac{\UtilityRange}{\epsilon})}{\delta}\!\right)\!\left(\!\frac{\UtilityRange\abs{\StratProfileSpace}}{\epsilon} + \frac{\norm{\UtilityVariance}_{1,\infty}}{\epsilon^{2}}\!\right)\!\right)
\enspace .
\end{align*}
%
\end{restatable}

Note that from standard lower bounds for mean estimation \citep{devroye2016sub}, the data complexity of \emph{any\/} estimator must be such that
$m \geq \frac{\norm{\UtilityVariance}_{\infty}}{\varepsilon^{2}}
  \ln \frac{\abs{\GameTuple}}{\delta}$;
similarly, query complexity must be such that
$m \geq  \frac{\norm{\UtilityVariance}_{1, \infty}}{\varepsilon^{2}} \ln \frac{\abs{\GameTuple}}{\delta}$.
Such lower-bounds hold \emph{even when variance is known}.
Although mean-estimation lower bounds apply, the union bound is tight for appropriately negatively-dependent random variables, so there exist conditional game structures for which the above complexities are necessary.
%
%
Other than constant factor terms, log log terms relating to schedule length, and fast decaying $\frac{\UtilityRange}{\varepsilon}$ terms, these lower bounds match our results.
It is thus straightforward to divide our sample complexity bounds by these lower bounds to obtain a competitive ratio over any other learning algorithm for simulation-based games that produces a uniform approximation.

\amy{Before deploying our algorithm in practice, it would be worthwhile to optimize $\beta$ for the game of interest.}
\cyrus{Need to minimize $\beta\ln \frac{3T\abs{\GameTuple}}{\delta}$}


\section{Experiments with Algorithms}
\label{sec:expts_algos}

In this section, we explore
the behavior of \PSP{} on a suite of games generated using GAMUT~\cite{nudelman2004run}, a state-of-the-art game generator.
We enhance these games with various noise structures, so that they suitably mimic simulation-based games.
Our findings can be summarized as follows:
\begin{itemize}

\item We show that the number of queries that can be saved by pruning is highly sensitive to the utility variance distribution.

\item We demonstrate that in practice, \PSP{} consistently outperforms theoretical upper bounds, and achieves significantly lower query complexities than \GS.

\end{itemize}

\subsection{Experimental Setup}
\label{subsec:experimentalSetup}

We focus on congestion and zero-sum games in these experiments, enhancing these games with noise structures so that they suitably mimic simulation-based games.


\paragraph{Finite Congestion Games.} 
Congestion games are a class of games known to exhibit pure strategy Nash equilibria~\citep{rosenthal1973class}.
A tuple $\smash{\CongestionGame = (\SetOfPlayers, \SetOfFacilities, \{ \FacilityCostFunction_\FacilityIndex \mid \FacilityIndex \in \SetOfFacilities \})}$ is a \mydef{congestion game}~\citep{christodoulou2005price}, where $\smash{\SetOfPlayers = \{1, \ldots, \NumberOfPlayers\}}$ is a set of agents and $\smash{\SetOfFacilities = \{1, \ldots, \NumberOfFacilities\}}$ is a set of facilities.
A strategy $\smash{\Strategy_\PlayerIndex} \in 2^\SetOfFacilities\setminus\emptyset$ for player $\PlayerIndex$ is a non-empty set of facilities, and $\smash{\FacilityCostFunction_\FacilityIndex}$ is a (universal: i.e., non-agent specific) cost function associated with facility $\FacilityIndex$. 

In a finite congestion game, among all their available strategies, each agent prefers one that minimizes their cost.
Given strategy profile $\StratProfile \in \StratProfileSpace$, agent $\PlayerIndex$'s cost is defined as $\smash{\PlayerCost_\PlayerIndex(\StratProfile) = \sum_{\FacilityIndex\in\Strategy_\PlayerIndex} \FacilityCostFunction_\FacilityIndex(n_\FacilityIndex(\StratProfile))}$, where $n_\FacilityIndex(\StratProfile)$ is the number of agents who select facility $\FacilityIndex$ in $\StratProfile$.
We denote GAMUT's random distribution over congestion games for fixed $\NumberOfPlayers$ and $\NumberOfFacilities$, with utilities bounded in the interval $(-\nicefrac{u_0}{2}, \nicefrac{u_0}{2})$, by $\smash{\RandomCongestionGame(\NumberOfPlayers, \NumberOfFacilities, u_0)}$.


\paragraph{Random Zero Sum.} 
A zero-sum game has $2$ agents, each with $k \in \mathbb{Z}_+$ pure strategies, with the agent's utilities defined as the negation of one another's: i.e., $\Utility_2 (\StratProfile) = -\Utility_1 (\StratProfile)$, for all $\StratProfile \in \StratProfileSpace$.
When generating such games at random, we assume a uniform distribution over utilities, with the first agent's utilities drawn i.i.d.\ as $\smash{\Utility_1 (\StratProfile) \sim U(\nicefrac{-u_0}{2}, \nicefrac{u_0}{2})}$, for some $u_0 \in \mathbb{R}$.
We denote GAMUT's distribution over Random Zero Sum games generated in this way by $\RandomZeroSum(k, u_0)$.

\paragraph{Noise distributions.}
Since the only distributional quantities our bounds depend on are scale 
and variance,
it suffices to study games with additive noise and non-uniform variance.
We use i.i.d.\ additive noise variables $\mathcal{N}_{\PlayerIndex, \StratProfile}$ for each index $(\PlayerIndex, \StratProfile)\in \SetOfPlayers\times \StratProfileSpace$, each with range $[-\nicefrac{d}{2}, \nicefrac{d}{2}]$, 
expected value 0, and fixed variance $\UtilityVariance_{\textup{max}}$.
We then control the variance at each individual index by scaling the corresponding noise variable.
In particular, utilities generated by the simulator have the form $\Utility_\PlayerIndex(\StratProfile; y)\sim \Utility_\PlayerIndex(\StratProfile) + \gamma_{\PlayerIndex, \StratProfile}\mathcal{N}_{\PlayerIndex, \StratProfile}$, where each scale $\gamma_{\PlayerIndex, \StratProfile}$ is sampled from a Beta($\alpha$, $\beta$) distribution over $[0, 1]$, for some $\alpha, \beta > 0$ upon generation of the game.
Using a Beta distribution,
we can model many different variance structures.
In our experiments, the noise variable 
is always a scaled and shifted Bernoulli random variable, generating either $-\nicefrac{d}{2}$ or $\nicefrac{d}{2}$ with equal probability.
It then has variance $\UtilityVariance_\textup{max} = \nicefrac{d^2}{4}$, enabling finer control over the variances of individual utilities.

\subsection{Effectiveness of \PSP{} w.r.t. Utility Variance}

\begin{figure}[!t]
    \centering
    \includegraphics[width=\columnwidth]
    {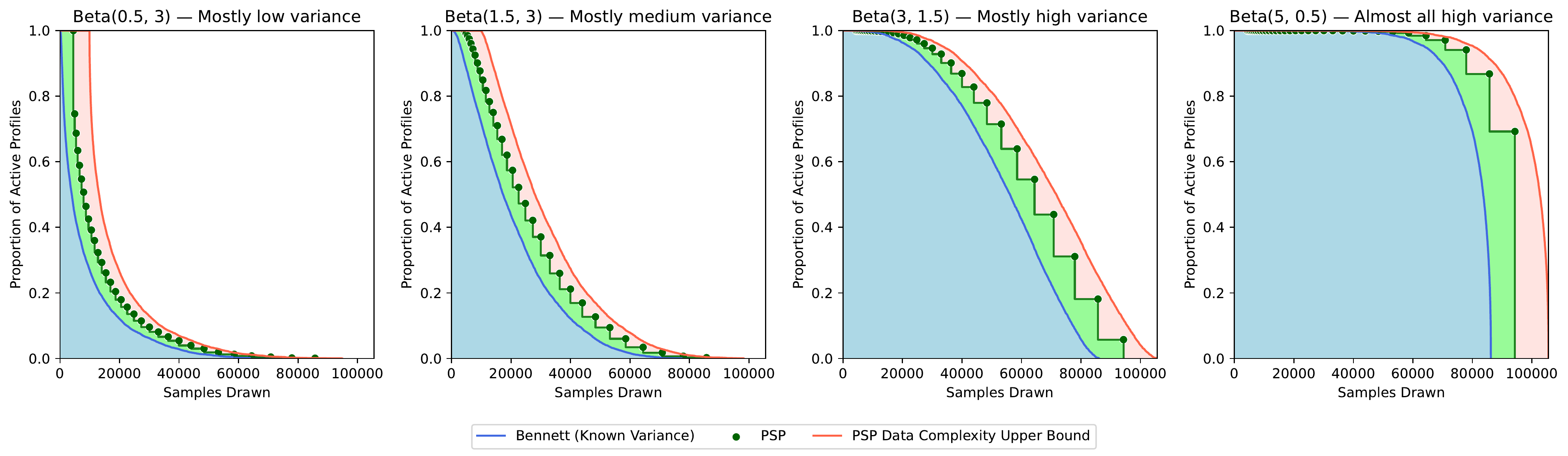}
    \caption{Proportion of pure strategy profiles that are active over the course of the \PSP{} algorithm for four games with different distributions of individual utility variances.\cyrus{Tiny figure text}}
    \label{fig:proportion_pruned}
\end{figure}


\PSP{}
queries strategy profiles until they all attain a target error guarantee.
Along the way, empirical estimates of variance are used to prune profiles roughly
in order of increasing maximum variance $\norm{\UtilityVariance(\StratProfile)}_\infty$.
How many profiles are active at any point thus depends on the distribution of the utility variances, because different variances lead to variations in the bounds on which \PSP{} depends.

This next set of experiments is designed to highlight how the amount of work done (i.e., the number of queries placed) by \PSP{} varies with the distribution of utility variances.
As expected, we find that the less (resp.\ more) the utility variances are concentrated near the maximum variance, the smaller (resp.\ larger) the number of queries \PSP{} requires to reach its target error guarantee.

In these experiments, we sample four games from $\RandomZeroSum(80, 2)$, where for each game, we use scaled and shifted Bernoulli additive noise variables with $d=20$ and $\gamma_{\PlayerIndex, \StratProfile} \sim \textup{Beta}(\alpha, \beta)$ for a certain $\alpha, \beta > 0$.
We then run \PSP{} on each game with utility range $c\doteq 22$, failure probability $\delta \doteq 0.05$, error threshold $\epsilon \doteq 0.2$, and geometric ratio $\beta\doteq 1.1$.
For each iteration of the algorithm, we plot the proportion of strategy profiles that were active during that iteration,
versus the number of samples that were drawn thus far.
\amy{i'm missing how/why these points are ``drawn''? they are the product of mkt research. is that research then input to the simulator.}
A point $(m, p)$ can thus be read as, ``during this run of \PSP{}, a proportion $p$ of strategy profiles were active, and hence queried, \emph{at each of the first $m$ drawn samples}.''
Since \PSP{} only prunes profiles at the end of each iteration, we connect these points with a decreasing step function.

\amy{please proofread}
We also report lower and upper (PAC-style) bounds on this $\PSP$ step function.
The number of samples for which a strategy profile $\StratProfile$ is active is lower bounded by $\smash{m^{B}(\epsilon, \frac{\delta}{\abs{\GameTuple}}; \UtilityRange, \norm{\UtilityVariance(\StratProfile)}_\infty)}$  (\Cref{thm:Bennett}) and upper bounded (with high probability) by $\smash{m^{\textup{PSP}}(\epsilon, \frac{\delta}{T\abs{\GameTuple}}; \UtilityRange, \norm{\UtilityVariance(\StratProfile)}_\infty; \beta)}$  (\Cref{lemma:efficiency}), respectively.
Since both of these these bounds are strictly increasing in $\norm{\UtilityVariance(\StratProfile)}_\infty$, they also bound the number of samples for which all strategy profiles whose maximum variance is greater than or equal to $\norm{\UtilityVariance(\StratProfile)}_\infty$ are active.
We report the proportion of active profiles corresponding to these lower and upper bounds, for all strategy profiles $\StratProfile\in\StratProfileSpace$.

\Cref{fig:proportion_pruned} depicts clear differences in how quickly strategy profiles are pruned, depending on $\alpha$ and $\beta$, which dictate how utility variances are distributed.
Since for each sample drawn, \PSP{} queries every active strategy profile, the total area under the \PSP{} curve is directly proportional to the total number of queries that were required to reach the target error guarantee.
By similar reasoning, the area under the lower and upper bound curves, respectively, are also proportional to lower and (high probability) upper bounds on this total number of queries.
We can therefore infer that in games where the utility variances are concentrated at a much smaller value than $\UtilityVariance_{\textup{max}}$ (modeled by $\textup{Beta}(0.5, 3)$), query complexity is dramatically reduced through pruning.
On the other hand, in games where the utility variances are concentrated closer to $\UtilityVariance_{\textup{max}}$ (modeled by $\textup{Beta}(5, 0.5)$), pruning has a much smaller effect on query complexity.

\subsection{Performance of \PSP{} vs. \GS{}}

Next, we compare the performance of the \PSP{} and \GS{} algorithms.
Such a comparison is a difficult due to the differing use cases for the two algorithms.
We run \PSP{} when we have a target error $\epsilon$ in mind, and the algorithm optimizes the query complexity to reach that target.
We run \GS{} when we are given a data set consisting of $m$ samples together with responses to the queries corresponding to each sample, and the algorithm outputs the tightest possible guarantee on the confidence radius.
Since \PSP{} has been our focus in this paperProportion, in our plots, we use the horizontal axis for the independent variable $\epsilon$, and the vertical axis for the resulting query and data complexities.
When plotting results from the \GS{} algorithm, however, we plot as if the vertical axis were the independent variable and the horizontal axis, dependent.
It is important that our results are interpreted this way, as \GS{} is not able to provide any strong \emph{a priori\/} guarantees regarding the query complexity required to produce a certain $\epsilon$ guarantee without access to the true maximum utility variance of a game.

\amy{wait! for Enrique results got much much better when we scaled up. what about for us?} \bhaskar{I think we generally get similar results. The comparisons stay the same. I can check again when I have time to be sure.}
\amy{well, he was using regret pruning, not well-estimated, so you should try to replicate that result for the AAAI paper.}

We experiment with congestion games sampled from $\RandomCongestionGame(3, 3, 2)$,%
\footnote{We use a relatively small game (three players, 343 strategies each) to show our \PSP{} algorithm is competitive even in games with few profiles. We see similar results for larger games.}
and Bernoulli additive noise variables with $d=20$ in two settings: a Single Game and Many Games.
In the Single Game setting, we sample one congestion game and draw each $\gamma_{\PlayerIndex, \StratProfile} \sim \textup{Beta}(1.5, 3)$.
We then run $\GS$ using uniform empirical Bennett bounds 10 times at a variety of sample sizes $m$, and $\PSP$ 10 times for a variety of target error $\epsilon$. 
Using $\textup{Beta}(1.5, 3)$ models a plausible empirical game where most utilities have moderate variance, but some outlier utilities have high variance.
In the Many Games setting, we instead sample 9 congestion games and draw each $\gamma_{\PlayerIndex, \StratProfile} \sim \textup{Beta}(1.5 + i, 3 + j)$, using each $(i, j) \in \{-0.1, 0, 0.1\}^2$ in one game.
For each of these 9 games, we then run $\GS$ once for a variety of sample sizes $m$, and $\PSP$ once for a variety of target errors $\epsilon$.

On each graph in \Cref{fig:eps_vs_complexity}, we plot the following:
\begin{itemize}
    \item GS-H: \GS{} using Hoeffding Bounds (\Cref{thm:Hoeffding}).
    
    \item GS-B: \GS{} using uniform Bennett bounds (\Cref{thm:Bennett}) and known maximum variance.
    
    \item GS-EB: \GS{} using uniform empirical Bennett bounds, 
    with dependence on $\norm{\EUtilityVariance}_\infty$ rather than individual empirical variances~\cite{areyan2020improved}.
    
    \item GS-EB Upper Bound: PAC-style upper bound on the empirical Bennett data complexity (\Cref{thm:eBennett}).
    
    \item PSP: \PSP{} using Hoeffding and Empirical Bennett bounds (\Cref{alg:psp}).
    
    \item PSP Upper Bound: PAC-style upper bound on the query complexity of \PSP{} (\Cref{thm:efficiency}).
\end{itemize}

\noindent
For the algorithms that compute an empirical Bennett bound (i.e., GS-EB and PSP), we plot each run of the algorithm in the single game case, and a min-max interval in the many games case.
    
\begin{figure}[!t]
    \centering
    \includegraphics[width=\columnwidth]
    {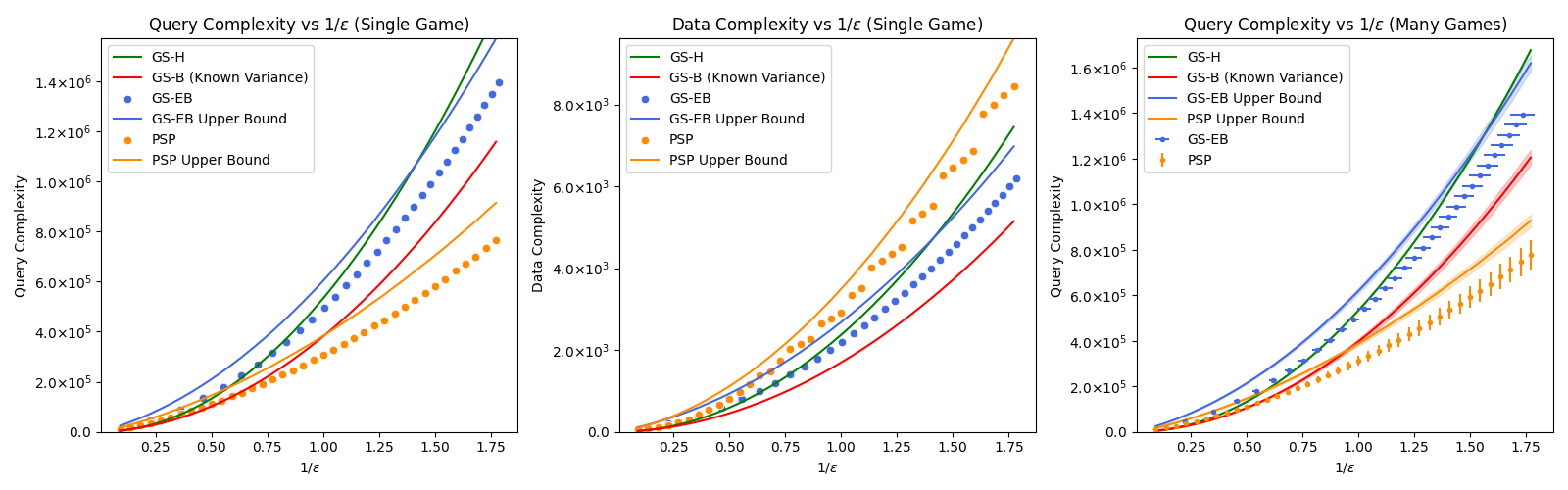}
    \caption{Performance comparisons of \GS{} using Hoeffding and Uniform Empirical Bennett vs.{} \PSP{}, along with their corresponding bounds, for various sample sizes $m$ and target errors $\epsilon$, respectively.}
    \label{fig:eps_vs_complexity}
\end{figure}

In the Single Game graphs, although we plot 10 runs of both algorithms, all that is visible is a single run's worth of points.
This suggests that for a single game, from one run to another, there is little variation in the $\epsilon$ guarantees produced by $\GS$ or in the query complexities of $\PSP$.
On the other hand, in the Many Games graph, we see much more variation, again showing that changes in the distribution of utility variances impact these algorithms.

We observe that \PSP{} significantly outperforms \GS-EB in query complexity, even consuming fewer queries than \GS-B, which assumes known maximum variance, while performing significantly worse with respect to data complexity, even when compared to \GS-H.
The same relations hold when comparing \PSP{} Upper Bound to $\GS$-EB Upper Bound.
We thus conclude that \PSP{} outperforms \GS{} using uniform empirical Bennett bounds,
thus empirically outperforming the state-of-the-art.

\bhaskar{I'm having the thought that if we had one more plot above for the $\textup{Beta}(5, 0.5)$ setting (variances concentrated around maximum variance), then we would have a clear edge case where even though PSP is able to prune a bit, it won't do better than GS.}
\amy{great idea!}


In summary, \PSP{} is the preferred algorithm when minimizing query complexity is paramount; it was designed to be thus, and our experiments confirm that it achieves this desideratum.
An exception is the pathological case where all utility variances are uniform, in which case \PSP{} results in both greater query complexity and greater data complexity.
We note anecdotally, however, that 
even in games with variances concentrated near the maximum variance ($\gamma_{\PlayerIndex, \StratProfile} \sim \textup{Beta}(3, 1.5)$), \PSP{} still matches \GS{} in terms of query complexity; it is only ever marginally better or worse.

\section{Experiments with Properties}
\label{sec:expts_props}

Upon termination, \PSP{} returns an empirical estimate of a simulation-based game with utilities which, with probability at least $1 - \delta$, are no more than some target error $\epsilon$ away from the true utilities.
It follows that 
all well-behaved (i.e., Lipschitz continuous in utilities) properties of this game can likewise be learned.
In this section, we further investigate through experimentation the learnability of two key properties of games, namely regret (i.e., equilibria) and power-mean welfare.
Our findings can be summarized as follows:
\begin{itemize}
\item Although regret is 2-Lipschitz, and thus amenable to statistical EGTA in theory, we exhibit games in which learning better approximations of the game steadily grants access to better $\epsilon$-equilibria, as well as games where this is often not the case.

\item We demonstrate that $\rho$-power-mean welfare is not well behaved for certain values of $\rho$, as the average supremum error between the empirical estimates and the true of $\rho$-power-mean welfare can grow very large.
\end{itemize}

\subsection{Regret}
\label{subsec:regret}

\begin{figure}
    \centering
    \includegraphics[width=\columnwidth]{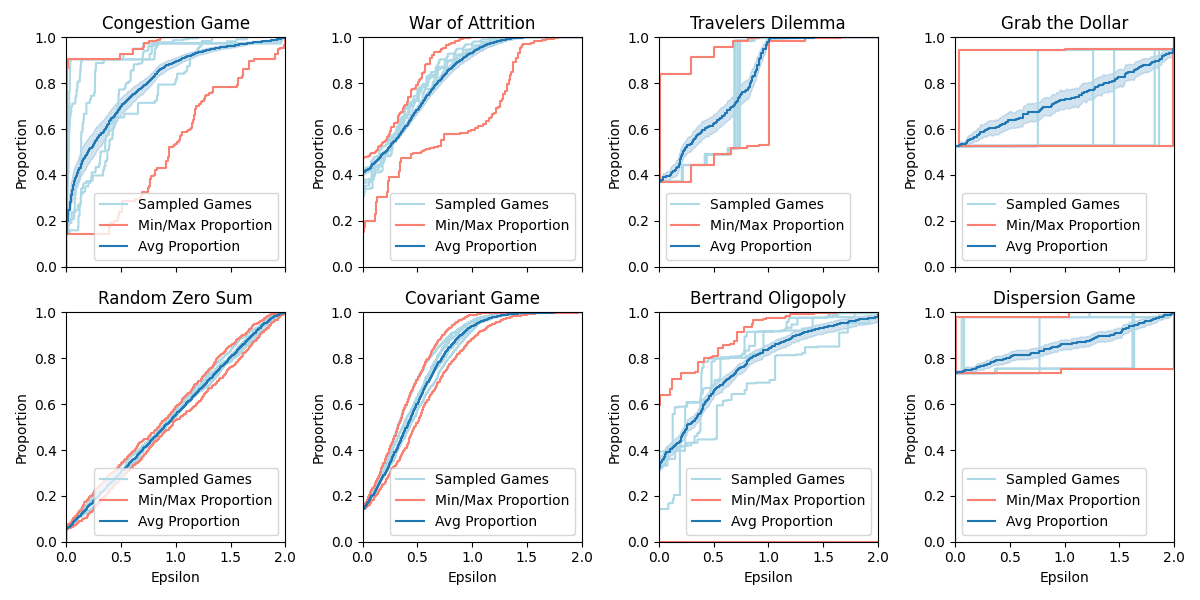}
    \caption{Average proportion of pure strategy profiles that are $\epsilon$-Nash equilibria over 100 sampled games vs.\ $\epsilon$ for 8 different GAMUT games. Exact proportions for 5 of the 100 sampled games are also plotted, along with the maximum and minimum proportions over the 100 games.}
    \label{fig:eps_nash_games}
\end{figure}



We ran these experiments on randomly generated game instances from eight GAMUT game classes.
For each game class, we generate 100 games with all parameters randomized except those needed to fix the game size between 300 and 350 pure strategy profiles (e.g., for Random Zero Sum games, we fix the number of strategies per player at 18, which yields $18^2=324$ pure strategy profiles).

In \Cref{fig:eps_nash_games}, we plot the proportion of pure%
\footnote{The fact that we consider only pure strategy profiles is a limitation of this analysis.}
strategy profiles that are $\epsilon$-Nash equilibria for various values of $\epsilon > 0$. 
A spurious equilibrium is an $\epsilon$-equilibrium 
that is not also a Nash equilibrium.
The figure shows how the rate at which better approximations of the game (smaller $\epsilon$) eliminate spurious equilibria can depend on the game's inherent structure.
In Random Zero Sum games, for example, spurious equilibria can be steadily removed as stronger $\epsilon$ guarantees are attained.

On the other hand, in Grab the Dollar and Dispersion, we see phase transitions.
In all instances of these games, there is a cutoff $\epsilon^*$ s.t.\ for all $\epsilon < \epsilon^*$, each $\epsilon$-Nash equilibrium is also a Nash equilibrium (i.e., there are no spurious $\epsilon$-Nash equilibria), and for all $\epsilon \geq \epsilon^*$, each $\epsilon$-Nash equilibrium is also an $\epsilon^*$-Nash equilibrium (i.e., all spurious Nash equilibria are $\epsilon^*$-Nash equilibria).
When $\epsilon^*$ is large, the first case renders it trivial (i.e., very few samples are required) to produce an $\epsilon$-uniform approximation, for $\epsilon < \epsilon^*$ and near $\epsilon^*$, and hence learn all Nash equilibria.
On the other hand, when $\epsilon^*$ is very small, the second case renders even spurious equilibria good approximate Nash equilibria, which again can be learned from just a few samples.
As a result, statistical EGTA methodology may be overkill in these games in some cases.
Still, when $\epsilon^*$ is not small enough to capture good approximate equilibria, and simultaneously not large enough to render learning
trivial, our EGTA methods provide a means of efficiently learning Nash equilibria,
as they are designed to gracefully adapt to the inherent, yet \emph{a priori\/} unknown, difficulty of the task at hand.

\subsection{Non-Lipschitz Properties}

\begin{figure}
    \centering
    \includegraphics[width=0.7\columnwidth]
    {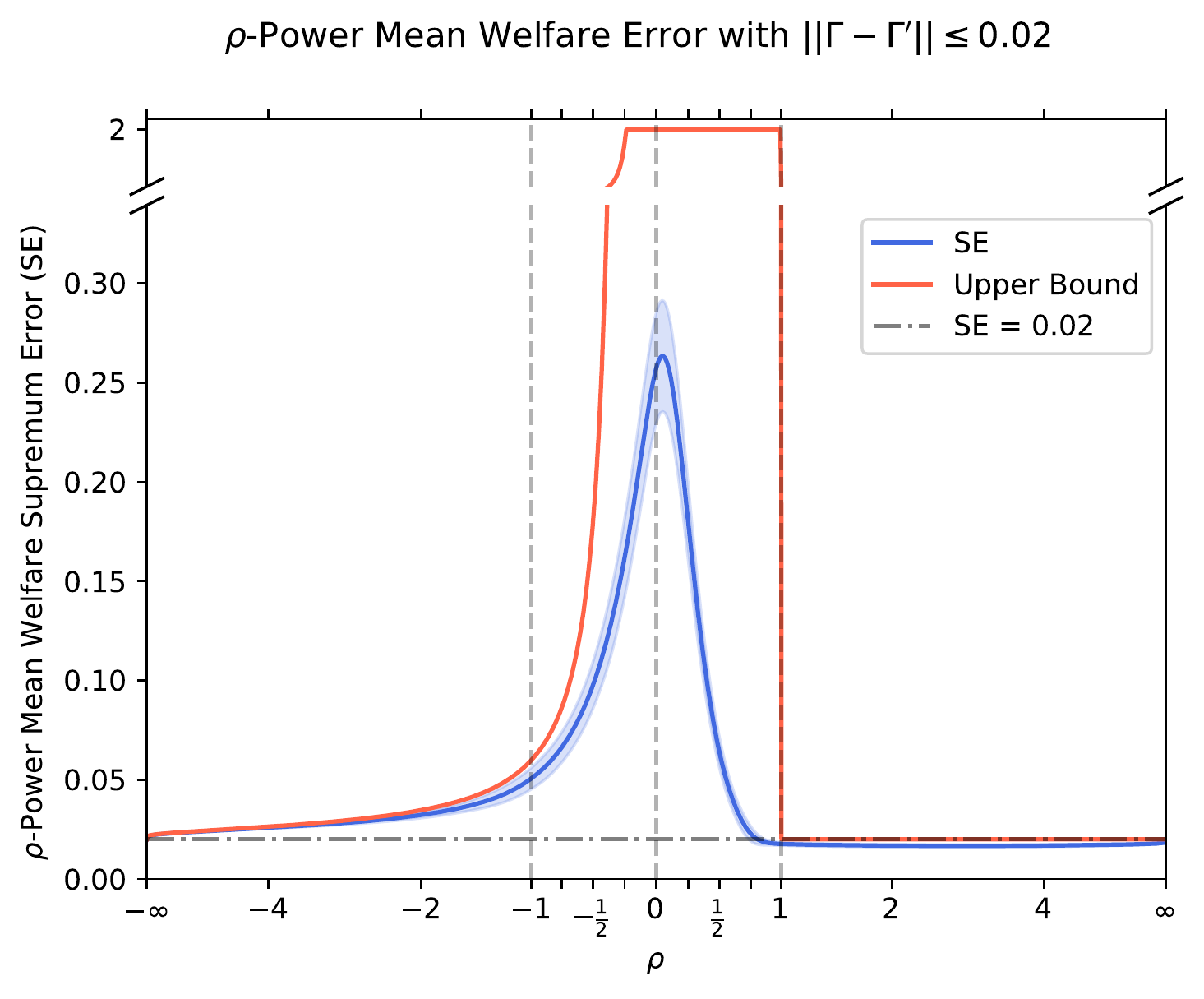}
    \caption{Average supremum error between the empirical estimates and the true $\rho$-power-mean welfares at each strategy profile for various values of $\rho$. Vertical lines are plotted at each $\rho \in \{-1, 0, 1\}$.
    \amy{change colors in plot, especially SE=0.02. and add spaces to legend: SE = 0.02}
    \cyrus{Side caption and/or wider fig?}
    }
    \label{fig:power_mean}
\end{figure}

In this experiment, we sample a congestion game $\GameTuple$ from $\RandomCongestionGame(3, 3, 2)$, and then uniformly draw 100 arbitrary normal form games $\GameTuple'$, each satisfying $\norm{\GameTuple - \GameTuple'}_\infty\leq \epsilon \doteq 0.02$.
Each draw represents a possible empirical estimate of $\GameTuple$ that might be learned by \PSP{} or \GS.
For values of $\rho \in [-10, 10]$, we then plot the average of the supremum, across all strategy profiles $\StratProfile \in \StratProfileSpace$, of the errors between the empirical estimate of the $\rho$-power-mean welfare at $\StratProfile$ and the (actual) $\rho$-power-mean welfare at $\StratProfile$.
In \Cref{fig:power_mean}, we observe that $\rho$-power-mean welfare is $\epsilon$-well-estimated for $\rho > 1$, as expected, and is also fairly well-estimated for $\rho < -1$.
Since $\rho$-power-mean welfare is $\max_{\PlayerIndex \in \SetOfPlayers} \bm{w}_{\PlayerIndex}^{\nicefrac1\rho}$-Lipschitz continuous, we expect high error as $\rho \to 0^{-}$, and as it is Lipschitz discontinuous for $\rho \in [0, 1)$, we also expect high error in this region.


Although power means are not Lipschitz-continuous, for $\rho \in [0, 1)$, and may consequently
be more difficult to estimate than Lipschitz-continuous properties, their sample complexity is not necessarily unbounded.
In particular, for $\rho = 0$, the difficulty with estimating the power mean stems from the fact that the derivative can grow in an unbounded fashion as some utility value approaches zero.
Said derivatives, however, are only unbounded when that utility value is exactly zero.
Thus, as we move away from this singularity, our estimates will converge, but the rate of convergence may be slow.
In contrast, other interesting properties of games, like the price of anarchy, are not even continuous---the price of anarchy depends on the maximum welfare at an equilibrium, which is not a continuous function of utilities---and thus no convergent estimator can exist for them.

\FloatBarrier

\ifpoa

\section{Estimating Extreme Welfare Equilibria}

We now turn our attention to extreme equilibria---the best or the worst, i.e., optimal or pessimal, respectively---measured in terms of the welfare achieved.
We call the welfare-maximal equilibria maximally \emph{consonant}, and the welfare-minimizing ones, maximally \emph{dissonant}.
Analogously, we call the values of these equilibria the maximal consonance and dissonance.

As it turns out, neither the maximally consonant nor maximally dissonant equilibria, nor their values, is amenable to statistical EGTA methodology, i.e., these properties are not well behaved.
Consequently, we define relaxations of these properties, that we call maximally consonant and maximally dissonant $\Lambda$-stable outcomes, which are similar in spirit to their equilibrium counterparts, but which are well behaved.
Specifically, we replace the equilibrium constraints with their Lagrangian relaxations based on a tunable parameter $\Lambda \ge 0$.
This change in the definitions does not alter the spirit of the properties, only the letter, because as $\Lambda$ goes to infinity, the penalty for violating the constraints becomes infinite, so they are not violated.
On the other hand, when $\Lambda$ is finite, these definitions can be understood as permitting small oscillations around extreme equilibria, thus expanding the scope of acceptable play.

For extreme $\Lambda$-stable outcomes, we obtain a positive result, as these properties are well behaved.
However, having overcome what might have appeared to be the shortcoming that was preventing the estimation of extreme equilibria, we find ourselves facing a more serious stumbling block.
Although we can now derive a more satisfying bound on these extreme
properties, given an $\epsilon$-uniform approximation of a game, this bound grows with $\Lambda$.
In other words, fixing the number of samples and letting $\Lambda$ increase so as to produce solutions closer and closer to an extreme equilibrium yields a larger and larger confidence interval around the property's estimate.
Alternatively, fixing $\epsilon$ and letting $\Lambda$ grow, $\epsilon$-accurate estimation of the property requires more and more samples.
This result is not entirely surprising, as we have effectively interpolated between two extreme cases: $0$-stable outcomes, which are well behaved, and $\infty$-stable outcomes, which are not.

\newcommand{\constant}{c}

\begin{restatable}[Extreme Equilibria are not Lipschitz]{observation}{obsInapproximable}
\label{obs:inapproximable}
Consider the following game family $\GameTuple (\gamma)$ parameterized by $\gamma \in (0,1)$, for any fixed $\constant \geq 0$, with utility matrix $\begin{bmatrix} (\bm{\gamma,-\gamma}) & (-\gamma,\gamma) \\
(\gamma-\constant,-\gamma) & (\bm{\constant-\gamma,\gamma}) \\
\end{bmatrix}$ 
(possible equilibria shown in bold).
For $\GameTuple (-\gamma)$ and $\GameTuple (\gamma)$ with corresponding utilities $\Utility_{-\gamma}, \Utility_{\gamma}$, it holds that
$\forall \epsilon \ge \frac{\abs{\gamma}}{2}, \norm{\Utility_{-\gamma} - \Utility_\gamma}_\infty \leq 2\gamma \leq \epsilon$, but
$\mathsmaller{\smallabs{
    \smallsup{\StratProfile \in \Nash(\Utility)}
        \Welfare(\StratProfile; \Utility) \, \ - 
    \smallsup{\StratProfile \in \Nash(\Utility_\gamma)}
        \Welfare(\Utility_\gamma, \StratProfile)
    }} =
\mathsmaller{\smallabs{
    \smallinf{\StratProfile \in \Nash(\Utility)}
        \Welfare(\StratProfile; \Utility) \, \ - 
    \smallinf{\StratProfile \in \Nash(\Utility_\gamma)}
        \Welfare(\Utility_\gamma, \StratProfile)
    }} = 
\constant$.
\end{restatable}

In fact, extreme equilibria are not even continuous! \amy{what does this mean? equilibria is a set-valued thing.}
For any fixed $\gamma \in (0,1)$, there exists a finite sample complexity that is sufficient to
distinguish between $\Utility_{-\gamma}$ and $\Utility_{\gamma}$.
However, that sample complexity tends to infinity as $\gamma$ tends to zero.
As a result, there does not exist a \emph{bound\/} on sample complexity within which we can learn to distinguish between $\Utility_{-\gamma}$ and $\Utility_{\gamma}$.
In other words, extreme welfare equilibria are inapproximable.

\begin{restatable}[Extreme Equilibria are Inapproximable]{corollary}{corInapproximable}
For all $\epsilon > 0, \delta \in (0, \mathsmaller{\frac{1}{2}})$, there does not exist a finite sample size $m < \infty$ that is sufficient to $\epsilon$-estimate maximally consonant and maximally dissonant equilibria, or their values, with probability at least $1 - \delta$.
\label{cor:inapproximable}
\end{restatable}

\newcommand{\MaxDis}[1][]{{\rm{MD}_{#1}}}
\newcommand{\MaxCon}[1][]{{\rm{MC}_{#1}}}

Next, we derive natural bidirectional bounds on extreme equilibria and extreme $\Lambda$-stable outcomes.
We write $\MaxDis(\Utility)$ and $\MaxCon(\Utility)$ to denote the set of maximally dissonant and maximally consonant equilibria of the game defined by $\Utility$.
As all elements of these sets have the same values, we overload this notation, using it to denote both the set (i.e., the witnesses) and the value of all its elements; the specific meaning should always be clear from context.

Given a game $\Utility$, the maximally dissonant equilibrium value $\MaxDis (\Utility)$ is upper bounded by the maximally consonant one $\MaxCon (\Utility)$.
To obtain a lower bound on $\MaxDis (\Utility)$,
we consider the value of the maximally dissonant $2\epsilon$-equilibria in an $\epsilon$-approximation $\Utility'$ of $\Utility$, which underestimates the maximally dissonant equilibrium value of $\Utility$.
Similarly, to obtain an upper bound on $\MaxCon (\Utility)$, we consider the value of the maximally consonant $2\epsilon$-equilibria in $\Utility'$, which overestimates the maximally consonant equilibrium value of $\Utility$.
Note that both these bounds can be arbitrarily loose, as the extreme $2\epsilon$-equilibria in $\Utility'$ may not correspond to extreme equilibria in $\Utility$ (see Observation~\ref{obs:inapproximable}).

We can tighten these bounds by making an artificial stability assumption, namely
that the welfare of \emph{every\/} $\alpha$-Nash equilibrium in $\Utility$ approximates 
(by $\alpha\gamma$, for some $\gamma > 0$)
that of an exact Nash equilibrium (also in $\Utility$), as this (unreasonable) assumption renders well behaved the welfare of extreme equilibria (see Appendix~\ref{app:poa_appendix}).
This condition is a variant of the \emph{$\epsilon$-$\Delta$ approximate stability} property introduced by Awasthi \emph{et al.\/} \cite{awasthi2010nash}, who find that games with certain structures, such as zero-sum games,
satisfy these conditions.  Interestingly, while these authors use the condition to bound the error induced by \emph{computational\/} approximation, our usage demonstrates that it can also be used to control the error due to \emph{statistical\/} approximation.

As the stability condition is difficult to justify, we instead propose an alternative quantity as a substitute measure for extreme equilibria, namely extreme $\Lambda$-stable outcomes, in which we replace the equilibrium constraints with their Lagrangian relaxation.
We now write $\MaxDis_{\Lambda} (\Utility)$ ($\MaxCon_{\Lambda} (\Utility)$) to denote the maximally dissonant (consonant) $\Lambda$-stable \emph{outcomes\/} of the game defined by $\Utility$.

\begin{definition}
\label{def:Lambda-stable}
Given a game $\GameTuple$ with corresponding utilities $\Utility$, and $\Lambda \ge 0$, a \mydef{maximally dissonant $\Lambda$-stable outcome} is an element of
$\MaxDis_{\Lambda} (\Utility) \doteq \inf_{\StratProfile \in \StratProfileSpace} \Welfare(\StratProfile; \Utility) + \Lambda \Regret(\StratProfile; \Utility)$.
A \mydef{maximally consonant $\Lambda$-stable outcome} is defined analogously: i.e.,
$\MaxCon_{\Lambda} (\Utility) \doteq \sup_{\StratProfile \in \StratProfileSpace} \Welfare(\StratProfile; \Utility) - \Lambda \Regret(\StratProfile; \Utility)$.
\end{definition}

These new properties are Lipschitz,
as we show presently, and thus are better behaved than extreme equilibria, in that they can be well estimated via a uniform approximation of the game.
We also provide a proof that
$\MaxDis_{\Lambda}^* (\Utility)$ and $\MaxCon_{\Lambda}^* (\Utility)$ are Lipschitz properties.
These quantities are analogous to their counterparts, but are defined in terms of the \mydef{excess regret} $\Regret^{*}(\StratProfile; \Utility) \doteq \Regret(\StratProfile; \Utility) - \inf_{\StratProfile' \in \StratProfileSpace} \Regret(\StratProfile'; \Utility)$,
which is relevant in case the smallest regret is not zero.
%
The Lipschitz constant of $\Regret^{*}(\StratProfile; \Utility)$ is 4, a result which can be obtained via the Lipschitz calculus (Theorem~\ref{thm:lipschitz-facts}).


\begin{restatable}[Approximating Extreme $\Lambda$-Stable Outcomes]{theorem}{thmLambdaStable}
\label{thm:LambdaStable}
Assume $\Welfare: \StratProfileSpace \times (\StratProfileSpace \to \R^{\NumberOfPlayers}) \to \R$ is 
a $\lambda_{\Welfare}$-Lipschitz property.
It holds that
$\MaxDis_{\Lambda} (\Utility)$ is $\lambda_{\Welfare} + 2\Lambda$-Lipschitz, which immediately implies that
$\abs{ \MaxDis_{\Lambda} (\Utility) - \MaxDis_{\Lambda} (\Utility') } \leq (\lambda_{\Welfare} + 2\Lambda) \epsilon$.
%
It also holds that
$\MaxDis_{\Lambda}^* (\Utility)$ is $\lambda_{\Welfare} + 4\Lambda$-Lipschitz, which immediately implies that
$\abs{ \MaxDis_{\Lambda}^* (\Utility) - \MaxDis_{\Lambda} (\Utility') } \leq (\lambda_{\Welfare} + 4\Lambda) \epsilon$.
%
Likewise, for $\MaxCon_{\Lambda} (\Utility)$ and $\MaxCon_{\Lambda}^* (\Utility)$.
\end{restatable}


Not only are extreme $\Lambda$-stable outcomes mathematically preferable to extreme equilibria because of their estimation properties, we believe they are also as, if not more, justifiable as descriptors or predictors of the play of a game.
The standard motivation for considering welfare \emph{at equilibrium\/} is that we expect ideal players to converge to equilibrium play.
On the other hand, it is reasonable to expect non-ideal players 
to play near, but not exactly at, equilibria~\cite{Fudenberg98}.
By expanding the scope of 
play to include approximate equilibria, we are able to accommodate more realistic behavior.
At the same time, the parameter $\Lambda$ allows us to control the extent to which non-equilibrium play is acceptable.
Taking $\Lambda = 0$ ignores equilibrium behavior entirely, while the other extreme approaches equilibrium behavior: letting $\Lambda \to \infty$, $\MaxDis_{\Lambda} (\Utility) \to \MaxDis (\Utility)$ (likewise, $\MaxCon_{\Lambda} (\Utility) \to \MaxCon (\Utility)$).

We can interpret $\Lambda$-stable outcomes that deviate from equilibrium play as ranging from fully cooperative to fully non-cooperative.
The welfare term is a cooperative one, as players are colluding to make the world as wonderful (or terrible) a place as possible.
The regret term, in contrast, is non-cooperative, but $\Lambda$ permits some flexibility in behavior.
In the case of a maximally dissonant (colluding) outcome,
$\Lambda = 0$ implies 
players are all playing pessimally: i.e., 
they are colluding to cause as much suffering to all players (themselves included) as possible.
The opposite is true of maximally consonant (colluding) outcome, where players are colluding to create a utopia.
On the other hand, as $\Lambda \to \infty$, the players move away from cooperative behavior towards rational
equilibrium 
behavior.
Indeed, $\Lambda$-stability allows us to model a range of cooperative to non-cooperative behaviors.

\input{poa/poa_thewellwellory}

Next, paralleling the structure of our discussion of extreme equilibria, where we introduced $\Lambda$-stable outcomes as an alternative to extreme equilibria, we now introduce an alternative to the anarchy ratio, namely the anarchy gap.

\begin{definition}[Anarchy and Stability Gaps]
The \mydef{anarchy gap} is defined as
$\AnarchyGap (\Utility) \doteq \sup_{\StratProfile \in \StratProfileSpace} \Welfare(\StratProfile) -
\inf_{\StratProfile \in \Nash(\Utility)} \Welfare(\StratProfile)$, 
and the \mydef{stability gap} is defined similarly.
Combining ideas, given $\Lambda \ge 0$,
the \mydef{$\Lambda$-anarchy gap} is defined as the \emph{Lagrangian relaxation} over the equilibrium set, i.e.,
$\AnarchyGap_{\Lambda} (\Utility) \doteq \sup_{\StratProfile \in \StratProfileSpace} \Welfare(\StratProfile) - \inf_{\StratProfile \in \StratProfileSpace} \Welfare(\StratProfile) + \Lambda \Regret(\StratProfile; \Utility)$,
and the \mydef{$\Lambda$-stability gap} is defined similarly.

If $\Utility$ has an equilibrium, then $\AnarchyGap_{\Lambda} (\Utility) \ge 0$.
However, as equilibria are not guaranteed to exist in games without mixed strategies, 
$\AnarchyGap_{\Lambda} (\Utility)$
can become very negative, as $\Lambda$ tends toward infinity.
Thus, we also define a relaxation of this gap:
given $\Lambda \ge 0$,
the \mydef{$\Lambda^{*}$-anarchy gap} is defined as
$\AnarchyGap^{*}_{\Lambda} (\Utility) \doteq \sup_{\StratProfile \in \StratProfileSpace} \Welfare(\StratProfile) - \inf_{\StratProfile \in \StratProfileSpace} \Welfare(\StratProfile) + \Lambda \Regret^* (\StratProfile; \Utility)$,
%
where, as above, $\Regret^{*} (\StratProfile; \Utility)$ is the \mydef{excess regret}.
As $\Regret^{*} (\StratProfile; \Utility) \ge 0$
and $\smallsup{x \in X} f(x) \ge \,\, \smallinf{x \in X} f(x)$,
it follows that $\AnarchyGap^{*}_{\Lambda} (\Utility) \ge 0$.
\end{definition}


These latter two measures of anarchy, parameterized by $\Lambda$, are Lipschitz properties, and thus are well behaved, so that they can be well approximated via a uniform approximation of the game.

\begin{restatable}{theorem}{thmApproxGap}[Approximating the Anarchy Gap]
Assume $\Welfare: \StratProfileSpace \times (\StratProfileSpace \to \R^{\NumberOfPlayers}) \to \R$ is 
a $\lambda_{\Welfare}$-Lipschitz property.
It holds that
$\AnarchyGap_{\Lambda} (\Utility)$ is $2 (\lambda_{\Welfare} + \Lambda)$-Lipschitz, which immediately implies
$\abs{ \AnarchyGap_{\Lambda} (\Utility) - \AnarchyGap_{\Lambda} (\Utility') } \leq 2 (\lambda_{\Welfare} + \Lambda) \epsilon$.
It also holds that
$\AnarchyGap^*_{\Lambda} (\Utility)$ is $2 (\lambda_{\Welfare} + \Lambda)$-Lipschitz, which immediately implies
$\abs{ \AnarchyGap^*_{\Lambda} (\Utility) - \AnarchyGap^*_{\Lambda} (\Utility') } \leq 4 (\lambda_{\Welfare} + \Lambda) \epsilon$.
\label{thm:anarchy-gap}
\end{restatable}

This theorem can be applied using one's preferred measure of welfare 
together with the corresponding Lipschitz constant.
When an equilibrium is known to exist, we can use the regret property.
If, however, an equilibrium is not known to exist (e.g., a pure strategy equilibrium in a normal-form game), then we resort to using the excess regret property, and suffer a factor of two loss of accuracy.

All of the aforementioned results apply not only to the anarchy but to stability as well, for the corresponding definitions.
Furthermore, analogs of these ratios and gaps can be defined for mixed games.
As stated earlier, all our bounds hold, regardless of whether the game is pure or mixed.

The $\Lambda$-anarchy gap, by design, is a measure for which we can derive a more satisfying bound than the corresponding $\Lambda$-anarchy ratio.
However, given an $\epsilon$-uniform approximation of a game, this bound, like our bound on extreme $\Lambda$-stable outcomes, grows with $\Lambda$.
Analogously, fixing the number of samples and letting $\Lambda$ increase so as to subtract a value that is closer and to that of a pessimal equilibrium, yields a larger and larger confidence interval around the gap's estimate.
Alternatively, fixing $\epsilon$ and letting $\Lambda$ grow, $\epsilon$-accurate estimation of the gap requires more and more samples.
We visualize this phenomenon in Appendix~\ref{app:poa_appendix}, Figure~\ref{fig:anarchygapresults}.
\amy{how did we generate these plots? looks like this might end up being a forward reference to the bounds, if we move this section to earlier, and leave the exp'ts in the appendix. but maybe these exp'ts should also be in the paper?}


Some have argued against Nash, or any static, equilibrium as the preferred solution concept, and in favor of alternatives that take into account the dynamics of learning agents, thereby deeming near-equilibrium behavior an acceptable outcome of game play~\cite{young1993evolution,omidshafiei2019alpha,wicks2005algorithm}.
Our results are yet another argument in this vein, with regard to optimal Nash equilibria, based not on computational complexity~\cite{Gilboa89,Conitzer08}, but rather on statistical estimation.
In particular, there does not exist a finite $\NumberOfSamples$ that is sufficient to estimate a welfare-optimizing Nash equilibrium well; on the other hand, estimating an extreme $\Lambda$-stable outcome for any finite 
$\Lambda$ is a problem we solve in this paper.

\fi

\section{Conclusion}
\label{sec:conc}

The main contributions of this paper are twofold.
First, we present a simple yet general framework in which to analyze game properties, which recovers known results on the learnability of game-theoretic equilibria.
We show that other properties of interest also fit into our framework, e.g., power-mean and Gini social welfare.
These choices show the flexibility and ease with which our methodology may be applied.
Furthermore, by analyzing these simple choices, we observe that not all game properties are well behaved.


Second, we develop a novel algorithm that learns well-behaved game properties with provable efficiency guarantees.
The key new insights that our paper offers are efficiency bounds with a dependence on utility variance, not just its range, as shown previously by \citet{tuyls2020bounds}.
While \citet{areyan2020improved} also employed variance-aware bounds in 
their regret-based pruning algorithm, they did not provide any efficiency guarantees. 
We improve upon their variance-aware bounds by deriving a novel sub-gamma tail bound, and we achieve efficiency guarantees via a novel progressive sampling schedule.
In future work, these improvements could be combined with \citeauthor{areyan2020improved}'s  regret-based pruning
to yield still tighter efficiency guarantees when the property of interest is specifically regret: i.e., equilibria.


While our algorithms are nearly optimal for uniformly approximating simulation-based games, which is \emph{sufficient} for learning well-behaved game properties, \citeauthor{areyan2020improved}'s regret-based pruning algorithm, which learns non-uniform approximations, demonstrates that doing so is not always \emph{necessary}.
In future work, we plan to develop additional pruning criteria that increase
efficiency by 
pruning any utilities that are not relevant to learning the property at hand.
We are hopeful that a generic algorithm can be developed that prunes efficiently while maintaining the generality of our theory, rather than resorting to multiple algorithms, each of which learns a specific property.

\pagebreak[3]
\FloatBarrier

\cyrus{Casing in references, pages incorrect}

{
\bibliographystyle{ACM-Reference-Format}

\renewcommand*{\bibfont}{\small}

\setlength{\bibsep}{0.95ex}
\bibliography{bibliography}
}

\pagebreak[4]

\iftrue
\appendix
\onecolumn

\renewcommand\thefigure{\thesection.\arabic{figure}}
\renewcommand\thealgorithm{\thesection.\arabic{algorithm}}

\section{Proofs of Lipschitz Properties in Games}
\label{app:lipschitz_proofs}

\subsection{Lipschitz Calculus}

\begin{theorem}[Lipschitz Facts]
\label{thm:lipschitz-facts}
%
1.~Linear Combination: If $g_{1:m}$ are $\lambda_{1:m}$-Lipschitz, all w.r.t.\ the same two norms, and $\bm{w} \in \R^{m}$, then the function $x \mapsto \bm{w} \cdot g(x)$
is $\sum_{i=1}^{m} \lambda_{i} \abs{\bm{w}_{i}}$-Lipschitz.
%
%
2.~Composition: If $h: A \to B$ and $g: B \to C$ are $\lambda_{h}$- and $\lambda_{g}$-Lispchitz w.r.t.\ norms $\norm{\cdot}_{A}$, $\norm{\cdot}_{B}$, \& $\norm{\cdot}_{C}$, then
$(g \circ h): A \to C$ is $\lambda_{h} \lambda_{g}$ Lispchitz w.r.t.\ $\norm{\cdot}_{A}$ and $\norm{\cdot}_{C}$.
%
3.~The infimum and supremum operations are 1-Lipschitz continuous: i.e.,
if for all $x \in \mathcal{X}$, $f(x; \Utility)$ is $\lambda$-Lipschitz in $\Utility$ then $\inf_{x \in \mathcal{X}} f(x; \Utility)$ is also $\lambda$-Lipschitz in $\Utility$.
Likewise, for the supremum.
\end{theorem}

\subsection{Mixed Utilities are Lipschitz Continuous}

\begin{lemma}
The functional $f: \MixedStratProfileSpace \times (\StratProfileSpace \to \R^{\NumberOfPlayers}) \to \R$ that computes a player's utility given a mixed strategy,
e.g., $\Utility^{\diamond}_{\PlayerIndex} (\StratProfile)$ for some $\PlayerIndex \in \SetOfPlayers$,
is 1-Lipschitz continuous in utilities.
\label{lem:mixed}
\end{lemma}

\begin{proof}
The functional $f$ converts a player's utility at a mixed strategy profile into the expected value of their utility over pure strategy profiles, given the mixture.
This operation is $1$-Lipschitz, because taking an expectation is a convex combination, which is 1-Lipschitz by the linear combination rule, because utilities are 1-Lipschitz themselves.
\end{proof}

\subsection{Lipschitz Continuous Game Properties}




\begin{restatable}[Lipschitz Properties]{theorem}{thmLipschitzProperties}
\label{thm:LipschitzProperties}
1. The Gini social welfare function is 1-Lipschitz:
i.e.,
\[
\sup_{\StratProfile \in \StratProfileSpace}
\abs{ \Gini_{\bm{w}^{\downarrow}} (\Strategy; \Utility) - \Gini_{\bm{w}^{\downarrow}} (\Strategy; \Utility') } \leq \norm{\Utility - \Utility'}_{\infty}
\enspace .
\]
2. Given a player $\PlayerIndex \in \SetOfPlayers$,
their adversarial value is 1-Lipschitz:
i.e.,
\[
\sup_{\Strategy \in \StrategySet_{\PlayerIndex}}
\abs{ \AdversarialValue_{\PlayerIndex} (\Strategy; \Utility) - \AdversarialValue_{\PlayerIndex} (\Strategy; \Utility') } \leq \norm{\Utility - \Utility'}_{\infty}
\enspace .
\]
3. Pure regret is 2-Lipschitz:
i.e.,
\[
\sup_{\StratProfile \in \StratProfileSpace}
\abs{ \Regret (\StratProfile; \Utility) - \Regret (\StratProfile; \Utility') } \leq 2 \norm{\Utility - \Utility'}_{\infty}
\enspace .
\]
\cyrus{Why no enumerate?}
\end{restatable}

\begin{proof}
1. Fix an arbitrary strategy profile $\StratProfile \in \StratProfileSpace$.
Recall the definition of Gini social welfare, namely $\Gini_{\bm{w}^{\downarrow}} (\StratProfile; \Utility) \doteq 
\bm{w}^{\downarrow} \cdot \Utility^{\uparrow} (\StratProfile)$.
Because $\bm{w}^{\downarrow} \in \triangle^{\NumberOfPlayers}$ is decreasing and $\Utility^{\uparrow} (\StratProfile)$ is increasing, it holds that
$\Gini_{\bm{w}^{\downarrow}} (\StratProfile; \Utility)
=
\inf_{\bm{w} \in \Pi (\bm{w}^{\downarrow})} \bm{w} \cdot \Utility^{\uparrow} (\StratProfile) =
\inf_{\bm{w} \in \Pi (\bm{w}^{\downarrow})} \bm{w} \cdot \Utility (\StratProfile)$,
where $\Pi(\bm{w}^{\downarrow})$ is the set of all possible permutations of 
$\bm{w}^{\downarrow}$.
Now, since $\norm{\bm{w}^{\downarrow}}_{1} = 1$, the dot product is a convex combination, which 
is $1$-Lipschitz.
Taking an infimum of $1$-Lipschitz functions is also $1$-Lipschitz.
Thus, by composition,  $\Gini_{\bm{w}^{\downarrow}} (\StratProfile; \Utility)$ is $1 \cdot 1 = 1$-Lipschitz.
As $\StratProfile$ was arbitrary, the result holds for all $\StratProfile \in \StratProfileSpace$.

2. Fix an arbitrary strategy $\Strategy \in \StrategySet_{\PlayerIndex}$.
By definition,
$\AdversarialValue_{\PlayerIndex} (\Strategy; \Utility) = \inf_{\StratProfile \in \StratProfileSpace \mid \StratProfile_{\PlayerIndex} = \Strategy} \Utility_{\PlayerIndex} (\StratProfile)$.
Utilities are 1-Lipschitz continuous in themselves.
Taking an infimum of $1$-Lipschitz functions is also $1$-Lipschitz.
Thus, by composition, $\AdversarialValue_{\PlayerIndex} (\Strategy; \Utility)$ is $1 \cdot 1 = 1$-Lipschitz.
As $\Strategy$ was arbitrary, the result holds for all $\Strategy \in \StrategySet_{\PlayerIndex}$.

3. Fix an arbitrary strategy profile $\StratProfile \in \StratProfileSpace$.
By definition,
\[
\Regret (\StratProfile; \Utility) = \max_{\PlayerIndex \in \SetOfPlayers} \Regret_{\PlayerIndex} (\StratProfile; \Utility) = \max_{\PlayerIndex \in \SetOfPlayers} \sup_{\StratProfile' \in \Adjacent_{\PlayerIndex, \StratProfile}} \left( \Utility_{\PlayerIndex} (\StratProfile') - \Utility_{\PlayerIndex} (\StratProfile) \right)
\enspace .
\]
Utilities are 1-Lipschitz continuous in themselves.
The difference between utilities is then 2-Lipschitz, because it can be construed as the sum of two 1-Lipschitz continuous functions, and negating a function does not change its Lipschitz constant.
Finally, as the maximum and supremum operations are both 1-Lipschitz,
by composition, $\Regret(\StratProfile; \Utility)$ is $1 \cdot 1 \cdot (1 + 1) = 2$-Lipschitz.
As $\StratProfile$ was arbitrary, the result holds for all $\StratProfile \in \StratProfileSpace$.
\end{proof}

\begin{figure*}
\begin{center}
\begin{tikzpicture}[xscale=2,yscale=1.2]

\node[rectangle,draw=black] (profiles) at (1, -1) {$\displaystyle\inf_{\StratProfile \in \StratProfileSpace \mid \StratProfile_{\PlayerIndex} = \Strategy} (\cdot) : 1 \times (1) = 1$};

\node[rectangle,draw=black] (utilities) at (1, -2) {$\displaystyle\Utility_{\PlayerIndex} (\StratProfile): 1$};

\draw[->] (profiles) -- (utilities);

\end{tikzpicture}
\quad\quad
\begin{tikzpicture}[xscale=2,yscale=1.4]

\node[rectangle,draw=black] (infimum) at (1, -1) {$\displaystyle \inf_{\bm{w} \in \pi (\bm{w}^{\downarrow})} (\cdot) : 1 \times (1) = 1$};

\node[rectangle,draw=black] (linear) at (1, -2) {$\displaystyle \bm{w} \cdot (\cdot): \sum_{i=1}^{\NumberOfPlayers} \bm{w}_{i} (1) = 1$};

\node[rectangle,draw=black] (utilities) at (1, -3) {$\displaystyle\Utility (\StratProfile): 1$};

\draw[->] (infimum) -- (linear);
\draw[->] (linear) -- (utilities);

\end{tikzpicture}
\quad\quad
\begin{tikzpicture}[xscale=2,yscale=1.1]

\node[rectangle,draw=black] (players) at (1, -1) {$\displaystyle\max_{\PlayerIndex \in \SetOfPlayers} (\cdot): 1 \times (2) = 2$};

\node[rectangle,draw=black] (profiles) at (1, -2) {$\displaystyle\sup_{\StratProfile' \in \Adjacent_{\PlayerIndex} (\StratProfile; \Utility)} (\cdot): 1 \times (2) = 2$};

\node[rectangle,draw=black] (minus) at (1, -2.9) {$\displaystyle (\cdot) - (\cdot) : (1) + (1) = 2$};

\node[rectangle,draw=black] (here) at (0.5, -3.6) {$\displaystyle\Utility_{\PlayerIndex} (\StratProfile'): 1$};

\node[rectangle,draw=black] (there) at (1.5, -3.6) {$\displaystyle\Utility_{\PlayerIndex} (\StratProfile): 1$};

\draw[->] (players) -- (profiles);

\draw[->] (profiles) -- (minus);

\draw[->] (minus) -- (here);
\draw[->] (minus) -- (there);

\end{tikzpicture}
\end{center}
\caption{\small
Abstract syntax tree derivation of the Lipschitz constant for $\AdversarialValue_{\PlayerIndex} (\Strategy; \Utility)$, left, 
$\Gini_{\bm{w}^{\downarrow}} (\StratProfile; \Utility)$, middle,
and $\Regret (\StratProfile; \Utility)$, right. The empty parentheses $(\cdot)$ indicate the arguments to the operators, while the filled parentheses (e.g., $(1)$) are filled with the respective arguments' Lipschitz constants.}
\label{fig:ast}
\end{figure*}

\subsection{Uniform Approximations of Game Properties}

We now show~\Cref{thm:lipschitz-properties}.

\thmExtremal*

\begin{proof}
To show (1),
first note the following:
\[
\sup_{x \in \mathcal{X}} f(x; \Utility) - \sup_{x \in \mathcal{X}} f(x; \Utility') \leq \sup_{x \in \mathcal{X}} (f(x; \Utility) - f(x; \Utility')) \leq \sup_{x \in \mathcal{X}} \abs{ f(x; \Utility) - f(x; \Utility')} \leq \lambda\epsilon \enspace.
\]
The final inequality follows from Observation~\ref{obs:lipschitz-uniform}.
Symmetric reasoning yields 
$\sup_{x \in \mathcal{X}} f(x; \Utility') - \sup_{x \in \mathcal{X}} f(x; \Utility) \leq \lambda \epsilon$, implying the statement for suprema.  
We obtain a similar result for infima, by noting that suprema of $-f$ are infima of $f$.

The infimum claim of (2) 
follows similarly from the supremum claim, and the supremum claim holds as follows, where the first inequality is due to Observation~\ref{obs:lipschitz-uniform}, the second, to the assumption that $\hat{x}$ is $\alpha$-optimal,
and the third, to (1):
\[
f (\hat{x}; \Utility) \geq f (\hat{x}; \Utility') - \lambda\epsilon \geq \sup_{x \in \mathcal{X}} f(x; \Utility') - \lambda\epsilon - \alpha \geq \sup_{x \in \mathcal{X}} f(x; \Utility) - 2 \lambda \epsilon - \alpha \enspace.
\]
\end{proof}

We now show \Cref{thm:lipschitz-properties-target-case}.

\thmWitness*

\begin{proof}
By Observation~\ref{obs:lipschitz-uniform},
\[
\forall x \in F_{0} (\Utility): \, f(x; \Utility) = v^* \implies f(x; \Utility') - v^* \leq \lambda \epsilon \implies x \in F_{\lambda \epsilon} (\Utility')
\]
\[
\forall x \in F_{\lambda \epsilon}(\Utility'): \, f (x; \Utility') - v^* \leq \lambda \epsilon \implies f (x; \Utility) - v^* \leq \lambda \epsilon + \lambda \epsilon = 2 \lambda \epsilon \implies x \in F_{2 \lambda \epsilon}(\Utility)
\]
\end{proof}

\section{Proofs of Finite-Sample Bounds for Expected NFGs}
\label{app:finite-sample-bounds}

We begin by showing ~\Cref{thm:Bennett}.

\thmBennett*

\begin{proof}
Bennett's Inequality (\citeyear{bennett1962probability}) states that


\begin{equation}
    \mathsmaller{
    {\Prob\!\left( \bar{X} \geq \Expect[\bar{X}] + \epsilon \right) \leq \exp\left(-\frac{\NumberOfSamples\sigma_{\GameTuple}^{\smash{2}}}{\UtilityRange_\GameTuple}h\left(\frac{\UtilityRange_\GameTuple \epsilon}{\sigma_{\GameTuple}^{\smash{2}}}\right)\right)}
    }
\label{eqn:sub-gamma}
\end{equation}
where $h(x)\doteq (1+x)\ln(1+x)-x$ for all $x\geq 0$. Substituting $-\bar{X}$ for $\bar{X}$, we get an identical reverse-tail bound, which when combined with a union bound yields
\begin{equation*}
    \mathsmaller{
    {\Prob\!\left( \abs{\bar{X} - \Expect[\bar{X}]} \geq \epsilon \right) \leq 2\exp\left(-\frac{\NumberOfSamples\sigma_{\GameTuple}^{\smash{2}}}{\UtilityRange_\GameTuple}h\left(\frac{\UtilityRange_\GameTuple \epsilon}{\sigma_{\GameTuple}^{\smash{2}}}\right)\right)}
    }
\end{equation*}
Taking the variance proxy $\sigma_{\Gamma}^2$ to be $\UtilityVariance_{\PlayerIndex}(\StratProfile)$ and the scale proxy $\UtilityRange_{\Gamma}$ to be $\UtilityRange$,
\sbhaskar{we derive Bennett's inequality, which bounds}{we obtain the following guarantee on } the deviation between $\Utility_{\PlayerIndex}(\StratProfile; \ConditionDistribution)$ and $\hat{\Utility}_{\PlayerIndex} (\StratProfile; \Samples)$ for a single $\smash{(\PlayerIndex, \StratProfile) \in \UtilityIndices}$ \sbhaskar{as}{:}
\begin{equation*}
    \mathsmaller{
    {\Prob\!\left( \abs{\hat{\Utility}_{\PlayerIndex} (\StratProfile; \Samples) - \Utility_{\PlayerIndex}(\StratProfile; \ConditionDistribution)} \geq \epsilon \right) \leq 2\exp\left(-\frac{\NumberOfSamples\UtilityVariance_{\PlayerIndex}(\StratProfile)}{\UtilityRange}h\left(\frac{\UtilityRange \epsilon}{\UtilityVariance_{\PlayerIndex}(\StratProfile)}\right)\right)}
    }
\end{equation*}
Thus, to arrive at an $\epsilon$-radius confidence interval in which
$\left\lvert \Utility_{\PlayerIndex}(\StratProfile; \ConditionDistribution) - \hat{\Utility}_{\PlayerIndex}(\StratProfile; \Samples) \right\rvert \le \epsilon$
with probability at least $1 -\delta$, \sbhaskar{we apply Bennett's inequality twice, each time using $\nicefrac{\delta}{2}$, to bound the upper and lower utility tails, which can be achieved with}{we upper bound the right hand side by $\delta$. This constraint is achieved with}

\sbhaskar{
\begin{equation}
\bm{\NumberOfSamples}_{\PlayerIndex}(\StratProfile) = \ln \frac{2}{\delta} \left(\frac{\UtilityRange}{3\epsilon} + \frac{\UtilityVariance_{\PlayerIndex}(\StratProfile)}{\epsilon^{2}} + \sqrt{\frac{2 \UtilityRange \UtilityVariance_{\PlayerIndex}(\StratProfile)}{3\epsilon^{3}} + \left( \frac{\UtilityVariance_{\PlayerIndex}(\StratProfile)^2}{\epsilon^2} \right)^{2}}\right) \leq 2\ln \frac{2}{\delta} \left( \frac{\UtilityRange}{3\varepsilon} + \frac{\UtilityVariance_{\PlayerIndex}(\StratProfile)}{\varepsilon^{2}}\right)
\end{equation}
}
{
\begin{equation}
    m^B(\epsilon, \delta; \UtilityRange, \UtilityVariance_{\PlayerIndex}(\StratProfile)) \doteq \frac{\UtilityRange^2\ln\left(\frac{2}{\delta}\right)}{\UtilityVariance_{\PlayerIndex}(\StratProfile) h\left(\frac{\UtilityRange\epsilon}{\UtilityVariance_{\PlayerIndex}(\StratProfile)}\right)} \leq 2\ln \frac{2}{\delta} \left( \frac{\UtilityRange}{3\varepsilon} + \frac{\UtilityVariance_{\PlayerIndex}(\StratProfile)}{\varepsilon^{2}}\right)
\label{eqn:mps}
\end{equation}
}
samples. \sbhaskar{}{Notice that $m^B$ is monotonically increasing and concave in $\UtilityVariance_{\PlayerIndex}(\StratProfile)$ (this can be confirmed by analyzing its first and second derivatives).} \bhaskar{Should we prove? Proofs are algebraically intensive, but conceptually direct (Second derivative is negative).}

To obtain a similar $\epsilon$-$\delta$ guarantee over strategy profile $\StratProfile$, namely $\norm{\Utility(\StratProfile; \ConditionDistribution) - \hat{\Utility}(\StratProfile; \Samples)}_{\infty} \le \epsilon$, it suffices (applying a union bound over $\NumberOfPlayers$ players) to draw at least

\sbhaskar{
\begin{equation}
\bm{\NumberOfSamples}(\StratProfile) = \max_{\PlayerIndex \in \SetOfPlayers} 2\ln \frac{2\NumberOfPlayers}{\delta} \left( \frac{\UtilityRange}{3\varepsilon} + \frac{\UtilityVariance_{\PlayerIndex}(\StratProfile)}{\varepsilon^{2}}\right) = 2\ln \frac{2\NumberOfPlayers}{\delta} \left( \frac{\UtilityRange}{3\varepsilon} + \frac{\norm{\UtilityVariance(\StratProfile)}_{\infty}}{\varepsilon^{2}}\right)
\end{equation}
}
{
\begin{equation}
    \max_{\PlayerIndex\in\SetOfPlayers} m^B(\epsilon, \nicefrac{\delta}{\NumberOfPlayers}; \UtilityRange, \UtilityVariance_{\PlayerIndex}(\StratProfile)) = m^B(\epsilon, \nicefrac{\delta}{\NumberOfPlayers}; \UtilityRange; \norm{\UtilityVariance(\StratProfile)}_\infty) \leq 2\ln \frac{2\NumberOfPlayers}{\delta} \left( \frac{\UtilityRange}{3\varepsilon} + \frac{\norm{\UtilityVariance(\StratProfile)}_{\infty}}{\varepsilon^{2}}\right)
\label{eqn:ms}
\end{equation}
}

samples.
Here, we must simulate a strategy profile until the utility indices of \emph{all\/} players are well-estimated.
Thus, the bottleneck is the worst-case (per-player) variance, which explains the use of the infinity norm over players.


Now note that the total \sbhaskar{sample}{data} complexity of $\epsilon$-$\delta$ estimating a game \sbhaskar{may then be bounded by}{is} the maximum \sbhaskar{sample}{data} complexity of $\epsilon$-$\delta$ estimating any strategy profile\sbhaskar{, which again by a union bound gives a probability $1-\delta$ guarantee:}{. The total data complexity is then}

\begin{align}
m 
&\leq \max_{\StratProfile \in \StratProfileSpace} 2\ln \frac{2\abs{\GameTuple}}{\delta} \left( \frac{\UtilityRange}{3\varepsilon} + \frac{\norm{\UtilityVariance (\StratProfile)}_{\infty}}{\varepsilon^{2}} \right) \\
&= 2\ln \frac{2\abs{\GameTuple}}{\delta}\left( \frac{\UtilityRange}{3\varepsilon} + \max_{\StratProfile \in \StratProfileSpace} \frac{\norm{\UtilityVariance(\StratProfile)}_{\infty}}{\varepsilon^{2}} \right) \\
&= 2\ln \frac{2\abs{\GameTuple}}{\delta}\left( \frac{\UtilityRange}{3\varepsilon} + \frac{\norm{\UtilityVariance}_{\infty}}{\varepsilon^{2}} \right)
\end{align}

\bhaskar{Replace above with}
\sbhaskar{}
{
\begin{equation}
    \max_{\StratProfile\in\StratProfileSpace} m^B(\epsilon, \nicefrac{\delta}{\abs{\GameTuple}}; \UtilityRange, \norm{\UtilityVariance(\StratProfile)}_\infty)) = m^B(\epsilon, \nicefrac{\delta}{\abs{\GameTuple}}; \UtilityRange; \norm{\UtilityVariance}_\infty) \leq 2\ln \frac{2\abs{\GameTuple}}{\delta} \left( \frac{\UtilityRange}{3\varepsilon} + \frac{\norm{\UtilityVariance}_{\infty}}{\varepsilon^{2}}\right)\enspace.
\end{equation}
}

\noindent
Similarly, to arrive at the \sbhaskar{simulation}{query} complexity, we simply sum rather than maximize over strategy profiles:
\begin{align}
m
&\leq \sum_{\StratProfile \in \StratProfileSpace} \max_{\PlayerIndex \in \SetOfPlayers} 2\ln \frac{2\abs{\GameTuple}}{\delta} \left( \frac{\UtilityRange}{3\varepsilon} + \frac{\UtilityVariance_{\PlayerIndex}(\StratProfile)}{\varepsilon^{2}} \right) \\
&= 2\ln \frac{2\abs{\GameTuple}}{\delta}\left( \frac{\UtilityRange \abs{\GameTuple}}{3\varepsilon} + \sum_{\StratProfile \in \StratProfileSpace} \max_{\PlayerIndex \in \SetOfPlayers} \frac{\UtilityVariance_{\PlayerIndex}(\StratProfile)}{\varepsilon^{2}} \right) \\
&= 2\ln \frac{2\abs{\GameTuple}}{\delta}\left( \frac{\UtilityRange \abs{\GameTuple}}{3\varepsilon} + \frac{\norm{\UtilityVariance}_{1,\infty}}{\varepsilon^{2}} \right) \enspace .
\end{align}
\bhaskar{Replace above with}
\sbhaskar{}{
\begin{align*}
    & & &\sum_{\StratProfile \in \StratProfileSpace}\left\lceil m^B(\epsilon, \nicefrac{\delta}{\abs{\GameTuple}}; \UtilityRange, \norm{\UtilityVariance(\StratProfile)}_\infty)\right\rceil\\
    &\le& &\abs{\StratProfileSpace} + \sum_{\StratProfile \in \StratProfileSpace}m^B(\epsilon, \nicefrac{\delta}{\abs{\GameTuple}}; \UtilityRange, \norm{\UtilityVariance(\StratProfile)}_\infty)\\
    &\le& &\abs{\StratProfileSpace} + \abs{\StratProfileSpace} m^B\left(\varepsilon, \nicefrac{\delta}{\abs{\GameTuple}}; \UtilityRange, \frac{1}{\abs{\StratProfileSpace}}\sum_{\StratProfile\in \StratProfileSpace}\norm{\UtilityVariance(\StratProfile)}_\infty\right)& &[{\textup{Jensen's Inequality + Concavity of $m^B$}}]\\
    &=& &\abs{\StratProfileSpace} + m^B(\varepsilon, \nicefrac{\delta}{\abs{\GameTuple}};c\abs{\StratProfileSpace}, \norm{\UtilityVariance}_{1, \infty})& &[{\textup{Definition of $m^B$ + Algebra}}]\\
    &\leq& &\abs{\StratProfileSpace} + 2\ln \frac{2\abs{\GameTuple}}{\delta}\left( \frac{\UtilityRange \abs{\StratProfileSpace}}{3\varepsilon} + \frac{\norm{\UtilityVariance}_{1,\infty}}{\varepsilon^{2}} \right)\enspace.
\end{align*}
}
\end{proof}

In order to show~\Cref{thm:eBennett}, we prove a more general theorem. ~\Cref{thm:eBennett} then follows through the application of a union bound over all utilities. We begin with the following Lemma:

\begin{lemma}
\label{lemma:variance_bound}
Let $X \doteq (X_1, \dots, X_m)$ be a vector of independent random variables each with finite variance, assume that $|X_i - \mathbb{E}[X_i]|\leq c$ almost surely for all $i$, and define $\Bar{X}\stackrel{.}{=} \frac{1}{m}\sum_{i=1}^m X_i$. Then $\hat{v}(X) \stackrel{.}{=} \frac{1}{m-1}\sum_{i=1}^m(X_i - \Bar{X})^2$ is $\left(1, \frac{c^2m}{4(m-1)^2}\right)$-self-bounding, and hence by \cite{boucheron2000sharp}, with probability at least $1 - \delta$ over the draw of $X$, it holds that
\begin{equation*}
    \Expect[\hat{v}(X)] \leq \Hat{v}(X) +  \frac{2c^2\ln\frac{1}{\delta}}{3m} + \sqrt{\left(\frac{1}{3} + \frac{1}{2\ln\frac{1}{\delta}}\right)\left(\frac{c^2\ln\frac{1}{\delta}}{m-1}\right)^2 + \frac{2c^2\Hat{v}(X)\ln\frac{1}{\delta}}{m}}\enspace.
\end{equation*}
Furthermore, we also have that $\hat{v}(X)$ is $\left(\frac{m}{m-1}, 0\right)$-self-bounding, and hence by \cite{boucheron2009concentration}, with probability at least $1 - \delta$ over the draw of $X$, it holds that
\begin{equation*}
    \hat{v}(X) \leq \Expect[\hat{v}(X)] + \frac{2m+1}{6(m-1)}\cdot\frac{c^2\ln\frac{3}{\delta}}{m} + \sqrt{\frac{2c^2v\ln\frac{3}{\delta}}{m-1}}\enspace.
\end{equation*}
\begin{proof}
Let $\textsf{Z}(X)\stackrel{.}{=}\hat{v}(X)$, and for each index $j$, define $\textsf{Z}_j(X)$ by
\begin{equation*}
    \textsf{Z}_j(X)\stackrel{.}{=}\frac{1}{m - 1}\sum_{i=1, i\neq j}^m(X_i - \Bar{X}_{\setminus j})^2
\end{equation*}
where $\Bar{X}_{\setminus j}$ denotes $\sum_{i=1,i\neq j} X_i$. We will first show that $\textsf{Z}_j(X) = \textsf{Z}(X) - \frac{1}{m}(X_j - \Bar{X}_{\setminus j})^2$, as this form comes in handy many times. Starting from the definition of $\textsf{Z}_j(X)$, we add and subtract $\frac{1}{m-1}(X_j - X_{\setminus j})^2$ to the entire expression, and then add and subtract $\Bar{X}$ to the argument of the sum, to obtain:
\begin{align*}
    \textsf{Z}_j(X)&= \frac{1}{m - 1}\left(\sum_{i=1}^m(X_i - \Bar{X}_{\setminus j})^2 - (X_j-\Bar{X}_{\setminus j})^2\right)\\
    &=\frac{1}{m - 1}\Bigg(\sum_{i=1}^m((X_i - \Bar{X}) + (\Bar{X} - \Bar{X}_{\setminus j}))^2 - (X_j-\Bar{X}_{\setminus j})^2\Bigg)\\
    &=\frac{1}{m-1}\Bigg(\sum_{i=1}^m(X_i-\Bar{X})^2 +2\sum_{i=1}^m(X_i-\Bar{X})(\Bar{X}-\Bar{X}_{\setminus j}) +\sum_{i=1}^m(\Bar{X}-\Bar{X}_{\setminus j})^2 -(X_j-\Bar{X}_{\setminus j})^2\Bigg).
\end{align*}
Since $\Bar{X} - \Bar{X}_{\setminus j}$ is independent of $i$ and $\sum_{i=1}^m(X_i - \Bar{X})=0$, we have that $2\sum_{i=1}^m(X_i-\Bar{X})(\Bar{X}-\Bar{X}_{\setminus j}) = 0$. This, in combination with the definition of $\textsf{Z}(X)$, gives us
\begin{align*}
    \textsf{Z}_j(X)&=\textsf{Z}(X) + \frac{1}{m - 1}\left[m(\Bar{X}-\Bar{X}_{\setminus j})^2 - (X_j-\Bar{X}_{\setminus j})^2\right].
\end{align*}
But we have that $\Bar{X} = \frac{1}{m}X_j + \frac{m - 1}{m}\Bar{X}_{\setminus j}$, which implies that
\begin{align}
    \textsf{Z}_j(X)\nonumber &= \textsf{Z}(X) + \frac{1}{m-1}\left[m\left(\frac{1}{m}X_j - \frac{1}{m}\Bar{X}_{\setminus j}\right)^2 - (X_j-\Bar{X}_{\setminus j})^2\right]\nonumber\\
    &=\textsf{Z}(X) + \frac{1}{m-1}\left[-\frac{m - 1}{m}\left(X_j - \Bar{X}_{\setminus j}\right)^2\right]\nonumber\\
    &=\textsf{Z}(X) - \frac{1}{m}\left(X_j - \Bar{X}_{\setminus j}\right)^2\label{ag_zj-equation}.
\end{align}
We will now show that $\textsf{Z}$ is a $\left(1, \frac{c^2m}{4(m-1)^2}\right)$-self-bounding function with scale $c^2/m$. Since Equation \ref{ag_zj-equation} implies that $\textsf{Z}(X) = \textsf{Z}_j(X) + \frac{1}{m}\left(X_j - \Bar{X}_{\setminus j}\right)^2$ for all indices $j$, and the definition of $X$ guarantees that $X_j-\Bar{X}_{\setminus j}\leq c$ for all indices $j$, we have that for all indices $j$, it must be the case that
\begin{equation*}
    \textsf{Z}_j(X)\leq \textsf{Z(X)}\leq \textsf{Z}_j(X)+\frac{c^2}{m}.
\end{equation*}
Finally, we will now show that
\begin{equation*}
    \sum_{i=1}^m(\textsf{Z}(X)-\textsf{Z}_j(X))\leq Z(X) + \frac{c^2m}{4(m-1)^2}.
\end{equation*}
By applying Equation \ref{ag_zj-equation}, we have
\begin{align*}
    \sum_{j=1}^m(\textsf{Z}(X)-\textsf{Z}_j(X)) &= \sum_{j=1}^m\left(\textsf{Z}(X)-\textsf{Z(X)}+\frac{1}{m}(X_j-\Bar{X}_{\setminus j})^2\right)\\
    &= \frac{1}{m}\sum_{j=1}^m(X_j-\Bar{X}_{\setminus j})^2
\end{align*}
But by the definition of $X$, it holds that $\Bar{X}_{\setminus j} = \frac{1}{m-1}\left(m\Bar{X} - X_j\right)$. Hence, we have that
\begin{align*}
\sum_{j=1}^m(\textsf{Z}(X)-\textsf{Z}_j(X)) &= \frac{1}{m}\sum_{j=1}^m\left(\frac{m}{m-1}X_j-\frac{m}{m-1}\Bar{X}\right)^2\\
&=\frac{m}{m-1}\frac{1}{m-1}\sum_{j=1}^m\left(X_j-\Bar{X}\right)^2\\
&=\frac{m}{m-1}\textsf{Z}(X)\\
&\leq \textsf{Z}(X) + \frac{c^2m}{4(m-1)^2}.
\end{align*}\enspace.
Therefore, we have shown that $\textsf{Z}(X) = \hat{v}(X)$ is both $\left(1, \frac{c^2m}{4(m-1)^2}\right)$-self-bounding and $\left(\frac{m}{m-1}, 0\right)$-self-bounding with scale $\frac{c^2}{m}$ in both cases.
\end{proof}
\end{lemma}

\begin{theorem}[General Bennett-Type Empirical-Variance Sensitive Finite-Sample Bounds]
\label{thm:generalEBennett}
Let $X_1, \dots, X_m$ be independent random variables with finite variance, and assume that $|X_i - \mathbb{E}[X_i]|\leq c$ almost surely for all $i$. Define $\Bar{X}\stackrel{.}{=}\frac{1}{m}\sum_{i=1}^mX_i$ and take
\begin{align*}
    \hat{v} &\stackrel{.}{=} \frac{1}{m-1}\sum_{j=1}^m (X_i - \Bar{X})^2\\
    \epsilon_{\hat{v}} &\stackrel{.}{=} \frac{2c^2\ln\frac{3}{\delta}}{3m} + \sqrt{\left(\frac{1}{3} + \frac{1}{2\ln\frac{3}{\delta}}\right)\left(\frac{c^2\ln\frac{3}{\delta}}{m-1}\right)^2 + \frac{2c^2\hat{v}\ln\frac{3}{\delta}}{m}}\\
    \epsilon_\mu &\stackrel{.}{=} \frac{c\ln\frac{3}{\delta}}{3m} + \sqrt{\frac{2(\hat{v} + \epsilon_{\hat{v}})\ln\frac{3}{\delta}}{m}}\enspace.
\end{align*}
Then with probability at least $1-\delta$ over the draw of $X$, we have that $|\Bar{X} - \Expect\Bar{X}|\leq \epsilon_\mu$. Furthermore, if $\delta \leq 0.03$, then for any draw of $X$, it holds that
\begin{equation*}
    \epsilon_\mu \leq  \frac{2c\ln\frac{3}{\delta}}{m-1} + \sqrt{\frac{2\hat{v}\ln\frac{3}{\delta}}{m}}\enspace.
\end{equation*}
\end{theorem}
\begin{proof}
By Bennett's Inequality, with probability at least $1 - \delta / 3$ over the draw of $X$, it holds that $\Bar{X} - \Expect \Bar{X} \leq \frac{c\ln\frac{3}{\delta}}{3m} + \sqrt{\frac{2\Expect[\hat{v}]\ln\frac{3}{\delta}}{m}}$. For the same reason, we also have that with probability at least $1 - \delta / 3$ over the draw of $X$, it holds that $\Expect\Bar{X} - \Bar{X} \leq \frac{c\ln\frac{3}{\delta}}{3m} + \sqrt{\frac{2\Expect[\hat{v}]\ln\frac{3}{\delta}}{m}}$. Finally, by \Cref{lemma:variance_bound}, with probability at least $1 - \delta / 3$ over the draw of $X$, it holds that $\Expect[\hat{v}] \leq \hat{v} + \epsilon_{\hat{v}}$. Applying a union bound, we get the above guarantee.

We will now derive the stated upper bound on $\epsilon_{\mu}$. First note that
\begin{align*}
    \hat{v} + \epsilon_{\hat{v}} &= \hat{v} + \frac{2c^2\ln\frac{3}{\delta}}{3m} + \sqrt{\left(\frac{1}{3} + \frac{1}{2\ln\frac{3}{\delta}}\right)\left(\frac{c^2\ln\frac{3}{\delta}}{m-1}\right)^2 + \frac{2c^2\hat{v}\ln\frac{3}{\delta}}{m}}\\
    &< \hat{v} + \left(\frac{2}{3} + \frac{1}{4\ln\frac{3}{\delta}}\right)\frac{c^2\ln\frac{3}{\delta}}{m-1} + \sqrt{\left(\frac{1}{3} + \frac{1}{2\ln\frac{3}{\delta}}\right)\left(\frac{c^2\ln\frac{3}{\delta}}{m-1}\right)^2 + \frac{2c^2\hat{v}\ln\frac{3}{\delta}}{m}} \\
    &= \left(\sqrt{\frac{c^2\ln\frac{3}{\delta}}{2(m-1)}} + \sqrt{\left(\frac{1}{6} + \frac{1}{4\ln\frac{3}{\delta}}\right)\frac{c^2\ln\frac{3}{\delta}}{m-1} + \hat{v}}\right)^2\enspace.
\end{align*}
We then have that
\begin{align*}
    \epsilon_\mu &= \frac{c\ln\frac{3}{\delta}}{3m} + \sqrt{\frac{2(\hat{v} + \epsilon_{\hat{v}})\ln\frac{3}{\delta}}{m}}\\
    &< \frac{4c\ln\frac{3}{\delta}}{3(m-1)} + \sqrt{\left(\frac{1}{3} + \frac{1}{2\ln\frac{3}{\delta}}\right)\left(\frac{c\ln\frac{3}{\delta}}{m-1}\right)^2 + \frac{2\hat{v}\ln\frac{3}{\delta}}{m}}\\
    &\leq \left(\frac{4}{3} + \sqrt{\frac{1}{3} + \frac{1}{2\ln\frac{3}{\delta}}}\right)\frac{c\ln\frac{3}{\delta}}{m-1} + \sqrt{\frac{2\hat{v}\ln\frac{3}{\delta}}{m}}\enspace.
\end{align*}
When $\delta \leq 0.03$, we get the stated upper bound.
\end{proof}

In order to show ~\Cref{cor:eBennettComplexity}, we show a more general corollary of ~\Cref{thm:generalEBennett}. Again, ~\Cref{cor:eBennettComplexity} follows after a union bound over utilities.

\begin{corollary}
\label{cor:generalEBennettComplexity}
Take the assumptions of ~\Cref{thm:generalEBennett}, and define $v \doteq \Expect[\hat{v}] =  \frac{1}{m}\sum_{i=1}^m\Expect[(X_i - \Bar{X})^2]$. For all $\epsilon > 0$, if
\begin{equation*}
    m\geq 1 + \frac{5c\ln\frac{3}{\delta}}{\epsilon} + \frac{2v\ln\frac{3}{\delta}}{\epsilon^2},
\end{equation*}
then with probability at least $1 - \frac{\delta}{3}$ over the draw of $X$, we have that $\epsilon_\mu \leq \epsilon$. 
\begin{proof}
From ~\Cref{thm:generalEBennett}, we have that
\begin{equation*}
    \epsilon_\mu < \frac{4c\ln\frac{3}{\delta}}{3(m-1)} + \sqrt{\left(\frac{1}{3} + \frac{1}{2\ln\frac{3}{\delta}}\right)\left(\frac{c\ln\frac{3}{\delta}}{m-1}\right)^2 + \frac{2\hat{v}\ln\frac{3}{\delta}}{m}}\enspace.
\end{equation*}
By \Cref{lemma:variance_bound}, we have with probability at least $1 - \frac{\delta}{3}$ over the draw of $X$ that 
\begin{align*}
    \epsilon_\mu &< \frac{4c\ln\frac{3}{\delta}}{3(m-1)} + \sqrt{\left(\frac{1}{3} + \frac{1}{2\ln\frac{3}{\delta}}\right)\left(\frac{c\ln\frac{3}{\delta}}{m-1}\right)^2 + \frac{2\left(v + \frac{2m+1}{6(m-1)}\cdot\frac{c^2\ln\frac{3}{\delta}}{m} + \sqrt{\frac{2c^2v\ln\frac{3}{\delta}}{m-1}}\right)\ln\frac{3}{\delta}}{m}}\\
    &< \frac{4c\ln\frac{3}{\delta}}{3(m-1)} + \sqrt{\left(\frac{1}{3} + \frac{(4m+2)(m-1)}{6m^2} + \frac{1}{2\ln\frac{3}{\delta}}\right)\left(\frac{c\ln\frac{3}{\delta}}{m-1}\right)^2 + 2\cdot \frac{c\ln\frac{3}{\delta}}{m-1}\sqrt{\frac{2v\ln\frac{3}{\delta}}{m}} + \frac{2v\ln\frac{3}{\delta}}{m}}\\
    &= \frac{4c\ln\frac{3}{\delta}}{3(m-1)} + \sqrt{\left(1 - \frac{m+1}{3m^2} + \frac{1}{2\ln\frac{3}{\delta}}\right)\left(\frac{c\ln\frac{3}{\delta}}{m-1}\right)^2 + 2\cdot \frac{c\ln\frac{3}{\delta}}{m-1}\sqrt{\frac{2v\ln\frac{3}{\delta}}{m}} + \frac{2v\ln\frac{3}{\delta}}{m}}\\
    & < \frac{4c\ln\frac{3}{\delta}}{3(m-1)} + \sqrt{\left(1 + \frac{1}{2\ln\frac{3}{\delta}}\right)\left(\frac{c\ln\frac{3}{\delta}}{m-1}\right)^2 + 2\sqrt{1 + \frac{1}{2\ln\frac{3}{\delta}}}\cdot \frac{c\ln\frac{3}{\delta}}{m-1}\sqrt{\frac{2v\ln\frac{3}{\delta}}{m}} + \frac{2v\ln\frac{3}{\delta}}{m}}\\
    &= \left(\frac{4}{3} + \sqrt{1 + \frac{1}{2\ln\frac{3}{\delta}}}\right)\frac{c\ln\frac{3}{\delta}}{(m-1)} + \sqrt{\frac{2v\ln\frac{3}{\delta}}{m}}\\
    &< \left(\frac{4}{3} + \sqrt{1 + \frac{1}{2\ln\frac{3}{\delta}}}\right)\frac{c\ln\frac{3}{\delta}}{(m-1)} + \sqrt{\frac{2v\ln\frac{3}{\delta}}{m-1}}\enspace.
\end{align*}
Setting the right hand side less than or equal to $\epsilon$ and then solving for $m$, we get 
\begin{equation*}
    m\geq 1 + \left(\frac{4}{3} + \sqrt{1 + \frac{1}{2\ln\frac{3}{\delta}}}\right)\frac{c\ln\frac{3}{\delta}}{\epsilon} + \frac{v\ln\frac{3}{\delta}}{\epsilon^2} + \sqrt{2\left(\frac{4}{3} + \sqrt{1 + \frac{1}{2\ln\frac{3}{\delta}}}\right)\frac{c\ln\frac{3}{\delta}}{\epsilon}\cdot \frac{v\ln\frac{3}{\delta}}{\epsilon^2} + \left(\frac{v\ln\frac{3}{\delta}}{\epsilon^2}\right)^2}\enspace.
\end{equation*}
Since $\sqrt{2ab + b^2}\leq a + b$, a looser lower bound on $m$ is
\begin{equation*}
    m\geq 1 + \left(\frac{8}{3} + \sqrt{4 + \frac{2}{\ln\frac{3}{\delta}}}\right)\frac{c\ln\frac{3}{\delta}}{\epsilon} + \frac{2v\ln\frac{3}{\delta}}{\epsilon^2}\enspace.
\end{equation*}
Noting that $\frac{8}{3} + \sqrt{4 + \frac{2}{\ln\frac{3}{\delta}}} < 5$ for $\delta \leq 0.75$ (a trivial assumption given that $1 - \delta$ is supposed to be a ``high'' probability), we are done.
\end{proof}
\end{corollary}

\section{Derived Games}

\paragraph{Correlated Equilibrium}
A \mydef{correlated equilibrium}, where players can correlate their behavior based on a private signal drawn from a joint distribution, is likewise defined in terms of regret: if no player regrets their strategy, given that all other players follow their equilibrium strategies, then the joint distribution is a correlated equilibrium.

One way to understand the correlated equilibria of a (base) game is by constructing an expanded game, where players condition their behavior on strategic recommendations, given some joint distribution over strategy profiles.
When obeying recommendations is an $\epsilon$-Nash equilibrium in this expanded game, the joint distribution is defined as an $\epsilon$-correlated equilibrium.
As constructing such an expanded game involves taking convex combinations of the utilities in the (base) game, using the weights specified by the joint distribution, a 1-Lipschitz operation, an $\epsilon$-uniform approximation of the (base) game is likewise an $\epsilon$-uniform approximation of the expanded game.
\Cref{cor:dual_containment} then gives us dual containment of Nash equilibria in the
expanded and approximate games.
As this argument applies to all joint distributions over strategy profiles simultaneously, \amy{why is simultaneity important?} whenever the obedient strategy profile is a Nash equilibrium of an expanded game, it is a $2\epsilon$-Nash equilibrium of a corresponding approximate expanded game, and further, any $2\epsilon$-Nash equilibrium of an approximate expanded game is a $4\epsilon$-Nash equilibrium of the expanded game.
Therefore, whenever a joint distribution is a correlated equilibrium of a 
game, it is a $2\epsilon$-correlated equilibrium of a corresponding approximate game, and further, any $2\epsilon$-correlated equilibrium of an approximate game is a $4\epsilon$-correlated equilibrium of the
game.

\paragraph{Team Games}
We define a \mydef{team game} as one derived from a base game $\GameTuple$, \amy{how exactly does the base game matter here? what if the utilities depended on all players? would you just throw away that information?} and a set of teams, which is a partition of the players in the base game into groups.
The strategy set for each group is the Cartesian product of the strategy sets of its members, and the utility function for each group is a welfare function of the utility functions of the group's members: e.g., utilitarian when team score is the sum of per-member scores, and egalitarian when a team is judged by its weakest member.
Given a uniform approximation of a normal-form game, a uniform approximation of the derived team game is preserved to within a factor of $\lambda$, so long as the welfare function is $\lambda$-Lipschitz.

\section{Sampling Schedule}
\label{sec:sampling_schedule}

Suppose a geometric schedule \emph{base} $\beta > 1$, and
let $\ScheduleLength > 1$ represent the iteration number.
It is mathematically convenient to let $\ScheduleLength$ be any real number $\geq 2$, so we take $\floor{\ScheduleLength}$ to be the schedule length.

To derive $\alpha$, we assume the best case, namely $0$ variance, in Bennett's inequality, which yields the lower-bound on the minimum sufficient sample size
$\alpha \doteq \frac{2 \UtilityRange}{3 \varepsilon} \ln \frac{3 \ScheduleLength \abs{\GameTuple}}{\delta}$.
\sbhaskar{As this lower-bound is loose, we instead take $\alpha \doteq \frac{7 \UtilityRange}{3 \varepsilon} \ln \frac{3 \ScheduleLength \abs{\GameTuple}}{\delta}$, as per the discussion following \Cref{thm:eBennett}.}{}\bhaskar{We will need to derive the correct tighter lower bound later.}
Likewise, to derive $\omega$, we assume the worst case, namely $\nicefrac{\UtilityRange^2}{4}$ variance in Bennett's inequality, which in fact invokes Hoeffding's inequality,
so the
necessary sample size
$\omega \doteq \frac{\UtilityRange^{2}}{2\varepsilon^{2}} \ln \frac{3 \ScheduleLength \abs{\GameTuple}}{\delta}$.


Together, these derivations yield the \emph{sample size ratio} \sbhaskar{$\frac{\omega}{\alpha} = \frac{3 \UtilityRange}{14 \varepsilon}$}{$\frac{\omega}{\alpha} = \frac{3 \UtilityRange}{4 \varepsilon}$}.
Thus, for a $\beta$-geometric schedule, we require
$\ScheduleLength + 1 \geq \ceil{\log_{\beta}\frac{\omega}{\alpha}} = \ceil{\log_{\beta}\frac{3 \UtilityRange}{4 \varepsilon}}$. 
%
This reasoning yields
$\ScheduleLength = \log_{\beta}\frac{3\UtilityRange}{4\varepsilon}$,
where the schedule ranges from iterations $1, \dots, \floor{\ScheduleLength}$. 
%
Finally, we take as our geometric sampling schedule $\bm{m}_{\TimeIndex} \doteq \ceil{\alpha\beta^{\TimeIndex}}$, for all $\TimeIndex \in 1, \dots, \floor{\ScheduleLength}$.


\section{Algorithm and Proofs of its Correctness and Efficiency}
\label{app:algo_corr_eff}

\subsection{Correctness Proof}



We now show \Cref{thm:correctness}.

\thmCorrectness*

\begin{proof}
At each time step $t \in [T]$ and
for each utility profile $\Utility_{\PlayerIndex}(\StratProfile)$,
we perform three tests in which we check that our empirical estimates satisfy their respective high probability bounds.
%
%
On Line~\ref{step:varianceUpperBound}, we upper bound the true variance $\UtilityVariance_{\PlayerIndex}(\StratProfile)$ by the empirical variance $\tilde{\UtilityVariance}_{\PlayerIndex}(\StratProfile)$.
On Line~\ref{step:empiricalBennett}, we bound the deviation between the empirical mean $\hat{\Utility}_{\PlayerIndex}(\StratProfile; \Samples)$ and the true mean $\hat{\Utility}_{\PlayerIndex}(\StratProfile; \ConditionDistribution)$.
By a union bound, 
$\UtilityVariance_{\PlayerIndex}(\StratProfile) \le \tilde{\UtilityVariance}_{\PlayerIndex}(\StratProfile)$ and
$\abs{\hat{\Utility}_{\PlayerIndex}(\StratProfile; \Samples) - \hat{\Utility}_{\PlayerIndex}(\StratProfile; \ConditionDistribution)} \le \bm{\epsilon}_{\PlayerIndex}(\StratProfile; \Samples)$,
for all $(\PlayerIndex, \StratProfile) \in \SetOfPlayers \times \StratProfileSpace$
simultaneously, with probability $1 - \nicefrac{\delta}{T}$.
By another union bound, the same holds at all time steps $t \in [T]$ simultaneously with probability $1-\delta$.

Thus, whenever an index is pruned on Line~\ref{alg:psp:pruning_well-estimated indices}, the corresponding utility is well estimated: i.e., it is sufficiently close to its true value w.h.p..
Therefore, if Algorithm~\ref{alg:psp} terminates, it does so having correctly estimated all utilities: i.e., $\tilde{\Utility}$ is an $\epsilon$-estimate of $\Utility$ w.h.p..

It remains to prove termination.
By construction of the schedule, when $t=T$, it holds that $\bm{\epsilon}_{\PlayerIndex}(\StratProfile) \le  c\sqrt{\frac{\ln({\frac{3\abs{\GameTuple} \ScheduleLength}{\delta}})}{2\NumberOfSamples}} \le \epsilon$, $(\PlayerIndex, \StratProfile) \in \SetOfPlayers \times \StratProfileSpace$.
Therefore, at time $T$, anything left to prune is necessarily pruned on Line~\ref{alg:psp:pruning_well-estimated indices}, so that $\UtilityIndices = \emptyset$, and the algorithm terminates.
\end{proof}

\subsection{Efficiency Proof}

We now show \Cref{thm:efficiency}.

\thmEfficiency*

\begin{proof}
\bhaskar{I've commented out old proof}
Fix a utility index $(\PlayerIndex, \StratProfile) \in \UtilityIndices$.
By \Cref{cor:generalEBennettComplexity}, we have that with probability at least $1 - \frac{\delta}{3T\abs{\GameTuple}}$, Algorithm~\ref{alg:psp} \Cref{step:empiricalBennett} produces a (high probability) $\epsilon$-bound for $(\PlayerIndex, \StratProfile)$ when

\begin{equation}
m\leq \bm{\NumberOfSamples}_{\PlayerIndex}^{*}(\StratProfile) \doteq 1 + 2\ln \frac{3T\abs{\GameTuple}}{\delta} \left(\frac{5\UtilityRange}{2\epsilon} + \frac{ \UtilityVariance_{\PlayerIndex}(\StratProfile)}{\epsilon^{2}} \right) \enspace.
\end{equation}

Although $\bm{\NumberOfSamples}_{\PlayerIndex}^{*}(\StratProfile)$ may not exist in the sampling schedule, there always exists some $\hat{\bm{\NumberOfSamples}}_{\PlayerIndex}(\StratProfile)$ in the schedule between $\bm{\NumberOfSamples}_{\PlayerIndex}^{*}(\StratProfile)$ and $\beta\bm{\NumberOfSamples}_{\PlayerIndex}^{*}(\StratProfile)$. Moreover,
\begin{equation}
\hat{\bm{\NumberOfSamples}}_{\PlayerIndex}(\StratProfile) \leq \beta\bm{\NumberOfSamples}_{\PlayerIndex}^{*}(\StratProfile)
\leq 1 + 2\beta\ln\frac{3T\abs{\GameTuple}}{\delta} \left(\frac{5\UtilityRange}{2\epsilon} + \frac{ \UtilityVariance_{\PlayerIndex}(\StratProfile)}{\epsilon^{2}}\right)
\in \mathcal{O}\left( \log\left(\frac{\abs{\GameTuple}\log(\nicefrac{\UtilityRange}{\epsilon})}{\delta}\right)\left(\frac{\UtilityRange}{\epsilon} + \frac{\UtilityVariance_{\PlayerIndex}(\StratProfile)}{\epsilon^{2}} \right)\right) \enspace.
\end{equation}

Having calculated the number of samples $\hat{\bm{\NumberOfSamples}}_{\PlayerIndex}(\StratProfile)$ which suffices to $\epsilon$-uniformly estimate $\Utility_{\PlayerIndex}(\StratProfile)$ at a single utility index $(\PlayerIndex, \StratProfile)$ with probability at least $1 - \nicefrac{\delta}{T\abs{\GameTuple}} = 1 - \nicefrac{\delta}{T\abs{\StratProfileSpace}\abs{\SetOfPlayers}}$,
we next apply a union bound over $\SetOfPlayers$
so that
$\hat{\bm{\NumberOfSamples}}(\StratProfile) = \max_{\PlayerIndex \in \SetOfPlayers}
\hat{\bm{\NumberOfSamples}}_{\PlayerIndex}(\StratProfile)$
simulations of profile $\StratProfile$ suffice to $\epsilon$-uniformly estimate $\Utility(\StratProfile)$ at strategy profile $\StratProfile$ with probability at least $1 - \nicefrac{\delta}{T\abs{\StratProfileSpace}}$.
Thus, we guarantee \Cref{alg:psp} prunes $\StratProfile$ w.h.p.\ before exceeding $\hat{\bm{\NumberOfSamples}}(\StratProfile)$ simulations of $\StratProfile$.

Continuing to follow the logic in the proof of \Cref{thm:Bennett} beyond \eqref{eqn:ms} yields the desiderata.
In brief, we apply another union bound, this time over strategy profiles; thus with probability at least $1 - \nicefrac{\delta}{T}$, the same holds simultaneously for $\Utility$: i.e., for all utility indices $(\PlayerIndex, \StratProfile) \in \UtilityIndices$.
We then bound the sample (maximum evaluations at any profile) and simulation (total evaluations over all profiles) complexities of $\epsilon$-$\delta$ uniformly estimating $\GameTuple$.
\end{proof}

\ifpoa{
\section{Proofs about Extreme Welfare Equilibria}
\label{app:poa_appendix}

\subsection{Inapproximability Result}

We now show Observation~\ref{obs:inapproximable}.

\obsInapproximable*

\begin{proof}
By construction, $\GameTuple (-\gamma)$ is a $2\gamma \leq \epsilon$-uniform approximation of $\GameTuple (\gamma)$.
When $\gamma < 0$, the column player plays $A$, to which the row player responds with $A$.
The only, and thus both the
best and worst
equilibrium of $\GameTuple$ is $(A,A)$, with utilitarian welfare $0$.
For $\gamma > 0$, by similar reasoning the only equilibrium is $(B,B)$, which has utilitarian welfare $c$.
The utilitarian welfare of the best (and worst) welfare equilibria in $\GameTuple(-\gamma)$ and $\GameTuple(\gamma)$, therefore, differ by $\constant > 0$: i.e., arbitrarily.
\end{proof}

We now show Corollary~\ref{cor:inapproximable}.

\corInapproximable*

\paragraph{Proof Sketch}
The game family $\GameTuple (\gamma)$ can be used to show that there does not exist a finite sample size $m < \infty$ that is sufficient to $\epsilon$-estimate maximally consonant and maximally dissonant equilibria, or their values.
A proof of this fact equates learning these properties with a special case of mean-estimation, namely the problem of determining whether a coin's bias for heads is $0$ or $\abs{\gamma}$ for some $\gamma \approx 0$.
Consider the following simulation-based game, based on the game in Observation~\ref{obs:inapproximable}: 
at any profile $\StratProfile$, the simulator returns $\Utility_{0} (\StratProfile) + (B,-B)$, where $B$ is a Bernoulli coin with probability $\abs{\gamma}$ of heads, multiplied by $\sgn(\gamma)$.
More specifically, it returns $\Utility_{0} (\StratProfile) + (\sgn(\gamma),-\sgn(\gamma))$ when a head is observed, and $\Utility_{0} (\StratProfile)$ otherwise.
Determining whether $(A,A)$ or $(B,B)$ is the (unique) equilibrium of this game, i.e., whether $\gamma$ is positive or negative, is 
is as difficult as flipping a single heads. 
As $\gamma \to 0$, the number of coin flips required to see even a single heads with probability at least $1 - \delta$, and hence the number of samples required to distinguish these equilibria with fixed positive probability, approaches infinity.
As $\gamma$ is simply a game parameter, we thus conclude that \emph{no\/} finite sample size is sufficient to uniformly \amy{uniform across $\gamma$'s; a sample size that works for all $\gamma$. $\gamma$ is the problem instance. result should hold for all problem instances.} estimate the extreme equilibria of games in this family, i.e., extreme equilibria are inapproximable.

\subsection{Approximation Proofs: $\MaxDis$ and $\MaxCon$}

As usual, we fix two compatible games
characterized by utility functions $\Utility$ and $\Utility'$, 
and we assume that $\Utility'$ is an $\epsilon$-uniform approximation of $\Utility$.



\begin{restatable}{theorem}{thmExtreme}
[Approximating Extreme Equilibria]
\label{thm:extremeNash}
Fix $\epsilon > 0$.
Assume $\Welfare: \StratProfileSpace \times (\StratProfileSpace \to \R^{\NumberOfPlayers}) \to \R$
is monotonic non-decreasing so that $\Utility \preceq \bm{v}$ component-wise implies $\Welfare(\StratProfile; \Utility) \le \Welfare(\StratProfile; \bm{v})$, for all $\StratProfile \in  \StratProfileSpace$ and utilities $\Utility$ and $\bm{v}$.
It holds that
%
%
$\inf_{\StratProfile \in \Nash_{2\epsilon}(\Utility')} \Welfare (\StratProfile; \Utility' - \epsilon) \leq \MaxDis (\Utility) \leq \MaxCon (\Utility) \leq
\sup_{\StratProfile \in \Nash_{2\epsilon}(\Utility')} \Welfare (\StratProfile; \Utility' + \epsilon)$.
\end{restatable}

Our proof relies on Thm. 2.2 of~\cite{areyan2020improved}, namely $\Nash(\Utility) \subseteq \Nash_{2\epsilon} (\Utility') \subseteq \Nash_{4\epsilon}(\Utility)$,
where $\Nash$ denotes the set of equilibria of the game $\Utility$.

\begin{proof}
\begin{align*}
& \inf_{\StratProfile \in \Nash_{2\epsilon} (\Utility')} \Welfare (\StratProfile; \Utility' - \epsilon) \\
& \leq \inf_{\StratProfile \in \Nash (\Utility)} \Welfare (\StratProfile; \Utility' - \epsilon) & \Nash (\Utility) \subseteq \Nash_{2\epsilon} (\Utility') \\
& \leq \inf_{\StratProfile \in \Nash (\Utility)} \Welfare (\StratProfile; \Utility) \ = \ \MaxDis(\Utility) & \textsc{monotonicity} \\
& \leq \sup_{\StratProfile \in \Nash (\Utility)} \Welfare (\StratProfile; \Utility) \ = \ \MaxCon(\Utility)  & \inf_{x \in X} f(x) \leq \sup_{x \in X} f(x) \\
& \leq \sup_{\StratProfile \in \Nash_{2\epsilon}(\Utility')} \Welfare (\StratProfile; \Utility) & \Nash (\Utility) \subseteq \Nash_{2\epsilon}(\Utility') \\
& \leq \sup_{\StratProfile \in \Nash_{2\epsilon}(\Utility')} \Welfare (\StratProfile; \Utility' + \epsilon) & \textsc{monotonicity}
\end{align*}
\end{proof}



\begin{restatable}{corollary}{corExtreme}
Suppose as in Theorem~\ref{thm:extremeNash}.
Further, suppose $\exists \gamma > 0$
s.t.\ $\forall \alpha > 0, \forall \StratProfile \in \Nash_{\alpha}(\Utility), \exists s' \in \Nash (\Utility)$
s.t.\ $\abs{\Welfare (\StratProfile; \Utility) - \Welfare (\StratProfile'; \Utility)} \leq \alpha \gamma$.
We may then refine the upper (lower) bound on $\MaxDis(\Utility)$ ($\MaxCon(\Utility)$) as
$\MaxDis(\Utility) \leq \inf_{\StratProfile \in \Nash (\Utility')} \Welfare (\StratProfile; \Utility' + \epsilon) + 2 \gamma \epsilon$
and
$\sup_{\StratProfile \in \Nash (\Utility')} \Welfare (\StratProfile; \Utility' - \epsilon) - 2 \gamma \epsilon \leq \MaxCon(\Utility)$,
respectively.
\label{cor:extremeNash}
\end{restatable}

\begin{proof}
$\MaxDis(\Utility) =
\inf_{\StratProfile \in \Nash (\Utility)} \Welfare(\StratProfile; \Utility) 
    \leq 
\inf_{\StratProfile \in \Nash_{2\epsilon} (\Utility)} \Welfare(\StratProfile; \Utility) + 2 \epsilon \gamma
    \leq 
\inf_{\StratProfile \in \Nash (\Utility')} \Welfare(\StratProfile; \Utility) + 2 \epsilon \gamma
    \leq 
\inf_{\StratProfile \in \Nash (\Utility')} \Welfare(\StratProfile; \Utility' + \epsilon) + 2 \epsilon \gamma$,
where the first inequality is by assumption, with $\alpha = 2\epsilon$; 
the second follows because 
$\Nash (\Utility') \subseteq \Nash_{2\epsilon} (\Utility)$; and
the last, by monotonicity.
Analogously,
$\sup_{\StratProfile \in \Nash (\Utility')} \Welfare(\StratProfile; \Utility' - \epsilon) - 2 \epsilon \gamma
    \leq 
\sup_{\StratProfile \in \Nash (\Utility')} \Welfare(\StratProfile; \Utility) - 2 \epsilon \gamma
    \leq 
\sup_{\StratProfile \in \Nash_{2\epsilon} (\Utility)} \Welfare(\StratProfile; \Utility) - 2 \epsilon \gamma
    \leq 
\sup_{\StratProfile \in \Nash (\Utility)} \Welfare(\StratProfile; \Utility) =
\MaxCon(\Utility)$,
where the first inequality follows by monotonicity;
the second follows because 
$\Nash (\Utility') \subseteq \Nash_{2\epsilon} (\Utility)$; and
the third inequality is by assumption, with $\alpha = 2\epsilon$.
\end{proof}

We now show Theorem~\ref{thm:LambdaStable}.

\thmLambdaStable*

\begin{proof}
We prove this theorem using the Lipschitz calculus (Thm.~\ref{thm:lipschitz-facts}).
Recall that regret is 2-Lipschitz.
Moreover $\inf$ provides a 1-Lipschitz multiplicative factor.
Finally, addition requires that we add the Lipschitz constants of the corresponding addends.
We begin at the leaves in Figure~\ref{fig:lipschitz}, computing Lipschitz constants, and back those values up the syntax tree to arrive at $2(\lambda_{\Welfare} + \Lambda)$ Lipschitz constant for $\MaxDis_{\Lambda} (\Utility)$.
The inequality then follows via Observation~\ref{obs:lipschitz-uniform}.

The proofs in the remaining three cases---$\MaxCon_{\Lambda} (\Utility)$, 
$\MaxDis_{\Lambda}^* (\Utility)$, and
$\MaxCon{\Lambda}^* (\Utility)$---are analogous.
\end{proof}

\begin{figure}
\begin{center}
\begin{tikzpicture}[xscale=2,yscale=0.9]


\node[rectangle,draw=black] (inf) at (1, -1) {$\displaystyle\inf_{\StratProfile \in \StratProfileSpace}: 1 \times (\lambda_{\Welfare} + 2\Lambda)$};


\node[rectangle,draw=black] (plus) at (1, -2) {$\displaystyle+: (\lambda_{\Welfare}) + (2\Lambda) $};

\node[rectangle,draw=black] (w2) at (0.5, -3) {$\displaystyle\Welfare (\StratProfile): \lambda_{\Welfare}$};

\node[rectangle,draw=black] (times) at (1.5, -3) {$\displaystyle\times: (\Lambda) \times (2)$};

\node[rectangle,draw=black] (Lam) at (1, -4) {$\Lambda: \Lambda$};

\node[rectangle,draw=black] (reg) at (2, -4) {$\Regret (\StratProfile; \Utility): 2$};



\draw[->] (inf) -- (plus);

\draw[->] (plus) -- (w2);
\draw[->] (plus) -- (times);

\draw[->] (times) -- (Lam);
\draw[->] (times) -- (reg);

\end{tikzpicture}
\end{center}
\caption{Abstract syntax tree depicting the derivation of the Lipschitz constant for $\MaxDis_{\Lambda} (\Utility)$.}
\label{fig:lipschitz}
\end{figure}

\subsection{Approximation Proofs: $\AnarchyRatio$, $\StabilityRatio$, and $\AnarchyGap$}

\begin{restatable}{theorem}{thmApproxPrices}
[Approximating the Prices of Anarchy and Stability]
Assume as in Theorem~\ref{thm:extremeNash}.
Then, so long as the numerator and denominator of each bound is positive, 
%
$\max \left\{ 1, \frac{\sup_{\StratProfile \in \StratProfileSpace} \Welfare (\Utility' - \epsilon, \StratProfile)}{\sup_{\StratProfile \in \Nash_{2\epsilon}(\Utility')} \Welfare (\Utility' + \epsilon, \StratProfile)} \right\} 
\leq \StabilityRatio (\Utility)
\leq \AnarchyRatio (\Utility) 
\leq  \frac{\sup_{\StratProfile \in \StratProfileSpace} \Welfare (\Utility' + \epsilon, \StratProfile)}{\inf_{\StratProfile \in \Nash_{2\epsilon}(\Utility')} \Welfare (\Utility' - \epsilon, \StratProfile)}$.
\label{thm:approxPrices}
\end{restatable}



\begin{proof}
It suffices to show that
$\sup_{\StratProfile \in \StratProfileSpace} \Welfare(\Utility' - \epsilon, \StratProfile) \leq \sup_{\StratProfile \in \StratProfileSpace} \Welfare(\Utility, \StratProfile) \leq \sup_{\StratProfile \in \StratProfileSpace} \Welfare(\Utility' + \epsilon, \StratProfile)$,
but this follows immediately from the monotonicity of $\Welfare$, and
$\inf_{\StratProfile \in \Nash_{2\epsilon}(\Utility')} \Welfare (\Utility' - \epsilon, \StratProfile) \leq
\inf_{\StratProfile \in \Nash (\Utility)} \Welfare (\Utility, \StratProfile) =
\MaxDis (\Utility) \leq
\MaxCon (\Utility) =
\sup_{\StratProfile \in \Nash (\Utility)} \Welfare (\Utility, \StratProfile) \leq
\sup_{\StratProfile \in \Nash_{2\epsilon} (\Utility')} \Welfare (\Utility' + \epsilon, \StratProfile)$,
which follows from Thm.~\ref{thm:extremeNash}.
\end{proof}

\begin{restatable}{corollary}{corApproxPrices}
Assume as in~\Cref{cor:extremeNash}.
%
%
We may then refine the lower bound, improving the supremum (i.e., maximally consonant) equilibrium value to an infimum (i.e., maximally dissonant), 
again so long as 
the denominator is positive, 
as
$\frac{\sup_{\StratProfile \in \StratProfileSpace} \Welfare (\Utility' - \epsilon, \StratProfile)}{\inf_{\StratProfile \in \Nash (\Utility')} \Welfare (\Utility' + \epsilon, \StratProfile) + 2 \gamma \epsilon} \leq \AnarchyRatio (\Utility)$.
Analogously,
%
$\StabilityRatio (\Utility) \leq \frac{\sup_{\StratProfile \in \StratProfileSpace} \Welfare (\Utility' + \epsilon, \StratProfile)}{\sup_{\StratProfile \in \Nash (\Utility')} \Welfare (\Utility' - \epsilon, \StratProfile) - 2 \gamma \epsilon}$.
\label{cor:approxPrices}
\end{restatable}



\begin{proof}
The claim follows by lower bounding the numerator of $\AnarchyRatio(\Utility)$ (monotonicity), and upper bounding its denominator via \Cref{cor:extremeNash}, and vice versa for $\StabilityRatio(\Utility)$.
\end{proof}

We now show Theorem~\ref{thm:anarchy-gap}.

\thmApproxGap*

\begin{proof}
The proof of this theorem parallels that of Theorem~\ref{thm:LambdaStable}.
\end{proof}

\subsection{Anarchy Gap Experiments}

To generate Figure~\ref{fig:anarchygapresults}, we sampled congestion games at random using GAMUT~\cite{nudelman2004run}.
Let $\GameTuple(n, k)$ be a congestion game drawn from GAMUT with $n$ players and $k$ facilities, and $\Utility_{n, k}$ its utility function. For these experiments, $\GameTuple(n, k)$ was our ground-truth game. Our goal was to observe the empirical performance of the bound given in~\Cref{thm:anarchy-gap}.
We chose $\Welfare$ to be utilitarian welfare.

\begin{figure}[t]
\centering
\includegraphics[width=0.5\textwidth]{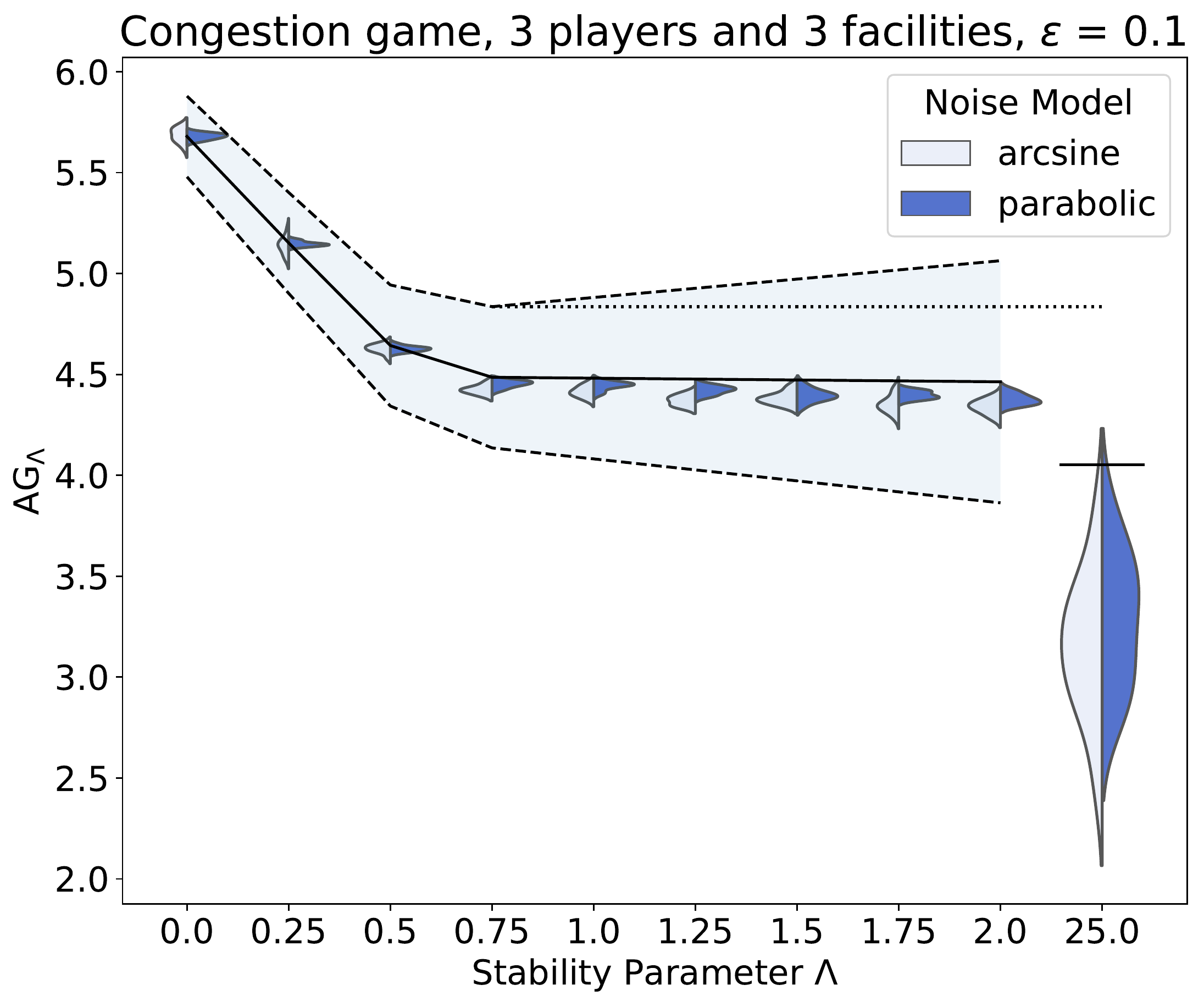}
\caption{Anarchy gap results}
\label{fig:anarchygapresults}
\end{figure}

For this purpose, we drew $\epsilon$-uniform approximations of $\GameTuple(n, k)$ by adding random noise to $\Utility_{n, k}$. Specifically, given $\Utility_{n, k}$, we constructed $\GameTuple'(n, k)$, an $\epsilon$-uniform approximation of $\GameTuple(n, k)$, by setting $\Utility'_{n, k} = \Utility_{n, k} + \bm{\epsilon}$. Here, $\bm{\epsilon}$ is a vector of randomly drawn components where each component, $\epsilon_{\PlayerIndex, \StratProfile}$, corresponds to a pair (player, strategy profile) of $\GameTuple(n, k)$. 

We experimented with two noise models: a mean-concentrated model and a biased model. 
In the mean-concentrated model, for each $(\PlayerIndex, \StratProfile)$, we draw $\epsilon_{\PlayerIndex, \StratProfile}$ from the \cyrus{Used to be: Wigner-semicircular distribution scaled to $[-\epsilon, \epsilon]$, which has (relatively small) standard deviation $\nicefrac{\varepsilon}{4}$.}
parabolic, i.e., $\beta(2, 2)$, distribution scaled to $[-\epsilon, \epsilon]$, which has (relatively small) standard deviation $\frac{\varepsilon}{4\sqrt{5}}$.
The goal of this model is to generate utilities of the estimated game that are closer to those of the ground-truth game (as one would expect via, for example, the central limit theorem) than worst-case or tail bounds would mandate.
In the biased model, for each $(\PlayerIndex, \StratProfile)$, we draw $\epsilon_{\PlayerIndex, \StratProfile}$ from the arcsine distribution, i.e., $\beta(\frac{1}{2}, \frac{1}{2})$, scaled to range $[-\epsilon, \epsilon]$, which has the relatively large standard deviation of $\nicefrac{\varepsilon}{\sqrt{2}}$. The goal of this model is to generate utilities of the estimated game that are farther away from those of the ground-truth game, as more extreme games are more likely to exhibit large changes in $\AnarchyGap_{\Lambda}(\cdot)$.

\Cref{fig:anarchygapresults} shows a representative result for $\AnarchyGap_{\Lambda}(\Utility_{3, 3})$. The figure shows several violin plots, one plot for each value of $\Lambda \in \{0, 0.25, 0.5, \ldots, 2.0\}$. Each half of a violin plots the distribution of $\AnarchyGap_{\Lambda}(\Utility'_{3, 3})$ for each of reduced bias and biased noise models. The plot also shows the upper and lower bound of~\Cref{thm:anarchy-gap}, as dashed lines, as well as a dotted line representing an improved upper-bound on $\AnarchyGap_{\Lambda}$, as this quantity decreases monotonically in $\Lambda$ (because the non-cooperative term is monotonic in $\Lambda$, and the cooperative term is constant).

Our experiments reveal that our estimator is biased.
Unsurprisingly, it appears less biased in the parabolic case than in the arcsine case, and the variance is also lower in the parabolic case.
At $\Lambda=20$, the empirical estimates are extremely noisy.
This result is consistent with our observation that extreme equilibria cannot be well estimated (\Cref{obs:inapproximable}).
Even a very small change in utilities can effect whether a strategy profile is an equilibrium or not, and can lead to a large change in extreme equilibrium values.


\cyrus{More justification of anarchy gap concept. Anarchy gap is measured in the same units as utilities and welfare.}

\cyrus{Why not PoA: if we double utils, PoA doesn't change. Not that useful as an objective in mechanism design.}


}
\fi

\end{document}